\documentclass[reqno,12pt]{amsart}

\usepackage{color} 
\usepackage{ifpdf}
\ifpdf
    \usepackage[pdftex]{graphicx}
    \usepackage[pdftex]{hyperref}
    \hypersetup{
        unicode=false,          
        pdftoolbar=true,        
        pdfmenubar=true,        
        pdffitwindow=false,     
        pdfstartview={FitH},    
        pdftitle={MCP Article},      
        pdfauthor={Michael Holst},   
        pdfsubject={Mathematics},    
        pdfcreator={Michael Holst},  
        pdfproducer={Michael Holst}, 
        pdfkeywords={PDE, analysis, mathematical physics}, 
        pdfnewwindow=true,      
        colorlinks=true,        
        linkcolor=red,          
        citecolor=blue,         
        filecolor=magenta,      
        urlcolor=cyan           
    }

\else
    \usepackage{graphicx}
    \usepackage{pstricks}
    
    \newcommand{\href}[2]{#2}
\fi

\usepackage{times}
\usepackage{amsfonts}

\usepackage{amsmath}
\usepackage{amsthm}
\usepackage{amssymb}
\usepackage{amsbsy}
\usepackage{amscd}

\usepackage{enumerate}
\usepackage{verbatim}
\usepackage{subfigure}

\newtheorem{theorem}{Theorem}[section]
\newtheorem{corollary}[theorem]{Corollary}
\newtheorem{lemma}[theorem]{Lemma}

\newtheorem{definition}[theorem]{Definition}

\newtheorem{remark}[theorem]{Remark}

\numberwithin{equation}{section}  

  \newcounter{mnote}
  \let\oldmarginpar\marginpar
    \renewcommand\marginpar[1]{\-\oldmarginpar[\raggedleft\footnotesize #1]%
    {\raggedright\footnotesize #1}}

\definecolor{myblue}{rgb}{0.2,0.2,0.7}
\definecolor{mygreen}{rgb}{0,0.6,0}
\definecolor{mycyan}{rgb}{0,0.6,0.6}
\definecolor{myred}{rgb}{0.9,0.2,0.2}
\definecolor{mymagenta}{rgb}{0.9,0.2,0.9}
\definecolor{mywhite}{rgb}{1.0,1.0,1.0}
\definecolor{myblack}{rgb}{0.0,0.0,0.0}

\def\mathbi#1{\textbf{\em #1}}
\newcommand{\beq}{\begin{equation}}
\newcommand{\eeq}{\end{equation}}
\newcommand{\beqa}{\begin{eqnarray}}
\newcommand{\eeqa}{\end{eqnarray}}
\newcommand{\IF}{I\!\!F}         

\newcommand{\IL}{I\!\!L}

\newcommand{\Ri}{{}^{\mbox{\tiny \rm 3}}\!R} 
\newcommand{\hRi}{{}^{\mbox{\tiny \rm 3}}\!\hat{R}} 

\newcommand{\Proof}{\noindent{\bf Proof.~}}

\newcommand{\leqs}{\leqslant}      
     
\newcommand{\geqs}{\geqslant}      
      
\newcommand{\Tr}{{\gamma}}

\newcommand{\tiD}{\mbox{{\tiny $D$}}}
\newcommand{\tiE}{\mbox{{\tiny $E$}}}

\newcommand{\tiI}{\mbox{{\tiny $I$}}}

\newcommand{\tiL}{\mbox{{\tiny $L$}}}

\newcommand{\tiN}{\mbox{{\tiny $N$}}}

\newcommand{\tiR}{\mbox{{\tiny $R$}}}

\newcommand{\tiIL}{\mbox{{\tiny $\IL$}}}

\newcommand{\tiwedge}{\mbox{{\tiny $\wedge$}}}
\newcommand{\tivee}{\mbox{{\tiny $\vee$}}}

\newcommand{\R}{{\mathbb R}}       

\newcommand{\cL}{{\mathcal L}}
\newcommand{\cM}{{\mathcal M}}

\newcommand{\cO}{{\mathcal O}}

\newcommand{\cY}{{\mathcal Y}}

\newcommand{\ttK}{{\tt K}}

\newcommand{\ttk}{{\tt k}}

\newcommand{\bV}{{\bf V}}
\newcommand{\bW}{{\bf W}}
\newcommand{\bX}{{\bf X}}
\newcommand{\ba}{{\bf a}}
\newcommand{\bb}{{\bf b}}

\newcommand{\bj}{{\bf j}}

\newcommand{\bw}{{\bf w}}

\def\ee{\epsilon}

\def\mathbi#1{\textbf{\em #1}}

\newcommand{\biC}{\mathbi{C}}

\newcommand{\biL}{\mathbi{L}}

\newcommand{\biW}{\mathbi{W\,}}

\newcommand{\bib}{\mathbi{b \!\!}}

\newcommand{\bif}{\mathbi{f\,}}

\newcommand{\bih}{\mathbi{h}}

\newcommand{\bij}{\mathbi{j\,}}

\newcommand{\biu}{\mathbi{u}}
\newcommand{\biv}{\mathbi{v}}
\newcommand{\biw}{\mathbi{w}}

\newcommand{\tbW}{\textbf{W\,}}

\newcommand{\tbb}{\textbf{b}}

\newcommand{\tbu}{\textbf{u}}

\newcommand{\tbw}{\textbf{w}}

\newcommand{\hh}{\hat h}

\newcommand{\hj}{\hat \jmath}
\newcommand{\hk}{\hat k}
\newcommand{\hl}{\hat l}

\newcommand{\hD}{\hat D}

\newcommand{\hrho}{\hat \rho}

\newcommand{\htau}{\hat \tau}

\newcommand{\e}{\epsilon}
\begin{document}

\title[ Non-CMC Solutions on Compact Manifolds with Boundary ]{Non-CMC Solutions of the Einstein Constraint Equations on Compact Manifolds with Apparent Horizon Boundaries  }

\author[M. Holst]{Michael Holst}
\email{mholst@math.ucsd.edu}

\author[C. Meier]{Caleb Meier}
\email{c1meier@math.ucsd.edu}

\author[G. Tsogtgerel]{G. Tsogtgerel}
\email{gantumur@math.mcgill.ca}

\address{Department of Mathematics\\
         University of California San Diego\\ 
         La Jolla CA 92093}

\thanks{MH was supported in part by 
        NSF Awards~1065972, 1217175, and 1262982.}
\thanks{CM was supported by NSF Award~1065972.}
\thanks{GT was supported by an NSERC Discovery Grant 
        and by an FQRNT Nouveaux Chercheurs Grant.}

\date{\today}


\begin{abstract}
In this article we continue our effort to do a systematic development of 
the solution theory for conformal formulations of the Einstein constraint 
equations on compact manifolds with boundary.
By building in a natural way on our recent work in
Holst and Tsogtgerel (2013), and Holst, Nagy, and Tsogtgerel (2008, 2009),
and also on the work of Maxwell (2004, 2005, 2009) and Dain (2004),
under reasonable assumptions on the data we prove existence of both 
near- and far-from-constant mean curvature solutions for a class of 
Robin boundary conditions commonly used in the literature for modeling 
black holes, with a third existence result for constant mean curvature 
(CMC) appearing as a special case.
Dain and Maxwell addressed initial data engineering for 
space-times that evolve to contain black holes, determining solutions to 
the conformal formulation on an asymptotically Euclidean manifold in the
CMC setting, with interior boundary conditions 
representing excised interior black hole regions.
Holst and Tsogtgerel compiled the interior boundary results covered by
Dain and Maxwell, and then developed general interior conditions to model 
the apparent horizon boundary conditions of Dain and Maxwell for compact 
manifolds with boundary, and subsequently proved existence of solutions 
to the Lichnerowicz equation
on compact manifolds with such boundary conditions.
This paper picks up where Holst and Tsogtgerel left off, 
addressing the general non-CMC case for compact manifolds with boundary.
As in our previous articles,
our focus here is again on low regularity data and on the interaction 
between different types of boundary conditions.
While our work here serves primarily to extend the solution theory for the 
compact with boundary case, we also develop several technical tools that 
have potential for use with the asymptotically Euclidean case.
\end{abstract}

\maketitle


\vspace*{-1.2cm}
{\tiny
\tableofcontents
}

\section{Introduction}
This article represents the second installment in a systematic development 
of the solution theory for conformal formulations of the Einstein constraint 
equations on compact manifolds with boundary.
Our development began in~\cite{HoTs10a} by leveraging the technical tools 
we had developed in~\cite{HNT07b} for both the CMC (constant mean curvature) 
and non-CMC cases in the simpler setting of closed manifolds.
The case of compact manifolds with boundary, while more complicated than
the closed case, and often viewed as simply an approximation to the more
physically realistic asymptotically Euclidean case, is itself an important 
problem in general relativity; it is particularly important in numerical 
relativity, where it arises in models of Cauchy surfaces containing 
asymptotically flat ends and/or trapped surfaces.
Moreover, various technical obstacles that arise when extending the solution 
theory for closed manifolds developed in~\cite{HNT07b,dM09} to the case of 
asymptotically Euclidean manifolds have analogues 
in the compact with boundary case.

Our results here build on the non-CMC analysis framework from~\cite{HNT07b},
and leverage a number technical tools developed in~\cite{HoTs10a} for the
Lichnerowicz equation on compact manifolds with boundary.
The framework in~\cite{HNT07b} is particularly effective for producing 
existence results for the non-CMC case without using the so-called
near-CMC assumptions primarily because it isolates any assumptions
about the strength of the nonlinear coupling between the two equations
to the global barrier construction.
The overall Schauder-type fixed-point argument in~\cite{HNT07b,dM09}
is based entirely on topological properties of the 
fixed-point map generated by the constraints, and on the properties of
the spaces on which the map operates.
Nearly all of the required properties can be established without 
resorting to any type of near-CMC condition that restricts the strength 
of the nonlinear coupling between the two constraints.
This allows one to focus entirely on the problem of constructing global 
barriers free of the near-CMC condition
(cf.~\cite{HNT07a,dM09} for the first such constructions).
It is useful to
note that the ``near-CMC'' assumption allows for the mean curvature to be 
non-constant, but requires that the mean curvature not vanish and be bounded 
by some multiple of its gradient, while the ``far-from-CMC'' assumption simply 
means that the mean curvature is free of the near-CMC hypothesis.

We began a systematic study of the case of compact manifolds with boundary 
immediately after the work on the closed case in~\cite{HNT07b},
which developed into~\cite{HoTs10a}.
The complexity of treating the boundary conditions carefully in~\cite{HoTs10a}
led us to focus that work on the Lichnerowicz equation alone, restricting 
the analysis of the boundary difficulties to that equation in isolation from
the momentum constraint, as it is the main source of nonlinearity in the 
coupled constraint system.
In~\cite{SD04} and~\cite{DM05a}, Dain and Maxwell had addressed
initial data engineering for Einstein's equations that would evolve into 
space-times containing black holes.
They determined solutions to the conformal formulation
on asymptotically Euclidean manifolds with interior boundary 
conditions in the CMC case.  
The interior boundary results from~\cite{SD04,DM05a}
were compiled in~\cite{HoTs10a}, and then general interior conditions 
were developed in order to model the apparent horizon boundary conditions 
of Dain and Maxwell for compact manifolds with boundary.
In~\cite{HoTs10a} we then proved existence of solutions to the 
Lichnerowicz equation on compact manifolds with the aforementioned 
boundary conditions.

The first difficulty encountered in completing the program in~\cite{HoTs10a}
was that even basic results such as Yamabe classification of nonsmooth 
metrics on compact manifolds with boundary were unavailable; only the smooth 
case had been previously examined (by Escobar~\cite{Esco92,Esco96a}).
In order to develop a theory that mirrors that of the closed case,
Yamabe classification was first generalized in~\cite{HoTs10a} to nonsmooth 
metrics on compact manifolds with boundary.
In particular, it was shown that two conformally equivalent rough metrics 
could not have scalar curvatures with distinct signs.
(In the case of closed manifolds, Yamabe classification of rough metrics
was also unavailable, and had to be established in~\cite{HNT07b}.)
Other results were also extended to compact manifolds with 
boundary, such as conformal invariance of the Hamiltonian constraint.
The analysis framework from~\cite{HNT07b} was then used in~\cite{HoTs10a}
to establish several existence results for a large class of problems 
covering a broad parameter regime, which included most of the cases 
relevant in practice.
As in the work on the non-CMC case for closed manifolds in~\cite{HNT07b},
the focus in~\cite{HoTs10a} (and in this article) is on low regularity data 
and on the interaction between different types of boundary conditions, 
which had not been carefully analyzed before.

We note that the Lichnerowicz equation was considered in
isolation in~\cite{HoTs10a}, so that the CMC case with marginally
trapped surface boundary conditions (cf.~\S\ref{subsec:boundary})
was not examined.
However, the results for the Lichnerowicz equation allowed for variable 
coefficients in the critical nonlinear term; this allows the results for 
the Lichnerowicz equation from~\cite{HoTs10a} to be used to build the 
non-CMC results in this paper without modification.
We should point out that the primary reason that the Lichnerowicz equation
alone was considered in~\cite{HoTs10a} instead of the CMC case,
is that unlike the setting of closed manifolds, the constraint equations
\emph{do not decouple} in the CMC case;
this is due to the boundary conditions remaining coupled when using
models of asymptotically Euclidean manifolds with apparent horizon boundaries.
Therefore, the treatment of the CMC case for manifolds with boundary
requires essentially the same fully coupled topological fixed-point 
argument as the non-CMC case.
In fact, in this article our result for the CMC case is simply a special
case of our near-CMC result, and is not stated separately.
An unfortunate impact is that the techniques used for closed manifolds
in~\cite{HNT07b,Choq04,dM06} to lower the regularity of the metric in the 
CMC case a full degree below the best known rough metric result in the 
non-CMC case (appearing in~\cite{HNT07b}) cannot be exploited for the
CMC case on compact manifolds with boundary 
(see also Remark~\ref{rem3:8oct13}).
We note that Dilts~\cite{Dilt13a} recently has independently obtained 
solutions to a similar boundary value problem to ours in this article 
by also using the framework and supporting tools from~\cite{HoTs10a,HNT07b}
in the special case of smooth ($W^{2,p}$) metrics;
however he does not account for the coupling that occurs on the boundary, 
and therefore does not obtain solutions with apparent horizon boundaries.
He also exploits the Green's function technique developed by Maxwell 
in~\cite{dM09} for smooth metrics in $W^{2,p}$ that avoids constructing 
a subsolution.
We have avoided using the Green's function approach here in order to develop
existence results for the roughest possible class of metrics 
($h_{ab} \in W^{s,p}$, $p \in (1,\infty)$, $s > 1+\frac3p$),
which will necessitate the explicit construction of subsolutions.

In this article, we push our program further by considering the conformal
formulation in the non-CMC case on a compact manifold with boundary,
subject to the class of boundary conditions considered for the 
Lichnerowicz equation alone in~\cite{HoTs10a}.
Under reasonable assumptions on the data, we establish 
existence of both near- and far-from-CMC solutions to the conformal 
formulation on compact manifolds with interior (Robin-type) boundary 
conditions similar to those in \cite{HoTs10a,SD04,DM05a}, and exterior 
boundary conditions that are consistent with asymptotically Euclidean decay.
As mentioned above, a third result, for the CMC case, now comes simply
as a special case of the near-CMC result.
While near-CMC and far-from-CMC existence results have been obtained
for closed manifolds~\cite{jI95,jIvM96,ACI08,HNT07b,dM09},
and near-CMC existence results have been obtained for asymptotically 
Euclidean manifolds~\cite{yCBjIjY00}, 
results for compact manifolds with boundary have required the CMC assumption.
We begin the discussion right where \cite{HoTs10a} left off, 
combining the technical tools and results from our prior work 
in~\cite{HNT07b} and~\cite{HoTs10a}.
By focusing on compact models of an asymptotically Euclidean manifold with 
truncated ends and excised black hole regions, in addition to extending 
the compact with boundary existence theory, we believe that the techniques 
developed in this article will prove useful for obtaining far-from-CMC 
solutions to the conformal formulation in the asymptotically Euclidean case.

{\bf Outline of the paper.}
In \S\ref{sec:constraints}, we present some preliminary material on notation, the conformal method, and boundary conditions, briefly summarizing the more extensive presentations of these topics in~\cite{HNT07b,HoTs10a}.
In particular, in \S\ref{subsec:notation} we give a brief overview of the notation we employ for spaces, norms, and related objects; in \S\ref{subsec:conformal} we summarize the conformal method; and in \S\ref{subsec:boundary} we give an overview of the boundary conditions of primary interest, following closely the presentation from~\cite{HoTs10a}.
In \S\ref{sec:main} we give an overview of the main results, summarized as two separate theorems for the near-CMC and far-from-CMC cases, analogous to two of the three main results for the closed case developed in~\cite{HNT07b}.
In \S\ref{sec:momentum}, we develop the necessary supporting results and estimate for the momentum constraint, and in \S\ref{sec:hamiltonian} we similarly develop a number of supporting results needed for treating the Hamiltonian constraint in the overall fixed-point argument.
In \S\ref{sec:barriers}, we subsequently give several distinct global barrier constructions for the near-CMC and far-from-CMC cases.
Finally, in \S\ref{sec:proof} we give the proofs of the two main theorems.
Although many of the technical tools we need have been established
in~\cite{HNT07b,HoTs10a}, some additional required results are included
in Appendix~\ref{sec:app}.

\section{Preliminary material}
   \label{sec:constraints}

The results in this article leverage and then build on the analysis 
framework and the supporting technical tools developed in our two
previous articles~\cite{HNT07b,HoTs10a}, including the material
contained in the appendices of both works.
We have made an effort to use completely consistent notation with 
these two prior works, and have also endeavored to avoid as much a 
possible any replication of the technical tools.
However, in an effort to make the paper as self-contained as possible,
we will give a brief summary below of the (quite standard) notation we use
throughout all three articles for Sobolev classes, norms, and other objects.

\subsection{Notation and conventions}
   \label{subsec:notation}
As in~\cite{HNT07b,HoTs10a}, the function spaces employed throughout the 
article are fractional Sobolev classes; an overview of the construction 
of fractional order Sobolev spaces of sections of vector and tensor bundles 
can be found in the Appendix of~\cite{HNT07b}, based on Besov spaces and 
partitions of unity.
The case of the sections of the trivial bundle of scalars can 
also be found in~\cite{Hebey96}, and the case of tensors 
can also be found in~\cite{Palais65}.
Throughout the article we will use standard notation for such function spaces;
cf.~the introduction to~\cite{HNT07b} for an extensive a summary.
In particular, our notation for $L^p$ and Sobolev spaces and norms of 
sections of vector bundles over compact manifolds is quite standard,
which we briefly summarize below.

Let $\cM$ be an $n$-dimensional smooth, compact manifold with
non-empty boundary $\partial\cM$.
Let $\nabla_a$ be
the Levi-Civita connection associated with the metric $h_{ab} \in C^{\infty}(T_2^0\cM)$, that is, the unique torsion-free connection satisfying
$\nabla_ah_{bc}=0$. 
Here, $X(T_s^r\cM)$ denotes a particular smoothness class of sections of the $(r,s)$-tensor bundle associated with the tangent bundle $T\cM$ of $\cM$.
Let $R_{abc}{}^d$ be the Riemann tensor of the
connection $\nabla_a$, where the sign convention used in this
article is $(\nabla_a\nabla_b -\nabla_b\nabla_a)v_c = R_{abc}{}^d
v_d$. Denote by $R_{ab} := R_{acb}{}^c$ the Ricci tensor and by $R
:=R_{ab}h^{ab}$ the Ricci scalar curvature of this connection.
Integration on $\cM$ can be defined with the volume form
associated with the metric $h_{ab}$, allowing for the construction
of $L^p$-type norms and spaces.
Given an arbitrary tensor $u^{a_1\cdots
a_r}{}_{b_1\cdots b_s}$ of type $m=r+s$, we define a
real-valued function measuring its magnitude at any point $x \in \cM$ as
\begin{equation}
\label{tensor-magnitude}
|u| := (u^{a_1\cdots b_s}u_{a_1\cdots b_s})^{1/2}.
\end{equation}
The $L^p$-norm of an arbitrary tensor field 
$u^{a_1\cdots a_r}{}_{b_1\cdots b_s}$ on $\cM$ can then be
defined for any ${1\leqs p < \infty}$ and for ${p=\infty}$ respectively
using~(\ref{tensor-magnitude}) as follows,
\begin{equation}
\label{N-Lp-norm}
\|u\|_p := \left(\int_{\cM} |u|^p\,dx\right)^{1/p},\qquad
\|u\|_{\infty}:= \mbox{ess}\,\sup_{x\in\cM}|u|.
\end{equation}
The Lebesgue spaces $L^p(T^{r}_{s}\cM)$ of
sections of the $(r,s)$-tensor bundle, for ${1\leqs p\leqs\infty}$
can be construction through completion
of $C^{\infty}(T^{r}_{s}\cM)$ with respect to the norm,
with the case $p=2$ giving Hilbert space structure.
Denoting covariant derivatives of tensor fields as
$\nabla^{k}u^{a_1\cdots a_m} 
:=\nabla_{b_1}\cdots\nabla_{b_k}u^{a_1\cdots a_m}$,
where $k$ denotes the total number of derivatives
represented by the tensor indices $(b_1,\ldots,b_k)$,
for any non-negative integer $k$ and for any ${1\leqs p \leqs \infty}$,
the Sobolev norm on $C^{\infty}(T^{r}_{s}\cM)$ is given as follows,
\begin{equation}
\label{N-Wkp-norm}
\|u\|_{k,p}^p := \sum_{l=0}^k \,\|\nabla^{l}u\|_p^p.
\end{equation}
The Sobolev spaces $W^{k,p}(T^r_{s}\cM)$ of sections
of the $(r,s)$-tensor bundle can be constructed through
completion of $C^{\infty}(T^{r}_{s}\cM)$ with respect to this norm. 
For the remainder of this paper, we let 
$W^{k,p} = W^{k,p}(\cM)$ and $\bW^{k,p} = W^{k,p}(T\cM)$.
The Sobolev spaces $W^{k,p}(T^r_s\cM)$ are Banach
spaces, and the case $p=2$ is a
Hilbert space.
We have $L^p(T^r_s\cM)=W^{0,p}(T^s_r\cM)$ and $\|s\|_p =\|s\|_{0,p}$.
See the Appendix of~\cite{HNT07b} for a more careful construction
that includes real order Sobolev spaces of sections of vector bundles.

We will also need to consider Sobolev spaces of functions and tensor bundles 
on the boundary components of $\cM$.  If $\partial\cM =\Sigma_1\cup \Sigma_2$,
where each $\Sigma_1$ and $\Sigma_2$ are disjoint boundary components of $\cM$, we explicitly
denote the boundary component when referring to Sobolev spaces on that component.
For example, the space of scalar valued Sobolev functions on $\Sigma_i$ will be denoted
by $W^{k,p}(\Sigma_i)$ and the space of $(r,s)$-tensors by $W^{k,p}(T^r_s\Sigma_i)$.
We use the following notation for the norm that defines these spaces:
\begin{align}\label{eq1:8aug13}
\|u\|_{k,p;\Sigma_i}^p := \sum_{l=0}^k \,\|\nabla^{l}u\|_{p;\Sigma_i}^p.
\end{align}

For the boundary value problem that we consider in this paper, we will need
to form new Banach spaces from old Banach spaces using the direct product.
Given two Banach spaces $X$ and $Y$ with norms $\|\cdot\|_{X}$ and $\|\cdot\|_{Y}$,
\begin{align}\label{eq1:26july13}
X\times Y~~ \text{is a Banach space with norm}~~ \left( \|\cdot\|_{X}^q+\|\cdot\|_{Y}^q \right)^{1/q},
\quad q \geqs 1.
\end{align}
In particular, if $\Sigma_i$ represents a boundary component of $\cM$ for $i \in \{1,2\}$, we will have need
to consider spaces of the form $$W^{s,p} \times W^{t,q}(\Sigma_1)\times W^{t,q}(\Sigma_2),$$
with a norm given by the sum of appropriate powers of the norms of the 
respective spaces.

Let $C^\infty_{+}$ be the set of nonnegative smooth (scalar) functions on $\cM$.
Then we can define order cone
\begin{equation}
\label{E:wkp-cone}
W^{s,p}_{+} := \bigl\{ \phi \in W^{s,p} :
\langle\phi,\varphi\rangle\geqs 0 \quad \forall \, \varphi\in C^{\infty}_{+}
\, \bigr \},
\end{equation}
with respect to which the Sobolev spaces $W^{s,p}=W^{s,p}(\cM)$ are ordered Banach spaces.
Here $\langle\cdot,\cdot\rangle$ represents the (unique) extension of the $L^2$-inner product to a bilinear form ${W^{s,p}\otimes W^{-s,p'}\to\R}$, 
with $\frac1{p'}+\frac1p=1$.
The order relation is then $\phi\geqs\psi$ iff $\phi-\psi\in
W^{s,p}_{+}$.
We note that this order cone is normal only for $s=0$.

Given two ordered Banach spaces $X$ and $Y$ with order cones $X_+$, $Y_+$, we will have a need to define an order cone
on the Banach space $X\times Y$.  We define the order cone 
\begin{align}\label{sum-cone}
(X\times Y)_+ = \{ (x,y) \in X\times Y~:~x \in X_+, ~~ y \in Y_+\}.
\end{align}
In particular, if $\cM$ is a manifold with boundary $\partial\cM = \Sigma_1\cup \Sigma_2$, then we may define
an order cone on $W^{s,p}(\cM)\times W^{t,q}(\Sigma_1)\times W^{t,q}(\Sigma_2)$ using definition
\eqref{E:wkp-cone}.  See Appendix of~\cite{HNT07b}, where the key ideas of
ordered Banach spaces are reviewed.

\subsection{The Einstein Constraint Equations and the Conformal Formulation}
\label{subsec:conformal}

We give a quick overview of the Einstein constraint equations in
general relativity, and then define weak formulations that
are fundamental to both solution theory and the development
of approximation theory, following closely~\cite{HNT07b,HoTs10a}.

Let $(M,g_{\mu\nu})$ be a 4-dimensional spacetime, that is, $M$ is a
4-dimensional, smooth manifold, and $g_{\mu\nu}$ is a smooth,
Lorentzian metric on $M$ with signature $(-,+,+,+)$. Let
$\nabla_{\mu}$ be the Levi-Civita connection associated with the
metric $g_{\mu\nu}$.  The Einstein equation is
\[
G_{\mu\nu} = \kappa T_{\mu\nu},
\]
where $G_{\mu\nu} = R_{\mu\nu} - \frac{1}{2}R\,g_{\mu\nu}$ is the
Einstein tensor, $T_{\mu\nu}$ is the stress-energy tensor, and
$\kappa = 8\pi G/c^4$, with $G$ the gravitation constant and $c$ the
speed of light. The Ricci tensor is $R_{\mu\nu} =
R_{\mu\sigma\nu}{}^{\sigma}$ and $R= R_{\mu\nu}g^{\mu\nu}$ is the
Ricci scalar, where $g^{\mu\nu}$ is the inverse of $g_{\mu\nu}$,
that is $g_{\mu\sigma} g^{\sigma\nu} =\delta_{\mu}{}^{\nu}$. The
Riemann tensor is defined by
$R_{\mu\nu\sigma}{}^{\rho} w_{\rho} =\big(\nabla_{\mu}\nabla_{\nu}
-\nabla_{\nu}\nabla_{\mu}\bigr) w_{\sigma}$, where $w_{\mu}$ is any
1-form on $M$. The stress energy tensor $T_{\mu\nu}$ is assumed to
be symmetric and to satisfy the condition
$\nabla_{\mu}T^{\mu\nu} = 0$ and the {\bf dominant energy
condition}, that is, the vector $-T^{\mu\nu}v_{\nu}$ is timelike and
future-directed, where $v^{\mu}$ is any timelike and future-directed
vector field. In this section Greek indices $\mu$, $\nu$, $\sigma$,
$\rho$ denote abstract spacetime indices, that is, tensorial
character on the 4-dimensional manifold $M$. They are raised and
lowered with $g^{\mu\nu}$ and $g_{\mu\nu}$, respectively. 
Latin indices $a$, $b$, $c$, $d$ will denote tensorial
character on a 3-dimensional manifold.

The map $t:M\to \R$ is a {\bf time function} iff the function $t$ is
differentiable and the vector field $-\nabla^{\mu}t$ is a timelike,
future-directed vector field on $M$. Introduce the hypersurface $\cM
:=\{ x\in M : t(x)=0\}$, and denote by $n_{\mu}$ the unit 1-form
orthogonal to $\cM$. By definition of $\cM$ the form $n_{\mu}$ can
be expressed as $n_{\mu} = -\alpha\,\nabla_{\mu}t$, where $\alpha$,
called the lapse function, is the positive function such that
$n_{\mu} n_{\nu}\, g^{\mu\nu} = -1$. Let $\hh_{\mu\nu}$ and
$\hk_{\mu\nu}$ be the first and second fundamental forms of $\cM$,
that is,
\[
\hh_{\mu\nu} := g_{\mu\nu}- n_{\mu} n_{\nu},\qquad
\hk_{\mu\nu} := -\hh_{\mu}{}^{\sigma} \nabla_{\sigma} n_{\nu}.
\]
The Einstein constraint equations on $\cM$ are given by
\[
\bigl( G_{\mu\nu} -\kappa T_{\mu\nu}\bigr) \, n^{\nu} =0.
\]
A well known calculation allows us to express these equations
involving tensors on $M$ as equations involving {\em intrinsic}
tensors on $\cM$. The result is the following equations,
\begin{align}
\label{CE-def-H}
\hRi + \hk^2 - \hk_{ab}\hk^{ab} - 2\kappa \hrho &=0,\\
\label{CE-def-M}
\hD^a\hk - \hD_b\hk^{ab} + \kappa \hj^a &= 0,
\end{align}
where tensors $\hh_{ab}$, $\hk_{ab}$, $\hj_a$ and $\hrho$ on a
3-dimensional manifold are the pull-backs on $\cM$ of the tensors
$\hh_{\mu\nu}$, $\hk_{\mu\nu}$, $\hj_{\mu}$ and $\hrho$ on the
4-dimensional manifold $M$. We have introduced the energy density
$\hrho := n_{\mu} n_{\mu} T^{\mu\nu}$ and the momentum current
density $\hj_{\mu} := -\hh_{\mu\nu} n_{\sigma} T^{\nu\sigma}$. We
have denoted by $\hD_{a}$ the Levi-Civita connection associated to
$\hh_{ab}$, so $(\cM,\hh_{ab})$ is a 3-dimensional Riemannian
manifold, with $\hh_{ab}$ having signature $(+,+,+)$, and we use the
notation $\hh^{ab}$ for the inverse of the metric $\hh_{ab}$.
Indices have been raised and lowered with $\hh^{ab}$ and $\hh_{ab}$,
respectively. We have also denoted by $\hRi$ the Ricci scalar curvature of the metric $\hh_{ab}$. Finally, recall that the
constraint equations~\eqref{CE-def-H}-\eqref{CE-def-M} are indeed
equations on $\hh_{ab}$ and $\hk_{ab}$ due to the matter fields
satisfying the energy condition $-\hrho^2 +\hj_a\hj^a \leqs 0$
(with strict inequality holding at points on $\cM$ where $\hrho \neq 0$;
see~\cite{Wald84}), which is
implied by the dominant energy condition on the stress-energy tensor
$T^{\mu\nu}$ in spacetime.

{\bf The Conformal Formulation.}
Let $\phi$ denote a positive scalar field on $\cM$, and decompose the
extrinsic curvature tensor $\hk_{ab} = \hl_{ab} + \frac{1}{3}\hh_{ab} \htau$,
where $\htau := \hk_{ab}\hh^{ab}$ is the trace and then $\hl_{ab}$ is
the traceless part of the extrinsic curvature tensor. Then, introduce
the following conformal re-scaling:
\begin{equation}
\label{CE-def-mf}
\begin{aligned}
\hh_{ab} &=: \phi^4 \, h_{ab},&
\hl^{ab} &=: \phi^{-10} \,l^{ab},&
\htau &=: \tau,\\
\hj^a &=: \phi^{-10}\; j^a,&
\hrho &=: \phi^{-8}\, \rho.
\end{aligned}
\end{equation}
We have introduced the Riemannian metric $h_{ab}$ on the 3-dimensional
manifold $\cM$, which determines the Levi-Civita connection $D_a$, and
so we have that $D_a h_{bc}=0$. We have also introduced the symmetric,
traceless tensor $l_{ab}$, and the non-physical matter sources $j^a$
and $\rho$. The different powers of the conformal re-scaling above are
carefully chosen so that the 
constraint equations \eqref{CE-def-H}-\eqref{CE-def-M} transform into 
the following
equations
\begin{gather}
\label{CE-cr1H}\textstyle
-8 \Delta \phi + \Ri \phi + \frac{2}{3}\tau^2 \phi^5
- l_{ab}l^{ab} \phi^{-7} -2\kappa \rho\phi^{-3} =0,\\
\label{CE-cr1M}\textstyle
-D_bl^{ab} + \frac{2}{3} \phi^6 D^a \tau +\kappa j^a =0,
\end{gather}
where in equation above, and from now on, indices of unhatted fields
are raised and lowered with $h^{ab}$ and $h_{ab}$ respectively. We
have also introduced the {\bf Laplace-Beltrami operator}
with respect to the metric $h_{ab}$, acting on smooth scalar fields;
it is defined as follows
\begin{equation}
\label{E:laplace-beltrami}
\Delta \phi:= h^{ab}D_aD_b\phi.
\end{equation}
Equations~\eqref{CE-cr1H}-\eqref{CE-cr1M} can be obtained by a
straightforward albeit long computation. In order to perform this
calculation it is useful to recall that both $\hD_a$ and $D_a$ are
connections on the manifold $\cM$, and so they differ on a tensor
field $C_{ab}{}^c$, which can be computed explicitly in terms of
$\phi$, and has the form
\[
C_{ab}{}^c = 4 \delta_{(a}{}^cD_{b)} \ln(\phi)
- 2 h_{ab}h^{cd}D_d \ln(\phi).
\]
We remark that the power four on the re-scaling of the metric
$\hh_{ab}$ and $\cM$ being 3-dimensional imply that $\hRi =\phi^{-5}
(\Ri\phi - 8\Delta\phi)$, or in other words, that $\phi$ satisfies
the {\bf Yamabe-type problem}:
\begin{equation}
  \label{E:yamabe}
-8\Delta \phi+\Ri \phi - \hRi \phi^5 = 0, \quad \phi > 0,
\end{equation}
where $\hRi$ denotes the scalar curvature corresponding to the physical metric ${\hh_{ab} = \phi^4 h_{ab}}$.
Note that for any other power in the re-scaling, terms proportional 
to ${h^{ab}(D_a\phi)(D_b\phi)/\phi^2}$ appear in the transformation.
The set of all metrics on a compact manifold can be classified into the
three disjoint Yamabe classes $\cY^{+}(\cM)$, $\cY^{0}(\cM)$, and
$\cY^{-}(\cM)$, corresponding to whether one can conformally transform the metric into a metric with strictly positive,
zero, or strictly negative scalar curvature, respectively, cf.~\cite{jLtP87} (See also the Appendix of~\cite{HoTs10a}).
We note that the {\bf Yamabe problem} is to determine,
for a given metric $h_{ab}$, whether there exists a conformal
transformation $\phi$ solving~(\ref{E:yamabe}) such that $\hRi = \mathrm{const}$.
Arguments similar to those above for $\phi$ force the power negative ten 
on the re-scaling of the tensor $\hl^{ab}$ and $\hj^a$, so terms proportional
to $(D_a\phi)/\phi$ cancel out in \eqref{CE-cr1M}. Finally, the
ratio between the conformal re-scaling powers of $\hrho$ and $\hj^a$
is chosen such that the inequality 
$-\rho^2 + h_{ab} j^aj^b \leqs 0$
implies the inequality 
$-\hrho^2 + \hh_{ab}\hj^a\hj^b \leqs 0$.
For a complete discussion of all possible choices of re-scaling
powers, see the Appendix of~\cite{HNT07b}.

There is one more step to convert the original constraint equation
\eqref{CE-def-H}-\eqref{CE-def-M} into a determined elliptic
system of equations. This step is the following: Decompose the
symmetric, traceless tensor $l_{ab}$ into a divergence-free part
$\sigma_{ab}$, and the symmetrized and traceless gradient of a vector,
that is, $l^{ab} =: \sigma^{ab} + (\cL w)^{ab}$, where
$D_a\sigma^{ab}=0$ and we have introduced the {\bf conformal Killing
operator} $\cL$ acting on smooth vector fields and defined as follows
\begin{equation}
\label{CF-def-CK}\textstyle
(\cL w)^{ab} := D^a w^b + D^b w^a - \frac{2}{3}(D_c w^c) h^{ab}.
\end{equation}
Therefore, the constraint equations~\eqref{CE-def-H}-\eqref{CE-def-M} are
transformed by the conformal re-scaling into the following equations
\begin{gather}
\label{CE-cr2H}\textstyle
\hspace*{-0.15cm}
- 8 \Delta \phi + \Ri \phi 
+ \frac{2}{3}\tau^2 \phi^5
- [\sigma_{ab}+(\cL w)_{ab}] [\sigma^{ab}+(\cL w)^{ab}]\phi^{-7} 
- 2\kappa \rho \phi^{-3} =0,\\
\label{CE-cr2M}\textstyle
-D_b(\cL w)^{ab} + \frac{2}{3} \phi^6 D^a \tau +\kappa j^a =0.
\end{gather}
In the next section we interpret these equations above as partial
differential equations for the scalar field $\phi$ and the vector
field $w^a$, while the rest of the fields are considered as given
fields. Given a solution $\phi$ and $w^a$ 
of equations~\eqref{CE-cr2H}-\eqref{CE-cr2M}, the physical metric $\hh_{ab}$
and extrinsic curvature $\hk^{ab}$ of the hypersurface $\cM$ are given
by
\[\textstyle
\hh_{ab} = \phi^4 h_{ab},\qquad
\hk^{ab} = \phi^{-10}
[\sigma^{ab} + (\cL w)^{ab}] + \frac{1}{3}\, 
\phi^{-4} \tau h^{ab},
\]
while the matter fields are given by Eq~(\ref{CE-def-mf}).

From this point forward, for simplicity we will denote the Levi-Civita connection
of the metric $h_{ab}$ on the 3-dimensional manifold $\cM$
as $\nabla_a$ rather than $D_a$, and the Ricci scalar of
$h_{ab}$ will be denoted by $R$ instead of $\Ri$.
Let $(\cM, h)$ be a 3-dimensional Riemannian manifold, where $\cM$ is
a smooth, compact manifold with non-empty boundary $\partial\cM$, and
$h\in C^{\infty}(T^0_2\cM)$ is a positive definite metric. 
With the shorthand notation ${C^{\infty}=C^\infty(\cM\times\R)}$ 
and ${\biC^{\infty}=C^\infty(T\cM)}$,
let
$L: C^{\infty}\to C^{\infty}$ and $\IL :\biC^{\infty}\to\biC^{\infty}$
be the operators with
actions on $\phi\in C^{\infty}$ and $\biw\in\biC^{\infty}$ given by
\begin{align}
\label{CF-def-L}
L\phi &:= -\Delta \phi,\\
\label{CF-def-IL}
(\IL \biw)^a &:= -\nabla_b (\cL \biw)^{ab},
\end{align}
where $\Delta$ denotes the Laplace-Beltrami operator 
defined in~\eqref{E:laplace-beltrami}, 
and where $\cL$ denotes the conformal Killing operator 
defined in~\eqref{CF-def-CK}.
We will also use the index-free notation $\IL\biw$ and $\cL\biw$.

The freely specifiable functions of the problem are a scalar function
$\tau$, interpreted as the trace of the physical extrinsic curvature;
a symmetric, traceless, and divergence-free, contravariant, two index 
tensor $\sigma$; the non-physical energy density $\rho$ and the non-physical
momentum current density vector $\bij$ subject to the requirement
$-\rho^2 +\bij\cdot\bij \leqs 0$.
The term non-physical refers here to
a conformal rescaled field, while physical refers to a conformally
non-rescaled term. The requirement on $\rho$ and $\bij$ mentioned
above and the particular conformal rescaling used in the
semi-decoupled decomposition imply that the same inequality is
satisfied by the physical energy and momentum current densities. This
is a necessary condition (although not sufficient) in order that the
matter sources in spacetime satisfy the dominant energy
condition. The definition of various energy conditions can be found
in~\cite[page 219]{Wald84}. Introduce the non-linear operators 
${f: C^{\infty}\times\biC^{\infty}\to C^{\infty}}$ and 
${\IF :C^{\infty}\to\biC^{\infty}}$ given by
\begin{equation*}
f(\phi,\biw) = a_{\tau} \phi^5 + a_{\tiR} \phi - a_{\rho} \phi^{-3}
- a_{w}\phi^{-7},
\quad\textrm{and}\quad
\IF (\phi) = \bib_{\tau} \, \phi^6 + \bib_{j},
\end{equation*}
where the coefficient functions are defined as follows
\begin{equation}
\label{CF-def-coeff2}
\begin{aligned}
a_{\tau} &:= \textstyle\frac{1}{12}\tau^2,&
a_{\tiR} &:= \textstyle\frac{1}{8}R,&
a_{\rho} &:= \textstyle\frac{\kappa}{4} \rho,\\
a_{\biw} &:= \textstyle\textstyle\frac{1}{8}(\sigma +\cL\biw)_{ab}(\sigma + \cL \biw)^{ab},&
\bib_{\tau}^a &:= \textstyle\frac{2}{3}\nabla^a \tau,&
\bib_{j}^a &:= \kappa j^a.
\end{aligned}
\end{equation}
Notice that the scalar coefficients $a_{\tau}$, $a_{w}$, and
$a_{\rho}$ are non-negative, while there is no sign restriction on
$a_{\tiR}$.

With these notations, the {\bf classical formulation} (or the strong formulation) of the
coupled Einstein constraint equations reads as: Given
the freely specifiable smooth functions $\tau$, $\sigma$, $\rho$, and
$\bij$ in $\cM$, find a scalar field $\phi$ and a vector field
$\biw$ in $\cM$ solution of the system
\begin{equation}
\label{CF-LY}
L \phi + f(\phi,\biw) = 0
\qquad\textrm{and}\qquad
\IL \biw  + \IF(\phi) =0
\qquad\textrm{in }\cM.
\end{equation}

\subsection{Boundary Conditions}
\label{subsec:boundary} 

Following~\cite{HoTs10a},
the two main types of boundary conditions that we consider in this paper are exterior boundary conditions and interior boundary conditions.
Exterior boundary conditions occur when the asymptotic ends of a manifold are removed and one needs to impose the correct decay conditions.
The interior boundary conditions arise when singularities are excised from the manifold and then conditions are imposed on the boundary
so that the region is either a trapped or marginally trapped surface (cf. \cite{HoTs10a,SD04,DM05a} ).  
We let $\partial \cM = \Sigma_I\cup \Sigma_E$, where $\Sigma_I$ and $\Sigma_E$
denote the interior and exterior boundary, respectively.  Moreover, we assume that the interior and exterior boundaries
are the union of finitely many disjoint components:
$$\Sigma_I = \bigcup_{i=1}^{M} \Sigma_{i} \quad \text{and} \quad \Sigma_E = \bigcup_{i=M+1}^N \Sigma_{i}.$$

On a $3$-dimensional manifold, the exterior boundary condition for the conformal factor $\phi$ is that it must satisfy
\begin{align}\label{eq1:25jun13}
\partial_{r}\phi+\frac{1}{r}(\phi-1) = \cO(r^{-3}), 
\end{align}
where $r$ is the flat-space radial coordinate.  This condition is chosen to ensure that the conformal
data accurately models initial data for the asymptotically Euclidean case.  More specifically, this condition is chosen
to ensure the correct decay estimates for $\phi$ and to give accurate values
for the total energy \cite{HoTs10a,tPjY82}.  

The solution $\biw$ to the momentum constraint
must also satisfy certain Robin type conditions to accurately
model asymptotically Euclidean data.  In \cite{tPjY82},
the vector Robin condition
\begin{align}\label{eq1:11july13}
(\cL \biw)^{bc}\nu_c\left(\delta^a_b-\frac12\nu^a\nu_b\right)+\frac{6}{7r}\biw^b\left(\delta^{a}_b-\frac18\nu^a\nu_b\right) = \mathcal{O}(r^{-3})
\end{align}  
is given for a $3$-dimensional asymptotically Euclidean manifold.  
Here $\nu$ is the outward pointing normal vector field to $\Sigma_E$ with respect to the
non physical metric $g$ and $r$ is the radius
of a large spherical domain. Taking the right hand side in the above expression to be zero, and 
using the fact that $(\delta^a_b+\nu^a\nu_b)$ is the inverse of $(\delta^a_b-\frac12 \nu^a\nu_b)$,
we can rewrite \eqref{eq1:11july13} as
\begin{align}\label{eq3:11july13}
(\cL \biw)^{ab}\nu_b +\frac{6}{7r}\biw^b\left(\delta^a_b+\frac34 \nu^a\nu_b\right) = 0.
\end{align}
Therefore, we impose the general vector Robin condition on the momentum constraint
for the exterior boundary:
\begin{align}\label{eq2:11july13}
(\cL \biw)^{ab}\nu_b + C^a_b\biw^b = 0,
\end{align}
which is general enough to include \eqref{eq3:11july13}.

There are many different interior boundary conditions that have been imposed in the
literature.  For the sake of completeness,
we will give a brief review of the boundary conditions mentioned in \cite{HoTs10a},
where Holst and Tsogtgerel compile a complete list of interior conditions
modeling marginally trapped surfaces.  While the following interior boundary
conditions are presented for $n$-dimensional manifolds, we will 
focus on the boundary condition given in~\eqref{eq8:26jun13} in the
$3$-dimensional case. 

Let $\Sigma_i$ denote an interior boundary component and let $\hat{\nu}$ be the outward pointing normal vector with respect to
the physical metric $\hat{g}$.  The expansion scalars corresponding to the outgoing and ingoing future directed geodesics 
to $\Sigma_i$ are then
given by
\begin{align}\label{eq2:26jun13}
\hat{\theta}_{\pm}= \mp (n-1)\hat{H}+\text{tr}_{\hat{g}}\hat{K}-\hat{K}(\hat{\nu},\hat{\nu}),
\end{align}
where $(n-1)\hat{H} = \text{div}_{\hat{g}}\hat{\nu}$ is the mean extrinsic curvature of $\Sigma_i$.  The surface
$\Sigma_i$ is called a trapped surface if $\hat{\theta}_{\pm} <0$ and a marginally trapped surface if 
$\hat{\theta}_{\pm} \leqs 0$.  See \cite{SD04,DM05a,Wald84} for details.

Writing the expansion scalars in terms of the conformal quantities and setting $\overline{q} = \frac{n}{n-2}$ as in \cite{HoTs10a}, we have
that
\begin{align}\label{eq3:26jun13}
\hat{\theta}_{\pm} = \mp (n-1)\phi^{-\overline{q}}\left( \frac{2}{n-2}\partial_{\nu}\phi + H\phi \right) + (n-1)\tau-\phi^{-2\overline{q}}S(\nu,\nu),
\end{align}
where $\nu = \phi^{\overline{q}-1}\hat{\nu}$ is the unit normal with respect to $g$, and $\partial_{\nu}\phi$ is the derivative of $\phi$
along $\nu$.   In~\eqref{eq3:26jun13}, we have also used that fact that the mean curvature $\hat{H}$ satisfies
\begin{align}
\hat{H} = \phi^{-\overline{q}}\left(\frac{2}{n-2}\partial_{\nu}\phi + H\phi\right),
\end{align}
where $H$ is the mean curvature with respect to $g$.

As in \cite{HoTs10a}, we let $\theta_{+}=\phi^{\overline{q}-e}\hat{\theta_{+}}$ be the specified, scaled expansion factor for some $e \in \mathbb{R}$, and obtain
\begin{align}\label{eq4:26jun13}
\frac{2(n-1)}{n-2}\partial_{\nu}\phi + (n-1)H\phi-(n-1)\tau \phi^{\overline{q}}+S(\nu,\nu)\phi^{-\overline{q}}+\theta_{+}\phi^e = 0.
\end{align} 
Similarly, by specifying $\theta_{-}=\phi^{\overline{q}-e}\hat{\theta_{-}}$ for some $e \in \mathbb{R}$,
we obtain
\begin{align}\label{eq5:26jun13}
\frac{2(n-1)}{n-2}\partial_{\nu}\phi + (n-1)H\phi+(n-1)\tau \phi^{\overline{q}}-S(\nu,\nu)\phi^{-\overline{q}}-\theta_{-}\phi^e = 0.
\end{align}
In~\eqref{eq4:26jun13}, $\theta_{-}$ remains unspecified, and in~\eqref{eq5:26jun13}, $\theta_{+}$ is unspecified.  So in either 
case, to ensure that $\theta_{\pm} \leqs  0$, conditions have to be imposed on either $\tau$ or $S$ in order to ensure that the unspecified
expansion factor satisfies the marginally trapped surface condition.  In \cite{HoTs10a}, Holst and Tsogtgerel developed general conditions
on the initial data to ensure that the trapped surface conditions are satisfied.  In the case of \eqref{eq4:26jun13} with specified $\theta_{+}$,
one assumes that $\phi_-$ satisfies $\phi_- \leqs \phi$, $\tau \leqs  0$ on $\Sigma_I$, $e = - \overline{q}$, and requires either that
\begin{align}\label{eq6:26jun13}
&S(\nu,\nu) \leqs  0,\\
&2|S(\nu,\nu)|+|\theta_{+}| \leqs  2(n-1)|\tau|\phi^{2\overline{q}}_{-}, \nonumber
\end{align}
or that
\begin{align}\label{eq7:26jun13}
&S(\nu,\nu) \geqs 0,\\
&|\theta_+| \leqs  2S(\nu,\nu)+2(n-1)|\tau|\phi_-^{2\overline{q}}. \nonumber
\end{align}

In the case of \eqref{eq5:26jun13}, where $\theta_-$ is specified, one assumes that
$e = \overline{q}$, $S(\nu,\nu) \geqs 0$, and
\begin{align}\label{eq8:26jun13}
2(n-1)\tau + |\theta_-| \leqs  2S(\nu,\nu)\phi_{+}^{-2\overline{q}}~~~\quad \text{on $\Sigma_I$},
\end{align}
where $\phi_+$ satisfies $\phi_+ \geqs \phi$.  In \cite{HoTs10a}, the authors
assume that $\tau \geqs 0$ on $\Sigma_I$.  However, in Theorem~\ref{T:main2} we assume that $\tau \leqs  0$ on $\Sigma_I$
and in Theorem~\ref{T:main1} we only assume that $\tau$ satisfies \eqref{eq8:26jun13}.

Conditions \eqref{eq4:26jun13}-\eqref{eq5:26jun13} are nonlinear, Robin conditions on the inner boundary components for the conformal
factor $\phi$.  Equations \eqref{eq6:26jun13}-\eqref{eq7:26jun13} constitute Robin boundary conditions on the inner boundary
components for the momentum constraint, which we will discuss in more detail in the next section when we formulate the 
our boundary value problem.  See \cite{HoTs10a,SD04,DM05a} for a complete discussion of the boundary conditions stated
above.

\section{Overview of the Main Results}
\label{sec:main}

The main results for this paper concern the existence of far-from-CMC and near-CMC solutions to the conformal formulation 
of the Einstein constraint equations on a compact, $3$-dimensional manifold $\cM$
with boundary $\Sigma$.  We assume that 
\begin{equation}\label{eq6:11july13}
\partial\cM =  \Sigma_I \cup \Sigma_E,
\end{equation}
where the boundary segments $\Sigma_I$ and $\Sigma_E$ are decomposed 
further into finite segments as
\begin{equation*}
\Sigma_I = \bigcup_{i=1}^{M} \Sigma_{i},
\quad \text{and} \quad
\Sigma_E = \bigcup_{i=M+1}^N \Sigma_{i},~~~(M < N),
\quad \text{with} \quad
\Sigma_i \cap \Sigma_j = \emptyset \quad \text{if ~~$i\ne j$}. \nonumber
\end{equation*}
We show that under certain conditions, the following system
\begin{align}
-\Delta \phi + a_R\phi+a_{\tau}\phi^5-a_{\bw}\phi^{7}-a_{\rho}\phi^{-3} &= 0,\label{eq4:11july13}\\
\IL \biw+b_{\tau}\phi^6+\bb_j &= 0 , \label{eq2:8aug13}
\end{align}
subject to the boundary conditions
\begin{align}
\partial_{\nu}\phi+\frac12H\phi+\left(\frac12\tau -\frac14\theta_- \right)\phi^3-\frac14S(\nu,\nu)\phi^{-3} &= 0, ~~~~ \text{on $\Sigma_I$}, \label{eq5:11july13}\\
(\cL \biw)^{ab}\nu_b &= V^a,~~~~\text{on $\Sigma_I$}, \label{eq4:8aug13}\\
\partial_{\nu}\phi+c\phi &= g, ~~~~ \text{on $\Sigma_E$},\label{eq3:8aug13} \\
(\cL\biw)^{ab}\nu_b +C^a_b\biw^b &= 0, ~~~~ \text{on $\Sigma_E$}, \label{eq5:8aug13}
\end{align}
has a solution.  In~\eqref{eq5:11july13}-\eqref{eq5:8aug13} we assume that 

\begin{align}\label{eq1:12july13}
&c > 0, \quad g>0 \quad \mbox{ and } \quad g = \delta(c+ \cO(R^{-3})), \quad \delta > 0, \\ 
&\exists \alpha> 0 \quad \mbox{ such that } \quad \int_{\partial \cM} C_{ab}V^aV^b \geqs \alpha|V|_{L^2(\partial\cM)}, \quad \forall V \in {\bf L}^2. \nonumber
\end{align}
The coefficients $a_R, a_{\tau}, a_{\biw}$ and $a_{\rho}$ are defined in \eqref{CF-def-coeff2},
and $H$ and $\theta_-$ are the mean extrinsic curvature for the boundary and expansion factor for the incoming null geodesics.  The operators $\cL$ and $\IL$
are the conformal Killing operator and its divergence, defined in \eqref{CF-def-L}.  

\begin{remark}\label{rem5:27sep13}
In~\eqref{eq4:8aug13}, the vector $\bV$ will be chosen so that 
\begin{align}\label{eq2:27sep13}
V^a\nu_a = (2\tau+|\theta_-|/2)B^6 -\sigma(\nu,\nu),
\end{align}
where $B$ is a positive function.
The 
condition implies that $S(\nu,\nu) = (2\tau+|\theta_-|/2)B^6$, 
which is similar to the marginally trapped surface condition 
\eqref{eq8:26jun13}.
The general approach is to solve \eqref{eq6:11july13} with boundary 
conditions \eqref{eq5:11july13}-\eqref{eq5:8aug13}, and then argue that we can
choose $B> \|\phi\|_{\infty}$ sufficiently large so that the marginally trapped surface condition is satisfied.
\end{remark}

Now that we have clarified the statement of our problem,
we can formally state our main results as the following two theorems.
Our first main result (Theorem~\ref{T:main2} below)
covers both the Near-CMC and CMC cases.

\begin{theorem}\label{T:main2}
{\bf (Near-CMC and CMC $W^{s,p}$ solutions,
$p \in (1,\infty)$, $s \in (1+\frac3p,\infty)$)}
\\
Let $(\cM,h_{ab})$ be a $3$-dimensional, compact Riemannian manifold 
with boundary satisfying the conditions \eqref{eq6:11july13}.
Let $h_{ab} \in W^{s,p}(T^0_2\cM)$, where 
$p \in (1,\infty)$ and $s \in (1+\frac3p,\infty)$ are given.
With $d= s-\frac3p$, select $q$ and $e$ to satisfy: 

\begin{itemize}
\item[$\bullet$] $\frac1q \in (0,1)\cap[\frac{3-p}{3p},\frac{3+p}{3p}]\cap[\frac{1-d}{3},\frac{3+sp}{6p})$,
\item[$\bullet$] $e \in [1,\infty)\cap[s-1,s]\cap[\frac3q+d-1, \frac{3}{q}+d]$.
\end{itemize}
Let Eq~\eqref{eq1:12july13} hold and assume the data satisfies:
\begin{itemize}
\item[$\bullet$] $ \theta_-  \in W^{s-1-\frac1p,p}(\Sigma_I),$
\item[$\bullet$] $c,g \in W^{s-1-\frac1p,p}(\Sigma_E),$
\item[$\bullet$] $ C^a_b \in W^{e-1-\frac1q,q}(T^1_1\Sigma_E),$ 
\item[$\bullet$] $\bV\in \bW^{e-1,q},$~~~$ V^a\nu_a = (2\tau + |\theta_-|/2)B^6-\sigma(\nu,\nu), $ 
\item[$\bullet$] $\tau \in W^{s-1,p}~~\text{if $e \geqs 2$, and $\tau \in W^{1,z}\cap L^{\infty}$ otherwise, with $z= \frac{3p}{3+\max\{0,2-s\}p}$,}$
\item[$\bullet$] $~(4\tau^{\tivee}+|\theta|^{\tivee}) >0$ on $\Sigma_I$, 
\item[$\bullet$] $ \sigma \in W^{e-1,q},$
\item[$\bullet$] $ \rho \in W^{s-2,p}_+,$
\item[$\bullet$] $ \bj \in \bW^{e-2,q}.$
\end{itemize}
In addition, assume that $a_{\tau}^{\tivee} > \ttk_1$, 
where $a_{\tau}$ is defined in \eqref{CF-def-coeff2}, 
and where
\begin{align}
\ttk_1 = 2 C^2( \|\bib_{\tau}\|_{z})^2,
\end{align}
with $C$ is a positive constant.
If at least one of the following additional conditions hold:
\begin{itemize}
\item[(a)] $\rho^{\tivee} > 0$,

\item[(b)] $a_{\sigma}^{\tivee} $ is sufficiently large, 
\end{itemize}
where $a_{\sigma}$ is defined in \eqref{CF-def-coeff2}, 
then there exists a solution $\phi \in W^{s,p}$ with $\phi>0$ 
and $\bw \in \bW^{e,q}$ to equations \eqref{eq4:11july13}-\eqref{eq2:8aug13}
with boundary conditions \eqref{eq5:11july13}-\eqref{eq1:12july13}.  
Moreover, with an additional smallness assumption on $\tau$ on $\Sigma_I$, 
the marginally trapped surface boundary condition in \eqref{eq8:26jun13} 
is satisfied.
\end{theorem}

Our second main result (Theorem~\ref{T:main1} below)
covers three distinct non-CMC cases for
which both the Near-CMC and CMC assumptions are violated.
Case (a) in Theorem~\ref{T:main1}
puts no restrictions on the size of $\tau$ or $D\tau$, but
requires a smallness condition on the exterior boundary data that leads 
to a departure of the model from faithfully approximating asymptotically 
Euclidean boundaries, while preserving the trapped surface conditions.
Cases (b) and (c) remove the smallness condition on the exterior boundary
to faithfully preserve the asymptotically Euclidean model by introducing
smallness conditions on $\tau$ and/or $D\tau$ to satisfy the trapped 
surface conditions, yet still allow for violation of the near-CMC condition.

\begin{theorem}\label{T:main1}
{\bf (Non-CMC $W^{s,p}$ solutions,
$p \in (1,\infty)$, $s \in (1+\frac3p,\infty)$)}
\\
Let $(\cM,h_{ab})$ be a $3$-dimensional, compact Riemannian manifold 
with boundary satisfying the conditions \eqref{eq6:11july13}.
Let $h_{ab} \in W^{s,p}(T^0_2\cM)$ and be in $\cY^+$, 
where $p \in (1,\infty)$ and $s \in (1+\frac3p,\infty)$ are given.
With $d= s-\frac3p$, select $q$ and $e$ to satisfy: 

\begin{itemize}
\item[$\bullet$] $\frac1q \in (0,1)\cap[\frac{3-p}{3p},\frac{3+p}{3p}]\cap[\frac{1-d}{3},\frac{3+sp}{6p})$,
\item[$\bullet$] $e \in [1,\infty)\cap[s-1,s]\cap[\frac3q+d-1, \frac{3}{q}+d]$.
\end{itemize}
Let Eq~\eqref{eq1:12july13} hold and assume the data satisfies:
\begin{itemize}
\item[$\bullet$] $ \theta_-  \in W^{s-1-\frac1p,p}(\Sigma_I)\cap L^{\infty}(\Sigma_I),$
\item[$\bullet$] $c,g \in W^{s-1-\frac1p,p}(\Sigma_E),$
\item[$\bullet$] $ C^a_b \in W^{e-1-\frac1q,q}(T^1_1\Sigma_E),$ 
\item[$\bullet$] $\bV\in \bW^{e-1,q}$,~~~$V^a\nu_a = (2\tau+|\theta_-|/2)B^6-\sigma(\nu,\nu),$
\item[$\bullet$] $\tau \in W^{s-1,p}~~\text{if $s \geqs 2$, and $\tau \in W^{1,z}\cap L^{\infty}$ otherwise, with $z= \frac{3p}{3+\max\{0,2-s\}p}$,}$
\item[$\bullet$] $(4\tau^{\tivee}+|\theta|^{\tivee}) >0$ on $\Sigma_I$,
\item[$\bullet$] $ \sigma \in W^{e-1,q}~~\text{with $\|\sigma\|_{\infty}$ sufficiently small,}\hspace{3.75 in}$
\item[$\bullet$] $ \rho \in W_+^{s-2,p}\cap L^{\infty}\backslash \{0\},~~\text{with $\|\phi\|_{\infty}$ sufficiently small,}$
\item[$\bullet$] $ \bj \in \bW^{e-2,q}~~\text{with $\|\bj\|_{e-2,q}$ sufficiently small.}$
\end{itemize}
Additionally assume that at least one of the following hold:
\begin{itemize}

\item[(a)] $\delta>0$ is sufficiently small in~\eqref{eq1:12july13};

\item[(b)] $a_R^{\tivee} >0$ is sufficiently large;

\item[(c)] $\|\theta_-\|_{\infty}$ is sufficiently small, and $D\tau$ is sufficiently small. 

\end{itemize}
Then:

Case (a):
The function $B$ in~\eqref{eq2:27sep13} can be chosen so that the marginally trapped surface boundary condition in \eqref{eq8:26jun13} is satisfied, and subsequently there exists a solution ${\phi \in W^{s,p}}$ with ${\phi>0}$ and $\biw \in \bW^{e,q}$ to equations \eqref{eq4:11july13}-\eqref{eq2:8aug13} with boundary conditions \eqref{eq5:11july13}-\eqref{eq1:12july13}.

Cases (b) and (c): 
There exists a solution $\phi \in W^{s,p}$ with $\phi>0$ and $\biw \in \bW^{e,q}$ to equations \eqref{eq4:11july13}-\eqref{eq2:8aug13} with boundary conditions \eqref{eq5:11july13}-\eqref{eq1:12july13}.
With an additional smallness assumption on $\tau$ on $\Sigma_I$, the marginally trapped surface boundary condition in \eqref{eq8:26jun13} may be satisfied.
\end{theorem}

\begin{remark}\label{rem3:8oct13}
We pointed out earlier that while Holst and Tsogtgerel in~\cite{HoTs10a} 
proved existence (and when possible, uniqueness) of solutions to the 
Lichnerowicz equation for a rather extensive collection of boundary 
conditions similar to those discussed in Section~\ref{subsec:boundary}, 
they did not attempt to prove existence of CMC solutions to the coupled system
\eqref{eq4:11july13}-\eqref{eq2:8aug13} satisfying the marginally trapped 
surface conditions \eqref{eq6:26jun13}-\eqref{eq8:26jun13}.
The dependence of the coefficient $S(\nu,\nu)$ on the size of the conformal 
factor $\phi$ as required by the marginally trapped surface conditions
leaves the equations coupled even in the CMC case; hence any results for the 
CMC case would require non-CMC techniques, and this was left for this 
second paper.
In the case of closed manifolds, the CMC condition decouples the equations
so that obtaining existence results for the Lichnerowicz equation for rough
metrics is essentially sufficient (modulo some well-known estimates for
the conformal Killing operator) to obtain analogous rough metric results
for the (decoupled) system, allowing for lower regularity solutions than in 
the non-CMC case.
This is due to the fact that the CMC decoupling in the case of closed manifolds
frees one from estimating $\biw$ in terms of $\phi$; such estimates require 
additional regularity assumptions on $h$.
Holst and Tsogtgerel in~\cite{HoTs10a} proved existence of Lichnerowicz 
solutions for metrics $h \in W^{s,p}$, with $s > \frac3p$; these results
are a direct analogue of the CMC results in~\cite{HNT07b}, made possible
by ignoring this boundary coupling that we now account for here.
In the case of compact manifolds with boundary, where the marginally trapped 
surface conditions produce a boundary coupling between the momentum and 
Hamiltonian constraints that remains even in the CMC setting, there appears 
to be little hope of obtaining a lower regularity CMC result using these
same techniques, along the lines of what was possible in~\cite{HNT07b,HoTs10a}.
Therefore, Theorem~\ref{T:main2} above will be sufficient for obtaining the 
roughest possible CMC solutions using our approach here, and we will not
state explicitly a separate CMC existence result.
\end{remark}

For the above problem, one views each $\Sigma_I = \bigcup_{i=1}^M \Sigma_i$ as the interior black hole regions contained within a compact
subset of an asymptotically Euclidean manifold.  The exterior boundary conditions \eqref{eq3:8aug13}-\eqref{eq5:8aug13} come
from the decay conditions \eqref{eq1:25jun13}-\eqref{eq1:11july13}.  We note that the condition $g = \delta(c+\cO(R^{-3}))$
in \eqref{eq1:12july13} is consistent with asymptotically Euclidean decay when the conformal factor $\phi$ tends to $\delta > 0$ at infinity.  
The interior boundary conditions \eqref{eq5:11july13}-\eqref{eq4:8aug13} on $\phi$ and $\biw$ are derived
from the marginally trapped surface condition \eqref{eq8:26jun13} discussed in Section~\ref{subsec:boundary} in the event that the expansion factor $\theta_-$ is specified.  
The components $\Sigma_E = \bigcup_{i=M+1}^N\Sigma_i$ represent the asymptotic
ends of this manifold and the exterior boundary conditions \eqref{eq3:8aug13}-\eqref{eq5:8aug13} are imposed so that the solutions on the compact region $\cM$ 
exhibit the correct asymptotic behavior.  
 
In order to solve the above problem, we will require the coupled fixed point 
solution framework used in \cite{HNT07b}, which is based on 
Theorem~\ref{T:FIXPT2} below.    
 The are two main difficulties we encounter in attempting to apply the solution framework of \cite{HNT07b} to this particular problem.
 The first difficulty lies in reformulating 
 the conformal equations \eqref{eq4:11july13}, with the boundary conditions \eqref{eq5:11july13}-\eqref{eq5:8aug13}, in a manner
 that allows us to utilize this framework. This requires adapting many of the supporting results in \cite{HNT07b} to
 incorporate our boundary problem and reformulating the boundary problem itself.  Following the approach taken in~\cite{HoTs10a}, we must 
 formulate the conformal equations with boundary conditions \eqref{eq5:11july13}-\eqref{eq5:8aug13} 
 as a nonlinear, fixed point problem on a certain closed, convex and bounded subset of a Banach space. 
 We must then show that the operator defining our fixed point problem is continuous, compact, and invariant on this subspace.
 In order to show that this operator is invariant, one requires what are known as global
 sub- and super-solutions for the above system.
 While this was done for the case on the Lichnerowicz equation on compact manifolds with boundary in~\cite{HoTs10a}, determining global sub- and super-solutions in the non-CMC setting for our boundary value problem is the other primary difficulty in applying the fixed point framework from~\cite{HNT07b}.
 
 In the following section, we restate the fixed point theorems used in \cite{HNT07b} for convenience.  Then
 the rest of the paper is dedicated to reformulating our problem in this framework, adapting results from \cite{HNT07b} and \cite{HoTs10a},
 and then determining global sub-and super-solutions.

\subsection{Coupled Fixed Point Theorems and Outline of Proofs}\label{sec1:26july13}

In Theorem~\ref{T:FIXPT2} below 
we give some abstract fixed-point results which form the basic
framework for our analysis of the coupled constraints.
These topological fixed-point theorems will be the main tool by which
we shall establish Theorems~\ref{T:main1}-\ref{T:main2} above.
They have the important feature that the required properties of the
abstract fixed-point operators $S$ and $T$ appearing in 
Theorem~\ref{T:FIXPT2} 
below can be established in the case of the Einstein constraints without
using the near-CMC condition; this is not the case for fixed-point
arguments for the constraints based on $k$-contractions
(cf.~\cite{jIvM96,ACI08}) and the Implicit Function Theorem (cf.~\cite{yCBjIjY00}) 
which require near-CMC conditions.
The bulk of the paper then involves establishing the
required properties of $S$ and $T$ without using the 
near-CMC condition, and finding suitable global barriers 
$\phi_-$ and $\phi_+$ for defining the required set $U$ 
that are similarly free of the near-CMC condition (when possible).

We now set up the basic abstract framework we will use.
Let $X$, $Y$, $\mathbb{X}$, and $\mathbb{Y}$ be Banach spaces,
and
let~ $F:X \times Y \to \mathbb{X}$ and~ $G:X \to \mathbb{Y}$
be (generally nonlinear) operators.
Let $A_{\tiIL}:Y \to \mathbb{Y}$ be a linear invertible operator,
and
let $A_{\tiL}:X \to \mathbb{X}$ be a linear invertible operator satisfying
the maximum principle, meaning that 
$A_{\tiL}u \leqs A_{\tiL}v \Rightarrow u \leqs v$.
The order structures on $X$ and $\mathbb{X}$ 
(and hence on their duals, which we denote respectively
as $X^*$ and $\mathbb{X}^*$) for interpreting the maximum 
principle will be inherited from ordered Banach spaces $Z$ and $\mathbb{Z}$
(see the Appendix of~\cite{HNT07b})
through the compact embeddings $X \hookrightarrow Z$
and $\mathbb{X} \hookrightarrow \mathbb{Z}$,
which will also make available compactness arguments.
To formulate our problem in this abstract setting, let $\gamma_{I}$ be the trace operator onto $\Sigma_I$ and $\gamma_{E}$ be the trace operator onto $\Sigma_E$.  As in \cite{HoTs10a} we define the following linear and
nonlinear operators:
\begin{align}
A_L(\phi) = &\left( \begin{array}{c} -\Delta \phi +a_R \phi \\ \gamma_I (\partial_{\nu}\phi) + \frac12H(\gamma_I\phi) \\ \gamma_E (\partial_{\nu} \phi) + c(\gamma_E\phi) \end{array}\right)\label{eq2:12july13}\\
A_{\IL}(\biw) = &\left( \begin{array}{c} \IL\biw\\ \gamma_I((\cL \biw)^{ab}\nu_b) \\ \gamma_E((\cL \biw)^{ab}\nu_b)+ C^a_b(\gamma_E\biw^b)\end{array}\right) \label{eq7:8aug13}\\
F(\phi,\biw) = &\left( \begin{array}{c}a_{\tau}\phi^5-a_{\biw}\phi^{-7}-a_{\rho}\phi^{-3} \\ \left(\frac12\gamma_I(\tau)- \frac14\theta_-\right)(\gamma_I(\phi))^3-\frac14S(\nu,\nu)(\gamma_I(\phi))^{-3} \\ -g \end{array}\right)\label{eq8:8aug13}\\
G(\phi) = &\left( \begin{array}{c} b_{\tau}\phi^6 +\bb_j \\ \bV \\ {\bf 0} \end{array}\right).\label{eq9:8aug13}
\end{align}
For $\phi \in W^{s,p}$ and $\biw \in W^{e,q}$ satisfying the exponent conditions of Theorems~\ref{T:main1}-\ref{T:main2}, we have that 
\begin{align}
A_L&: X = W^{s,p} \to W^{s-2,p}\times W^{s-1-\frac{1}{p},p}(\Sigma_I) \times  W^{s-1-\frac1p,p}(\Sigma_E) = \mathbb{X},\\
A_{\IL}&: Y = W^{e,q} \to \bW^{e-2,q} \times W^{e-1-\frac1q,q}(T\Sigma_I)\times W^{e-1-\frac1q,q}(T\Sigma_E) = \mathbb{Y}, \nonumber\\
F&: X \times Y = W^{s,p} \times \bW^{e,q} \to W^{s-2,p} \times W^{s-1-\frac{1}{p},p}(\Sigma_I) \times  W^{s-1-\frac1p,p}(\Sigma_E) = \mathbb{X},\nonumber\\
G&: X = W^{s,p} \to  \bW^{e-2,q} \times  W^{e-1-\frac1q,q}(T\Sigma_I)\times W^{e-1-\frac1q,q}(T\Sigma_E) = \mathbb{Y}.\nonumber
\end{align}
Then following the discussion in~\cite{HNT07b},
the coupled Hamiltonian and momentum constraints with boundary conditions
\eqref{eq5:11july13} can be viewed 
abstractly as coupled operator equations of the form:
\begin{eqnarray}
A_L(\phi) + F(\phi,w) & = & 0, 
\label{E:ham-abstract} \\
A_{\IL}( \bw) + G(\phi) & = & 0,
\label{E:mom-abstract}
\end{eqnarray}
or equivalently as the coupled fixed-point equations
\begin{eqnarray}
\phi & = & T(\phi,w),
\label{E:ham-abstract-fixpt} \\
w & = & S(\phi),
\label{E:mom-abstract-fixpt}
\end{eqnarray}
for appropriately defined fixed-point
maps $T : X \times Y \to X$ and $S : X \to Y$.
The obvious choice for $S$ is the
{\em Picard map} for~(\ref{E:mom-abstract})
\begin{equation}
  \label{E:mom-abstract-picard}
S(\phi) = -A_{\IL}^{-1} G(\phi),
\end{equation}
which also happens to be the solution map for~(\ref{E:mom-abstract}).
On the other hand, there are a number of distinct possibilities
for $T$, ranging from the solution map for~(\ref{E:ham-abstract}),
to the {\em Picard map} for~(\ref{E:ham-abstract}),
which inverts only the linear part of the operator in~(\ref{E:ham-abstract}):
\begin{equation}
  \label{E:ham-abstract-picard}
T(\phi,w) = -A_L^{-1}F(\phi,w).
\end{equation}

Assume now that $T$ is as in~(\ref{E:ham-abstract-picard}),
and (for fixed $w \in Y$) that
$\phi_-$ and $\phi_+$ are sub- and super-solutions
of the semi-linear operator equation~(\ref{E:ham-abstract})
in the sense that
$$
A_L(\phi_-) + F(\phi_-,w) \leqs 0,
\quad \quad
A_L(\phi_+ )+ F(\phi_+,w) \geqs 0.
$$
The linear operator $A_{\tiL}$ is invertible and satisfies the maximum principle,
which we will show in Section~\ref{sec:hamiltonian}.  These conditions
imply (see~\cite{HNT07b})
that for fixed $w \in Y$, 
$\phi_-$ and $\phi_+$ are also sub- and super-solutions of 
the equivalent fixed-point equation:
$$
\phi_- \leqs T(\phi_-,w),
\quad \quad
\phi_+ \geqs T(\phi_+,w).
$$
For developing results on fixed-point iterations in ordered
Banach spaces, it is convenient to work with maps which are
monotone increasing in $\phi$, for fixed $w \in Y$:
$$
\phi_1 \leqs \phi_2 \quad \Longrightarrow \quad T(\phi_1,w) \leqs T(\phi_2,w).
$$
The map $T$ that arises as the Picard map for a semi-linear 
problem will generally not be monotone increasing; 
however, if there exists a continuous, linear, monotone increasing
map $J : X \rightarrow \mathbb{X}$,
then one can always introduce a positive shift $s$ into the 
operator equation
$$
A_{\tiL}^s(\phi) + F^s(\phi,w) = 0,
$$
with $A_L^s = A_L + sJ$ and $F^s(\phi,w) = F(\phi,w) - sJ \phi$.
Since $s > 0$ the shifted operator $A_{\tiL}^s$ 
retains the maximum principle property of $A_{\tiL}$, and if $s$ is 
chosen sufficiently large, then $F^s$ is monotone decreasing in 
$\phi$. 
Under the additional condition on $J$ and $s$ that $A_{\tiL}^s$ is 
invertible, the shifted Picard map 
$$
T^s(\phi,w) = -(A_{\tiL}^s)^{-1} F^s(\phi,w)
$$
is now monotone increasing in $\phi$.  See Section~\ref{sec:hamiltonian} for
verification of these properties of $T^s$.
\medskip

We now give the main abstract existence result from \cite{HNT07b}
for systems of the 
form~(\ref{E:ham-abstract-fixpt})--(\ref{E:mom-abstract-fixpt}).
\begin{theorem}
{\bf (Coupled Fixed-Point Principle \cite{HNT07b})}
\label{T:FIXPT2}
Let $X$ and $Y$ be Banach spaces, 
and let $Z$ be a real ordered Banach space
having the compact embedding $X \hookrightarrow Z$.
Let $[\phi_-,\phi_+] \subset Z$ be a nonempty interval which
is closed in the topology of $Z$,
and set $U = [\phi_-,\phi_+] \cap \overline{B}_M \subset Z$
where $\overline{B}_M$ is the closed ball
of finite radius $M>0$ in $Z$ about the origin.
Assume $U$ is nonempty, and let the maps
\[
S:U \to \mathcal{R}(S) \subset Y,
\quad
\quad
T:U \times \mathcal{R}(S) \to U \cap X,
\]
be continuous maps.
Then there exist 
$\phi\in U \cap X$
and 
$w\in\mathcal{R}(S)$
such that
\[
\phi=T(\phi,w)\quad\textrm{and}\quad w=S(\phi).
\]
\end{theorem} 
\begin{proof}
See \cite{HNT07b}.
\end{proof}

\begin{remark}
We make some brief remarks about Theorem~\ref{T:FIXPT2}
(see also the discussion following this results in~\cite{HNT07b}).
Theorem~\ref{T:FIXPT2} was specifically engineered for the analysis of 
the fully coupled Einstein constraint equations; it allows one to 
establish simple sufficient conditions on the map $T$ to yield the core 
invariance property by using barriers in an ordered Banach space 
(for a review of ordered Banach spaces, see the Appendix of~\cite{HNT07b}).
If the ordered Banach space $Z$ in Theorem~\ref{T:FIXPT2}
had a {\em normal} order cone, then the closed interval
$[\phi_-,\phi_+]$ would automatically be bounded in the norm
of $Z$ (see the Appendix of~\cite{HNT07b} for this result).
The interval by itself is also non-empty and closed by assumption,
and trivially convex (see the Appendix of~\cite{HNT07b}), so that 
Theorem~\ref{T:FIXPT2} would follow immediately from a variation
of the Schauder Theorem by simply taking $U = [\phi_-,\phi_+]$.
Note that the closed ball $\overline{B}_M$ in Theorem~\ref{T:FIXPT2}
can be replaced with any non-empty, convex, closed, and bounded subset 
of $Z$ having non-trivial intersection with the interval $[\phi_-,\phi_+]$.
\end{remark}

Following our approach in~\cite{HNT07b}, the overall argument used here to 
prove the non-CMC results in Theorems~\ref{T:main2} and \ref{T:main1}
using Theorem~\ref{T:FIXPT2} involves the following steps: 
\begin{itemize}
\item[{\em Step 1:}]{\em The choice of function spaces}.
      We will choose the spaces for use of Theorem~\ref{T:FIXPT2} as follows:
\begin{itemize}
\item $X=W^{s,p}$, with $p\in(3,\frac{\alpha+1}3),~\alpha > 8$, and $s(p)\in(1+\frac{3}{p},2)$.
      
\item $Y=\biW^{e,q}$, 
      with
      $e$ and $q$ as given in the theorem statements.
\item $Z=W^{\tilde{s},p}$, $\tilde{s} \in (1+\frac{3}{p}-\frac4{\alpha},1+\frac3p)$,
      so that $X=W^{s,p}\hookrightarrow W^{\tilde{s},p} = Z$ is compact.
\item $U=[\phi_-,\phi_+]_{\tilde{s},p} \cap \overline{B}_M
      \subset W^{\tilde{s},p} = Z$,
      with $\phi_-$ and $\phi_+$ global barriers
      (sub- and super-solutions, respectively) for
      the Hamiltonian constraint equation which satisfy the compatibility
      condition: $0 < \phi_- \leqs \phi_+ < \infty$.
\end{itemize}
\item[{\em Step 2:}]{\em Construction of the mapping $S$}.
      Assuming the existence of ``global'' weak sub- and super-solutions
      $\phi_-$ and $\phi_+$, and assuming the fixed function
      $\phi \in U = [\phi_-,\phi_+]_{\tilde{s},p} \cap \overline{B}_M
      \subset W^{\tilde{s},p} = Z$ is taken as data in the
      momentum constraint, we establish continuity and related properties
      of the momentum constraint solution map
      $S : U \to \mathcal{R}(S) \subset \biW^{e,q} = Y$.
      (\S\ref{sec:momentum})
\item[{\em Step 3:}]{\em Construction of the mapping $T$}.
      Again existence of ``global'' weak sub- and super-solutions
      $\phi_-$ and $\phi_+$, 
      with fixed $w \in \mathcal{R}(S) \subset \biW^{e,q} = Y$
      taken as data in the Hamiltonian constraint, we establish
      continuity and related properties of the Picard map
      $T: U \times \mathcal{R}(S) \to U \cap W^{s,p}$.
      Invariance of $T$ on $U=[\phi_-,\phi_+]_{\tilde{s},p} 
          \cap \overline{B}_M \subset W^{\tilde{s},p}$
      is established using a combination of {\em a priori} order cone bounds
      and norm bounds.
      (\S\ref{sec:hamiltonian})
\item[{\em Step 4:}]{\em Barrier construction}.
      Global weak sub- and super-solutions $\phi_-$ and $\phi_+$
      for the Hamiltonian constraint are explicitly constructed to
      build a nonempty, convex, closed, and bounded
      subset $U=[\phi_-,\phi_+]_{\tilde{s},p} \cap \overline{B}_M \subset W^{\tilde{s},p}$, 
      which is a strictly positive interval.
      These include variations of known barrier constructions which 
      require the near-CMC condition, and also some new barrier
      constructions which are free of the near-CMC condition.
      (\S\ref{sec:barriers})
      {\bf\em Note: This is the only place in the argument where
      near-CMC conditions may potentially arise.}
\item[{\em Step 5:}]{\em Application of fixed-point theorem}.
      The global barriers and continuity properties
      are used together with the abstract topological fixed-point
      result (Theorem~\ref{T:FIXPT2}) to 
      establish existence
      of solutions $\phi \in U \cap W^{s,p}$ and $w \in \biW^{e,q}$
      to the coupled system: $w=S(\phi), \phi=T(\phi,w).$ (\S\ref{sec:proof})
\item[{\em Step 6:}]{\em Bootstrap}.
      The above application of a fixed-point theorem is actually performed for some low regularity spaces,
      i.e., for $s\leqs2$ and $e\leqs2$
,      and a bootstrap argument is then given to extend the results to
      the range of $s$ and $p$ given in the statement of the Theorem.
      (\S\ref{sec:proof})
\end{itemize}

As was the case in~\cite{HNT07b,HoTs10a},
the ordered Banach space $Z$ plays a central role in Theorem~\ref{T:FIXPT2}
and its application here.
We will use $Z=W^{t,q}$, ~$t \geqs 0$, ~$1 \leqs q \leqs \infty$, 
with order cone defined as in~(\ref{E:wkp-cone}).
Given such an order cone, one can define the closed interval
$$
[\phi_-,\phi_+]_{t,q}
  = \{ \phi \in W^{t,q} : \phi_- \leqs \phi \leqs \phi_+ \} \subset W^{t,q},
$$
which as noted earlier is denoted more simply as
$[\phi_-,\phi_+]_q$ when $t=0$,
and as simply
$[\phi_-,\phi_+]$ when $t=0$, $q=\infty$.
If we consider the interval $U=[\phi_-,\phi_+]_{t,q} \subset W^{t,q} = Z$
defined using this order structure, for use with Theorem~\ref{T:FIXPT2}
it is important that $U$ be 
convex (with respect to the vector space structure of $Z$),
closed (in the topology of $Z$),
and (when possible) bounded (in the metric given by the norm on $Z$).
It will also be important that $U$ be nonempty as a subset of $Z$;
this will involve choosing compatible $\phi_-$ and $\phi_+$.
Regarding convexity, closure, and boundedness, we have the
following lemma from~\cite{HNT07b}.
\begin{lemma}
{\bf (Order cone intervals in $W^{t,q}$ \cite{HNT07b})}
\label{L:wsp-interval}
For $t \geqs 0$, $1 \leqs q \leqs \infty$, the set
$$
U = [\phi_-,\phi_+]_{t,q}
  = \{ \phi \in W^{t,q} : \phi_- \leqs \phi \leqs \phi_+ \} \subset W^{t,q}
$$
is convex with respect to the vector space structure of \ $W^{t,q}$
and closed in the topology of \ $W^{t,q}$.
For $t=0$, $1 \leqs q \leqs \infty$, the set $U$ is also bounded with 
respect to the metric space structure of $L^q=W^{0,q}$.
\end{lemma}
\begin{proof}
See~\cite{HNT07b}.
\end{proof}

\section{Momentum Constraint}
\label{sec:momentum}
In this section we fix a particular scalar function $\phi\in
W^{{s},p}$ with ${s}p>3$, and consider separately the momentum constraint
equation~(\ref{eq2:8aug13}) with boundary conditions \eqref{eq4:8aug13}-\eqref{eq5:8aug13} 
to be solved for the vector valued function $\biw$. The result is a linear elliptic system of
equations for this variable $\biw=\biw_{\phi}$.  
Our goal is not only to develop some existence results for the momentum
constraint, but also to derive the estimates for the momentum constraint
solution map $S$ that we will need later in our analysis of the coupled
system.

Let $(\cM, h)$ be a 3-dimensional Riemannian manifold, where $\cM$
is a smooth, compact manifold with boundary satisfying \eqref{eq6:11july13} with $p\in(1,\infty)$, $s\in(1+\frac3p,\infty)$, and $h\in W^{s,p}$ is a
positive definite metric. With 
\begin{align}\label{eq1:13aug13}
\textstyle
q\in(1,\infty),
\qquad\textrm{and}\qquad
e\in(2-s,s]\cap(-s+\frac3p-1+\frac3q,s-\frac3p+\frac3q], 
\end{align}
fix the source terms 
\begin{align}\label{eq2:13aug13}
&\bib_{\tau}, \bib_{j} \in \biW^{e-2,q}, \bV \in \bW^{e-1,q}, ~~~
\text{and} ~~~ C^{a}_b \in W^{e-1-\frac1q,q}(T^1_1\Sigma_E)\cap L^{\infty}(T^1_1\Sigma_E),
\end{align}
where $C^{a}_b$ satisfies \eqref{eq1:12july13}.
Fix a function $\phi\in W^{{s},p}$, and let
\begin{align*}
&A_{\IL}: W^{e,q} \to \bW^{e-2,q} \times W^{e-1-\frac1q,q}(T\Sigma_I)\times W^{e-1-\frac1q,q}(T\Sigma_E), \nonumber\\
&G:W^{s,p} \to  \bW^{e-2,q} \times  W^{e-1-\frac1q,q}(T\Sigma_I)\times W^{e-1-\frac1q,q}(T\Sigma_E).\nonumber
\end{align*}
be as in \eqref{eq7:8aug13} and \eqref{eq9:8aug13}.

The momentum constraint equation with Robin boundary conditions is the
following: find an element $\biw\in\biW^{e,q}$ that is a solution of
\begin{equation}
\label{MC-LYm1}
A_{\IL}\biw + G(\phi) = 0.
\end{equation}

\subsection{Weak Formulation}

In order to show that~\eqref{MC-LYm1} has a solution, we employ the Lax-Milgram Theorem in the case
when $p\in (1,\infty)$, $s>1+ \frac3p$, $e=1$ and $q=2$ to show that the weak formulation has a solution.  We will then utilize {\em a priori} estimates
to show that solutions exist for the exponent ranges specified above.

Using the volume form given by $h$ and integration by parts, the weak formulation of \eqref{MC-LYm1} is then to
find $\biw \in {\bf W}^{1,2}$ such that for all $\biv \in {\bf W}^{1,2}$,
\begin{align}\label{eq1:16july13}
\int_{\cM}(\cL \biw)_{ab}(\cL \biv)^{ab}~dx +  \int_{\Sigma_E} C^a_b\gamma_E\biw^b\gamma_E\biv_a~ds &= -\int_{\cM} (b^a_{\tau}\phi^6+b^a_j)\biv_a~dx \nonumber \\
& \quad + \int_{\Sigma_I} \gamma_I\bV^a\gamma_I\biv_a~ds, 
\end{align}
where $dx$ is the measure induced by $h$ and $ds$ is the measure induced by the metric on $\partial\cM$ that is inherited from $h$.  

\begin{remark}\label{rem1:16july13}
We observe that the bilinear form \eqref{eq1:16july13} is well defined for $\biv \in \bW^{1,2}$, given that $\gamma_i\biv \in W^{\frac12,2}(T\Sigma_i)$, $i\in \{I,E\}$,
and $\bV, (\bb_{\tau}\phi^6+\bb_j) \in {\bf L}^2$
and $C^a_b \in L^{\infty}(T^1_1\Sigma_I)$.  
\end{remark}

Letting 
\begin{align}\label{eq2:15july13}
a_{\cL}(\biw, \biv) = \int_{\cM}(\cL \biw)_{ab}(\cL \biv)^{ab}~dx +  \int_{\Sigma_E} C^a_b\gamma_E\biw^b\gamma_E\biv_a~ds, 
\end{align}
and
\begin{align}\label{eq8:15july13}
\bif(\biv) = -\int_{\cM} (b^a_{\tau}\phi^6+b^a_j)\biv_a~dx + \int_{\Sigma_I} \gamma_I\bV^a\gamma_I\biv_a~ds,
\end{align}
we say that $A_{\IL}(\biw) +G(\phi) = 0~$ weakly if $~a_{\cL}(\biw,\biv) = f(\biv)~$ for all $\biv \in \bW^{1,2}$.  

Our approach to proving that~\eqref{MC-LYm1} is weakly solvable will be to verify that the 
shifted, bounded linear operator 
\begin{align}\label{eq1:15july13}
a^s_{\cL}(\biw,\biv) = a_{\cL}(\biu,\biv)+s (\biw,\biv),
\end{align}
is coercive for some $s>0$.  We will then apply the Lax-Milgram Theorem and Riesz-Schauder Theory 
to conclude that \eqref{MC-LYm1} has a unique, weak solution in $\bW^{1,2}$. 

\subsubsection{G\aa rding's Inequality}

The primary inequality that we will need to establish in order to show that \eqref{eq1:15july13} is coercive is 
the G\aa rding inequality.  We just mention here that the G\aa rding type inequality for the
particular case of the space $W^{1,2}_0$ can be proven for a
general class of bilinear forms called strongly
elliptic. See~\cite{Zeidler-IIA}, exercise 22.7b, page 396. A bilinear
form $a: W_0^{1,2}\times W_0^{1,2}\to\R$ with action
\begin{align*}
a(u,v)&= \int_{\cM} a_{ac_1\cdots c_nbd_1\cdots d_n}
\nabla^a u^{c_1\cdots c_n} \nabla^b v^{d_1\cdots d_n} \,dx\\
&\quad +\int_{\cM} b_{c_1\cdots c_nd_1\cdots d_n}
u^{c_1\cdots c_n} v^{d_1\cdots d_n} \,dx
\end{align*}
is {\bf strongly elliptic} iff there exists a positive constant
$\alpha_0$ such that
\[
a_{ac_1\cdots c_nbd_1\cdots d_n} \zeta^a\zeta^b 
u^{c_1\cdots c_n} u^{d_1\cdots d_n} \geqs\alpha_0\,
\zeta_a\zeta^a\, u_{c_1\cdots c_n} u^{c_1\cdots c_n}
\]
for all vectors $\zeta \in \R^3$ and all tensors $u_{c_1\cdots
c_n}\in\R^{3n}$.  Notice that the bilinear form
$a_{\cL} :\biW^{1,2}_{0}\times\biW^{1,2}_0\to\R$ given by
$a_{\cL}(\biu,\biv) = (\cL\biu,\cL\biv)$ is strongly elliptic, as the
following calculation shows:
\begin{gather*}
\bigl[\zeta^a u^c + \zeta^c u^a 
-\frac{2}{3} h^{ac} (\zeta_d u^d)\bigr]
\bigl[\zeta_a u_c + \zeta_c u_a 
-\frac{2}{3} h_{ac} (\zeta_e u^e)\bigr] 
\\
=2 (\zeta_a\zeta^a)(u_bu^b) 
+ \frac{2}{3}\, (\zeta_au^a)^2 
\geqs 2 (\zeta_a\zeta^a)(u_bu^b). 
\end{gather*}
Hence, a G\aa rding type inequality is satisfied by the bilinear form
$a_{\cL}$ on the Hilbert space $\biW^{1,2}_0$. However, this space is
too small in our case where we need the same inequality on the space
$\biW^{1,2}$. 

We extend the G\aa rding inequality to the space $\biW^{1,2}$ in the following
two results.
\begin{lemma}
\label{L:GKI}{\bf (G\aa rding's inequality for $\cL$)}
Let $(\cM,h_{ab})$ be a 3-dimensional, compact, Riemannian manifold,
with Lipschitz boundary, and with a metric $h\in W^{s,p}, p \in (1,\infty), s \in (1+\frac3p,
\infty)$. 
Then, there exists a positive constant $k_0$
such that the following inequality holds
\begin{equation}
\label{GKI}
k_0 \,\|\tbu\|^2_{1,2} \leqs \|\tbu\|_2^2 + \|\cL\tbu\|_2^2
\qquad \forall \tbu \in \tbW^{1,2}.
\end{equation}
\end{lemma}

\Proof {\it (Lemma~\ref{L:GKI}.)~}
See~\cite{sD06} for the proof.
\qed

\medskip

Using Lemma~\ref{L:GKI}, we can immediately 
establish the same type of inequality for 
the bilinear form $a_{\cL}(\biu,\biu)$ in~\eqref{eq2:15july13} provided that
$C^a_b$ is positive definite in the sense of~\eqref{eq1:12july13}.
\begin{corollary}
\label{C:GI}{\bf(G\aa rding's inequality for $a_{\cL}(\biw,\biv)$)}
Let $(\cM,h_{ab})$ be a 3-dimensional, compact, Riemannian manifold,
with Lipschitz boundary and with a metric $h\in W^{s,p},\\
 p \in (1,\infty), s \in (1+\frac3p,\infty)$. 
Let $a_{\cL}(\biu,\biu)$ be the bilinear form defined 
in~\eqref{eq2:15july13} for a positive definite tensor $C^a_b\in
{\bf L}^{\infty}(T^1_1(\Sigma_E))$ in the sense of~\eqref{eq1:12july13}.
Then, there exists a positive
constant $k_1$ such that the following inequality holds
\begin{equation}
\label{GI}
k_1 \,\|\biu\|^2_{1,2} \leqs \|\biu\|_2^2 + a_{\cL}(\biu,\biu)
\qquad \forall \biw \in \bW^{1,2}.
\end{equation}
\end{corollary}

\Proof {\it (Corollary~\ref{C:GI}.)~}
The definition of the bilinear form in~\eqref{eq2:15july13} implies that
\[
a_{\cL}(\biu,\biu) = \|\cL\biu\|_2^2
+ \langle C \gamma_E \biu,\gamma_E\biu\rangle_{\Sigma_I},
\quad \forall \biu\in \biW^{1,2}.
\]
Lemma~\ref{L:GKI} and the fact that
$C^a_b$ is positive definite imply the result.
\qed

\medskip

The above results combined with Riesz-Schauder theory allow us to conclude 
that~\eqref{eq2:15july13} is weakly solvable in 
Theorem~\ref{thm1:15july13} below.
We note that while the positivity assumption~\eqref{eq1:12july13} used
in Corollary~\ref{C:GI} can be removed by using a more complex proof involving
a trace inequality, the positivity assumption~\eqref{eq1:12july13}
is essential to showing injectivity in Theorem~\ref{thm1:15july13} below.
However, the positivity property is available in the practical
situations of interest such as~\eqref{eq2:11july13}.

\begin{theorem}{\bf (Momentum constraint)}
\label{thm1:15july13}
Suppose $(\cM,h)$ is a connected, $3$-dimensional manifold with boundary satisfying \eqref{eq6:11july13} and with $h\in W^{s,p}$, $p \in (1,\infty)$, $s>1+\frac3p$.  
Assume that the data $\bb_{\tau}, \bb_j \in {\bf W}^{-1,2}$, $\bV \in {\bf L}^{2}$, $\sigma\in L^2(T^0_2\cM)$, and let the tensor $C^a_b \in L^{\infty}(T^1_1(\Sigma_i))$ be positive definite in the sense of~\eqref{eq1:12july13}.
Then there exists a unique solution to the weak formulation of the momentum constraint \eqref{eq2:15july13},
and there exists a constant $C>0$ independent $\tau, \bj, \bV$ and $ \sigma$ such that the following estimate holds:
\begin{align}\label{eq3:15july13}
\|\bw\|_{1,2} \leqs & C \left( \|\bb_{\tau}\phi^6\|_{-1,2} +\|\bb_j\|_{-1,2} + \|\Tr_{\tiI}\bV \|_{-\frac12,2;\Sigma_I}\right). 
\end{align}
\end{theorem}

\begin{proof}
Setting $s >0$, Corollary~\ref{C:GI} implies that the bilinear form
\eqref{eq1:15july13} is coercive.
By the Lax-Milgram Theorem, for any $\bih \in {\bf W}^{-1,2}$, we have that 
there exists a unique element $\biw \in {\bf W}^{1,2}$ which satisfies
$$
a^s_{\cL}(\biw,\biv) = a_{\cL}(\biw,\biv)+(\biw,\biv) = \bih(\biv).
$$
This defines a bounded, invertible operator $L^s:\bW^{1,2} \to \bW^{-1,2}$,
where
\begin{align}\label{eq4:15july13}
L^s\biw = \bih \Longleftrightarrow a^s_{\cL}(\biw,\biv) = h(\biv),\quad \text{for all}~~ \biv \in \bW^{1,2}.
\end{align} 
If we let $L$ be a similar operator defined by
\begin{align}\label{eq5:15july13}
L\biw = \bih  \Longleftrightarrow a_{\cL}(\biw,\biv) = \bih(\biv) ,\quad \text{for all}~~ \biv \in \bW^{1,2},
\end{align}
we have that 
\begin{align}\label{eq6:15july13}
&L+ sI = L^s, ~~\text{and}\\
L\biw &= \bih \Longleftrightarrow L^s\biw = \bih + s \biw.\nonumber
\end{align}
Therefore, rewriting~\eqref{eq6:15july13}, we have that
\begin{align}
L\biw &= \bih \Longleftrightarrow \biw-s (L^s)^{-1}\biw = (L^{s})^{-1}\bih.
\end{align}
Standard elliptic PDE theory tells us that the operator
$$
K\biw = s (L^{s})^{-1}\biw,
$$
is compact, and we can therefore apply the Fredholm alternative to conclude that
the operator $L$ is Fredholm with index zero.  Thus, $dim(ker(L)) = codim(R(L))$ and
to conclude that the operator $L$ is invertible (which implies the existence and uniqueness of solutions to the
weak formulation), we need
only show that its kernel is trivial. 

Assume that $L$ has a nontrivial kernel.  This implies that there exists some $\biw \in \bW^{1,2}$ such
that
\begin{align}
a_{\cL}(\biw,\biv) =  \int_{\cM}(\cL \biw)_{ab}(\cL \biv)^{ab}~dx +\int_{\Sigma_E} C^a_b\gamma_E\biw^b\gamma_E\biv_a~ds = 0,
\end{align}
for all $\biv \in \bW^{1,2}$.
Therefore, 
\begin{align}\label{eq7:15july13}
 0 &\leqs  C\|\nabla \biw\|_{2}^2 \leqs  \int_{\cM}(\cL \biw)_{ab}(\cL \biv)^{ab}~dx =-  \int_{\Sigma_E} C^a_b\gamma_E\biw^b\gamma_E\biv_a~ds \nonumber \\
& \leqs  -\alpha \| \gamma_E\biw\|^2_{2;\Sigma_E} \leqs  0,
\end{align}
where $\alpha > 0$ by the positive definite assumption on $C^a_b$.
If $\biw \ne 0$, then one of the above two inequalities must be strict.  In particular,
if $\|\nabla\biw\|_{2} = 0$, then $\biw$ is constant and the assumption that $\biw\ne 0$ implies that $\|\biw\|_{L^2;\Sigma_E}>0$.  On the other hand,
if $\|\biw\|_{2;\Sigma_E}=0$ and $\biw \ne 0$, then $\biw$ is non-constant and $\|\nabla\biw\|_{2} >0$.  In either case, we have a contradiction which
allows us to conclude that $L$ has a trivial kernel and is invertible.  
Given that $\bif\in \bW^{-1,2}$, where $\bif$ is defined in \eqref{eq8:15july13},
the weak formulation momentum constraint \eqref{MC-LYm1} has a solution.

In order to establish the {\em a priori} estimate \eqref{eq3:15july13}, we apply the open mapping theorem to conclude that $L$ is open.
Given that $L$ is invertible, we can then conclude that $L^{-1}:\bW^{-1,2} \to \bW^{1,2}$ is a bounded linear operator.  So there exists some
$C>0$ such that for any $\bih \in \bW^{-1,2}$,
\begin{align}
\|L^{-1}\bih \|_{1,2} \leqs  C\|\bih\|_{-1,2}.
\end{align}
This implies that if $\biw$ is our unique, weak solution to the momentum
constraint for ${\bif \in {\bf W}^{-1,2}}$ given by \eqref{eq8:15july13}, then
\begin{align}\label{eq4:16july13}
\|\biw\|_{1,2} \leqs  C\|\bif\|_{-1,2}.
\end{align}
The above bound \eqref{eq4:16july13} implies that 
\begin{align}
\|\biw\|_{1,2}\leqs  C\left(\|\bb_{\tau}\phi^6+\bb_j\|_{-1,2} + \|\Tr_{\tiI}\bV\|_{-\frac12,2;\Sigma_I}\right),
\end{align}
which is the desired estimate in~\eqref{eq3:15july13}.
\end{proof}

Now that we have shown that weak solutions $\biw \in\bW^{1,2}$ exist,  we utilize the following regularity theorem
to show that the Momentum constraint~\ref{MC-LYm1} has a solution in $\bW^{e,q}$, with $e,q$ satisfying \eqref{eq1:13aug13},
provided that the Robin data and coefficients functions satisfy \eqref{eq2:13aug13}.

\begin{theorem}{\bf (Regularity $\tbW^{e,q}$)}
\label{T:MC-E-Reg2}
Let $(\cM,h)$ be a connected, $3$-dimensional compact manifold with boundary satisfying \eqref{eq6:11july13}, and with 
${h\in W^{s,p}}$, ${p \in (1,\infty)}$, ${s\in (1+\frac3p,\infty)}$.
Let $\partial\cM$ be $C^{k}$, where ${k \geqs e > 1}$, and suppose that
the Robin data and coefficients satisfy the regularity 
assumptions \eqref{eq1:13aug13}. 
Then there exists a solution $\biw$ to the momentum constraint
~\eqref{MC-LYm1} in $\tbW^{e,q}$, and there exist
positive constants $C_1$ and $C_2$ such that the following estimate
holds, 
\begin{align}\label{eq3:18july13}
\|\biw\|_{e,q}\leqs & C\left(\|\bb_{\tau}\phi^6\|_{e-2,q} +\|\bb_j\|_{e-2,q} + \|\Tr_{\tiI}\bV\|_{e-1-\frac1q,q;\Sigma_I}\right). 
\end{align}
\end{theorem}

\Proof {\it (Theorem~\ref{T:MC-E-Reg2}.)~}
We give just an outline of a proof following a standard approach.
Viewing the metric in local coordinates and applying interior and boundary 
estimates, and then applying a partition of unity argument, one obtains the 
above result.  In particular, one decomposes both the vector Laplacian $\IL$ 
and the boundary operators into a sum of constant coefficient and slightly 
perturbed non-constant coefficient operators as in 
Proposition 5.1 in \cite{DM05a} and Lemma B.3 in \cite{HoTs10a}.
One then applies results for constant coefficient elliptic operators, 
interpolation inequalities, and \eqref{eq3:15july13} to obtain the result.
\qed

\begin{remark}\label{rem1:18july13}
If the data $\bV, \tbb_{\tau}, \tbb_j$ satisfies the hypotheses of Theorem~\ref{T:MC-E-Reg2},
and $q \in (3,\infty)$, $1+\frac3q < e \leqs 2 $, 
then $\bW^{e,q} \hookrightarrow \bW^{1,\infty}$ and we have that
\begin{align*}
\|\cL\biw\|_{\infty} \leqs  \|\biw\|_{1,\infty} \leqs  C\|\biw\|_{e,q}.
\end{align*}
Combining this and the a priori estimate \eqref{eq3:18july13} and using the embedding
$L^z \hookrightarrow W^{e-2,q}$ for appropriate $z$, we have that 
\begin{align}\label{eq2:18july13}
\|\cL\biw\|_{\infty}\leqs  C\left(\|\phi\|^6_{\infty}\|\bb_{\tau}\|_{z} +\|\bb_j\|_{e-2,q} + \|\Tr_{\tiI}\bV\|_{e-1-\frac1q,q;\Sigma_I}\right).
\end{align}
Furthermore, if $\bX \in W^{e-1-\frac1q,q}(T\cM)$ and $\bV$ is a vector field such that
\begin{align*}
\Tr_{\tiI}\bV  = \left((2\tau+|\theta_-|/2)B^6-\sigma(\nu,\nu)\right)\nu + \bX,
\end{align*}
we can utilize \eqref{eq2:18july13} to obtain
\begin{align}\label{eq5:27sep13}
\|\cL\biw\|_{\infty} 
   \leqs & C\left(\|\phi_+\|^6_{\infty}\|\bb_{\tau}\|_{z} +\|\bb_j\|_{e-2,q} + \|\tau B^6\|_{e-1-\frac{1}{q},q;\Sigma_I} \right. \nonumber \\
 & \quad \left. +\|\theta_-B^6\|_{e-1-\frac{1}{q},q;\Sigma_I} +\|\sigma(\nu,\nu)\|_{e-1-\frac{1}{q},q;\Sigma_I}+\|\bX\|_{e-1-\frac{1}{q},q;\Sigma_I}\right).
\end{align}
Note that we have replaced $\phi$ with $\phi_+$
in \eqref{eq2:18july13} given that we are assuming $\phi_+$ is an {\em a priori} upper bound on $\phi$.

The bounds in equation~\eqref{eq5:27sep13}
will be essential to control $a_{\biw}$ in the Hamiltonian
constraint in terms of the global super-solution $\phi_+$.  This will be necessary in order to obtain our
global sub-and super-solutions later.
\end{remark}

Theorems~\ref{thm1:15july13} and \ref{T:MC-E-Reg2} imply that the 
Picard map (in this case the solution map),
\begin{align*}
&S: W^{s,p} \to \bW^{e,q},\\
&S(\phi) = -A_{\IL}^{-1}G(\phi),
\end{align*}
is well-defined.  In order to apply the Coupled Fixed
Point Theorem~\ref{T:FIXPT2}, we will additionally require
that $S$ be continuous, which we show in the following Lemma.

\begin{lemma}
{\bf (Properties of the map $S$)}
\label{T:MC-E-Lip1}
In addition to the conditions \eqref{eq1:13aug13} and \eqref{eq2:13aug13} imposed in the beginning of this section, let 
$e \in [0,2]$ and $\bib_{\tau} \in L^z$ with $z = \frac{3q}{\max\{0,(2-e)\}q+3}$.
Let the assumptions for Theorems~\ref{thm1:15july13} and \ref{T:MC-E-Reg2}, 
so that in particular the momentum constraint \eqref{MC-LYm1} is uniquely
solvable in $\tbW^{e,q}$.
With some $\phi_+\in W^{s,p}$ satisfying $\phi_{+}>0$, let $\tbw_1$ and $\tbw_2$ be the solutions
to the momentum constraint with the source
functions $\phi_1$ and $\phi_2$ from the set $[0,\phi_+]\cap W^{s,p}$, respectively.
Then,
\begin{equation}
\|\tbw_1-\tbw_2\|_{e,q}
\leqs
C\, \|\phi_+\|_{\infty}^5
\|\tbb_{\tau}\|_{z}\,
\|\phi_1-\phi_2\|_{s,p}.
\end{equation}
\end{lemma}

\begin{proof}
The functions $\phi_1$ and $\phi_2$ pointwise satisfy the following inequalities
\begin{equation}\label{HC-phin}
\begin{split}
\phi_2^n-\phi_1^n &=
\Bigl(\sum_{j=0}^{n-1} \phi_2^j \phi_1^{n-1-j} \Bigr)
(\phi_2- \phi_1)
\leqs n\, (\phi_{+})^{n-1} \, |\phi_2- \phi_1|,\\
-\bigl[ \phi_2^{-n} -\phi_1^{-n} \bigr] &=
\frac{\phi_2^n-\phi_1^n}{(\phi_2\phi_1)^n}
\leqs n\, \frac{(\phi_{+})^{n-1}}{(\phi_{-})^{2n}} \, |\phi_2- \phi_1|,
\end{split}
\end{equation}
for any integer $n>0$.

By Theorems~\ref{thm1:15july13} and \ref{T:MC-E-Reg2}, 
for a fixed $\phi \in (0,\phi^+)$, $S^{-1}$ is an invertible operator between
$$
\mathbb{Y} = W^{e-2,q}(T\cM)\times W^{e-1-\frac1q,q}(T\Sigma_I)\times W^{e-1-\frac1q,q}(T\Sigma_E)
$$  
and  $Y=W^{e,q}(T\cM)$.  Hence, by the Bounded Inverse Theorem
$$
\|\biw\|_{e,q} \leqs  C\|G(\phi)\|_{\mathbb{Y}}.
$$
Therefore
\begin{align}\label{eq1:3oct13}
\|\biw_1-\biw_2\|_{e,q} \leqs  C\|G(\phi_1)-G(\phi_2)\|_{\mathbb{Y}} = \|b_{\tau}(\phi_1^6-\phi_2^6)\|_{e-2,q},
\end{align}
given that the boundary terms on $W^{e-1-\frac1q,q}(\Sigma_I)$ and $W^{e-1-\frac1q,q}(\Sigma_E)$ do not depend on $\phi$,
and so the norms corresponding to these terms in the $\mathbb{Y}$ norm vanish.

Using \eqref{eq1:3oct13}, the inequalities \eqref{HC-phin}, and the embeddings
${W^{s,p} \hookrightarrow L^{\infty}}$, $L^z \hookrightarrow W^{e-2,q}$, we obtain 
\begin{align*}
\|\biw_1-\biw_2\|_{e,q}
&\leqs
C\|\bib_{\tau}(\phi_1^6-\phi_2^6)\|_{e-2,q}\,
\leqs
C\|\bib_{\tau}(\phi_1^6-\phi_2^6)\|_z
\leqs
C\|\bib_{\tau}\|_z\|\phi_1^6-\phi_2^6\|_{\infty} \\
&\leqs
6C\|\phi_+\|_\infty^5
\|\bib_{\tau}\|_{z}\,
\|\phi_1-\phi_2\|_{s,p}.
\end{align*}
\end{proof}

\section{The Hamiltonian constraint and the Picard map $T$}
\label{sec:hamiltonian}
In this section we fix a particular function $a_{\biw}$ in an
appropriate space and we then separately look for weak solutions of the
Hamiltonian constraint \eqref{eq4:11july13} with Robin boundary conditions \eqref{eq5:11july13}-\eqref{eq3:8aug13}.
For convenience, we reformulate the problem here in a self-contained manner.
Our goal here is primarily to establish some properties and derive some
estimates for a Hamiltonian constraint fixed-point map $T$ that we will 
need later in our analysis of the coupled system.

Let $(\cM, h)$ be a 3-dimensional Riemannian
compact manifold with
boundary satisfying \eqref{eq6:11july13} and with $p\in(1,\infty)$ and $s\in(\frac3p,\infty)\cap[1,\infty)$, $h\in W^{s,p}$ is a positive definite metric. 
Recall the operators 
\begin{align}\label{eq11:15july13}
A_L(\phi) = &\left( \begin{array}{c} -\Delta \phi +a_R \phi \\ \gamma_I (\partial_{\nu}\phi) + \frac12H(\gamma_I\phi) \\ \gamma_E (\partial_{\nu} \phi) + c(\gamma_E\phi) \end{array}\right),\\
F(\phi,\biw) = &\left( \begin{array}{c}a_{\tau}\phi^5-a_{\biw}\phi^{-7}-a_{\rho}\phi^{-3} \\ \left(\frac12\gamma_I(\tau)- \frac14\theta_-\right)(\gamma_I(\phi))^3-\frac14S(\nu,\nu)(\gamma_I(\phi))^{-3} \\ -g \end{array}\right),
\end{align}
introduced in Section~\ref{sec1:26july13}.  The dependence of $F(\phi, \biw)$ on $\biw$ is hidden in the fact that the coefficient $a_{\biw}$ depends on $\biw$
and $S(\nu,\nu) = \cL\bw(\nu,\nu) + \sigma(\nu,\nu)$, cf. \eqref{CF-def-coeff2}.

Fix the source functions
\begin{align*}
a_\tau, a_\rho, a_{\biw} \in W^{s-2,p}_{+},~~
a_{\tiR}=\frac18R \in W^{s-2,p},\quad \text{and} ~~\\
\theta_-, H \in W^{s-1-\frac1p,p}(\Sigma_I), \quad c,g \in W^{s-1-\frac1p,p}(\Sigma_E),
\end{align*}
where $R$ is the scalar curvature of the metric $h$ and
$H$ is the mean extrinsic curvature on $\Sigma_I$ induced by $h$.
(By Corollary A.5(b) in~\cite{HoTs10a},
we know $h_{ab} \in W^{s,p}$ implies $R \in W^{s-2,p}$ and $H \in W^{s-1-\frac1p,p}(\Sigma_I)$. 
Here the pointwise multiplication by an element of $W^{{s},p}$ defines a bounded linear map in
$W^{s-2,p}$ since $s-2\geqs-s$ and $2(s-\frac3p)>0>2-3$, cf. Corollary A.5(a) in~\cite{HoTs10a}.
Therefore we have that 
\begin{align}
&A_{L}: W^{s,p} \to W^{s-2,p}\times W^{s-1-\frac{1}{p},p}(\Sigma_I) \times  W^{s-1-\frac1p,p}(\Sigma_E),\\
&F: W^{s,p} \times \bW^{e,q} \to W^{s-2,p} \times W^{s-1-\frac{1}{p},p}(\Sigma_I) \times  W^{s-1-\frac1p,p}(\Sigma_E).
\end{align}

\noindent
We then formulate the Hamiltonian constraint equation with Robin boundary conditions as follows: find an element that is a solution of
\begin{equation}
\label{HC-LYs1}
A_L(\phi) + F(\phi,\biw) = 0.
\end{equation}

Recall from Section~\ref{sec1:26july13} that our approach for finding weak solutions
to~\eqref{HC-LYs1} is to reformulate the problem as a fixed point problem of
the form
\begin{align}\label{eq2:29july13}
\phi = (A_L^s)^{-1}F^s(\phi,\biw) = T^s(\phi,\biw),
\end{align}
where we assume that $\biw$ is fixed and $A_L^s$ and $F^s$ are
the shifted operators defined in Section~\ref{sec1:26july13}.  In order for this map to be
well-defined, we obviously require $A_L^s$ to be an invertible map.  Furthermore,
we also will require the map $T^s$ to be monotonically increasing in $\phi$,
which will require $A_L^s$ to satisfy the maximum principle.
These two properties of $A^s_L$ are verified in 
Lemmas~B.7 and B.8 in~\cite{HoTs10a}.

Now that we are sure that the map $T^s$ is well-defined, we
discuss some key properties of this map that are essential in
applying the coupled fixed point Theorem.

\subsection{Invariance of $T^s$ given Global Sub-and Super-Solutions}

To establish existence results for weak solutions to the 
Hamiltonian constraint equation using fixed-point arguments, we must show that
the fixed point operator $T^s$ in \eqref{eq2:29july13} is invariant on a certain subspace.
This will require the existence of generalized (weak) sub- and super-solutions 
(sometimes called barriers) which will be derived later 
in~\S\ref{sec:barriers}.
Let us recall the definition of sub- and super-solutions in the following, 
in a slightly generalized form that will be necessary in our study of the 
coupled system.

A function $\phi_{-}\in (0,\infty)\cap W^{s,p}$ is called a {\bf sub-solution} of \eqref{eq11:15july13} iff the function
$\phi_{-}$ satisfies the inequality
\begin{equation}
\label{WF-Sb} 
A_{\tiL}\phi_{-} + F(\phi_{-},\biw)\leqs0,
\end{equation}
for some $a_{\biw}\in W^{s-2,p}$.
A function $\phi_{+} \in (0,\infty)\cap W^{s,p}$ is called a {\bf super-solution} of \eqref{eq11:15july13} iff the function $\phi_{+}$ satisfies the inequality
\begin{equation}
\label{WF-Sp}
A_{\tiL}\phi_{+} + F(\phi_{+},\biw)\geqs0,
\end{equation}
for some $a_{\biw}\in W^{s-2,p}$.
We say a pair of sub- and super-solutions is {\em compatible} if they satisfy
\begin{equation}
   \label{eqn:compat}
0 < \phi_- \leqs \phi_+ < \infty,
\end{equation}
so that the interval $[\phi_-,\phi_+] \cap W^{s,p}$ 
is both nonempty and bounded. 
In the following discussion, we will assume that $\phi_-$ and
$\phi_+$ are a compatible pair of barriers.

Now the we have discussed the basic properties of the linear
mapping $A_L$, 
we turn to the properties of the fixed-point mapping 
$T^s : U \times \mathcal{R}(S) \to X$ for the Hamiltonian constraint, 
where we define $T^s$ as in \eqref{eq2:29july13}.
In the following, we analyze the behavior of $T^s(\phi)$ for
$\phi_- \leqs  \phi \leqs  \phi_+$.  
For ease of notation, we let
\begin{align}\label{eq1:31july13}
\mathbb{X} = W^{s-2,p}\times W^{s-1-\frac1p,p}(\Sigma_I)\times W^{s-1-\frac1p,p}(\Sigma_E).
\end{align}

\begin{lemma}
{\bf (Properties of the map $T$)}
\label{l:shift}
In the above described setting, assume that $p\in(\frac{\alpha+1}{\alpha-1},\frac{\alpha+1}{3})$ 
for $\alpha>4$ and $s\in(\frac3p,\infty)\cap [1,3-\frac1{p'}]$.
With $\ba = (a,a_I,a_E)  \in \mathbb{X}_{+} $ 
satisfying $a_i\neq0$ and $\psi \in W^{s,p}_+$, let 
$$
\Psi = (\psi,0, 0) \quad \text{and} \quad a_{\tbw} \Psi = (a_{\tbw}\psi,0,0) \in \mathbb{X}.
$$
Then let $\ba_s=\ba+ a_{\tbw}\Psi\in \mathbb{X}$.
Fix the functions $\phi_{-},\phi_{+}\in W^{s,p}$ such that $0<\phi_{-}\leqs\phi_{+}$, and define the shifted operators
\begin{align}
\label{HC-def-As}
A_L^s & :W^{s,p} \to \mathbb{X},&
A_L^s\phi
&:= A_L\phi +  \ba_s\phi,\\
\label{HC-def-fs}
F_{\tbw}^s & :[\phi_{-},\phi_{+}]_{s,p} \to \mathbb{X},&
F^s_{\tbw}(\phi)
&:= F_{\tbw}(\phi) - \ba_s\phi.
\end{align}
For $\phi\in[\phi_{-},\phi_{+}]_{s,p}$ and $a_{\tbw}\in W^{s-2,p}$, let
\begin{equation}
\label{T:HC-E-def-Tw}
T^s(\phi,a_{\tbw}) := -(A_{\tiL}^s)^{-1} F_{\tbw}^s(\phi).
\end{equation}
Then, the map $T^s : [\phi_{-},\phi_{+}]_{s,p}\times W^{s-2,p}\to W^{s,p}$ is continuous in both arguments.
Moreover, there exists $\tilde{s}\in(1+\frac3p-\frac{4}{\alpha},1+\frac3p)$ and constants $C_1,C_2$ such that
\begin{equation}\label{e:T-cpt}
\|T(\phi,a_{\tbw})\|_{s,p}\leqs C_1(1+ \|a_{\tbw}\|_{s-2,p})\|\phi\|_{\tilde{s},p}+ C_2,
\end{equation}
for all $\phi\in[\phi_{-},\phi_{+}]_{s,p}$ and $a_{\tbw}\in W^{s-2,p}$.
\end{lemma}

\begin{proof}
We first bound
\begin{align}\label{eq2:1aug13}
\| F_{\biw}^s(\phi)\|_{\mathbb{X}} &\leqs  \|a_{\tau}\phi^5-a_{\tbw}\phi^{-7}-a_{\rho}\phi^{-3}-(a+a_{\biw}\psi)\phi\|_{s-2,p}\\
&+  \left\|\left(\frac12\tau-\frac14\theta_-\right)(\gamma_I(\phi))^3 -\frac14S(\nu,\nu)(\gamma_I(\phi))^{-3}-a_I\gamma_I(\phi)\right\|_{s-1-\frac1p,p;\Sigma_I}\nonumber\\
&+\|g-a_E\gamma_E(\phi)\|_{s-1-\frac1p,p;\Sigma_E},\nonumber \\
& = \|f_{\biw}^s(\phi)\|_{s-2,p}+ \|h^s(\phi)\|_{s-1-\frac1p,p;\Sigma_I}+\|g^s\|_{s-1-\frac1p,p;\Sigma_E}.\nonumber
\end{align}

By applying Lemma~29 from~\cite{HNT07b} (recalled as Lemma A.6 in~\cite{HoTs10a}), for any $\tilde{s}\in(\frac3p,s]$, $s-2\in[-1,1]$ and $\frac1p\in(\frac{s-1}2\delta,1-\frac{3-s}2\delta)$ with $\delta=\frac1p-\frac{\tilde{s}-1}3$, we have
\begin{align}\label{eq3:1aug13}
\|f_{\biw}^s(\phi)\|_{s-2,p}
&\leqs  C \Big(\|a_{\tau}\|_{s-2,p}\,\|\phi_{+}^4\|_{\infty}
        + \|a_{\rho}\|_{s-2,p} \, \|\phi_{-}^{-4}\|_{\infty}\Big.
\\
&\quad \Big. + \|a_{\biw}\|_{s-2,p} \, (\|\phi_{-}^{-8}\|_{\infty}+\|\psi\|_{\tilde{s},p})
  + \|a\|_{s-2,p}\Big)\|\phi\|_{\tilde{s},p}.\nonumber
\end{align}

Let us verify that $\frac1p$ is indeed in the prescribed range.
First, given the assumptions on $\tilde{s}$ we have
${\delta=\frac13+\frac1p-\frac{\tilde{s}}3<\frac{4}{3\alpha}}$.
By subsequently taking into account $s\geqs1$,
we infer that
${1-\frac{3-s}2\delta\geqs1-\frac{4}{3\alpha}=\frac{3\alpha-4}{3\alpha}}$.
This shows $\frac1p<1-\frac{3-s}2\delta$ for $p>\frac{3\alpha}{3\alpha-4}$, which is not sharp, but will be sufficient for our analysis.
For the other bound, we need $\frac1p>\frac{s-1}2\delta$.  Given
that $\delta < \frac{4}{3\alpha}$, we have that $\frac{s-1}2\delta < \frac{2(s-1)}{3\alpha}$
and because $1 \leqs  s \leqs  3-\frac1{p'}$, we have $\frac{2(s-1)}{3\alpha} <\frac{4}{3\alpha}$.  The
assumption that $p< \frac{\alpha+1}{3}$ implies that $p < \frac{3\alpha}{4}$,
and therefore that $\frac{s-1}2\delta<\frac{4}{3\alpha}< \frac1p$. 
So $\frac1p$ is in the prescribed range.

Applying Lemma~29 from~\cite{HNT07b} again we have that 
\begin{align}\label{eq1:1aug13}
\|h(\phi)\|_{s-1-\frac1p,p;\Sigma_I} 
   &\leqs  C_1\left( \|\tau\|_{s-1-\frac1p,p;\Sigma_I}\|(\gamma_I(\phi))^2\|_{\infty;\Sigma_I} \right.\nonumber \\
   & \qquad + \|\theta_-\|_{s-1-\frac1p,p;\Sigma_I}\|(\gamma_I(\phi))^2\|_{\infty;\Sigma_I}  \nonumber\\
   & \qquad \left.+\|a_I\|_{s-1-\frac1p,p;\Sigma_I} \right) \|\gamma_I(\phi)\|_{\tilde{s}-\frac1p,p;\Sigma_I}  \\
   & \qquad +C_2\left( \|\tau\|_{s-1-\frac1p,p;\Sigma_I} + \|\theta_-\|_{s-1-\frac1p,p;\Sigma_I}\right)\|B\|^5_{\infty}\|B\|_{\tilde{s}-\frac1p,p;\Sigma_I}   \nonumber
   \end{align}
where we have used the fact that 
$S(\nu,\nu) = (2\gamma_I(\tau)+|\theta_-|/2)B^6$.

We again verify the conditions of Lemma~29 from~\cite{HNT07b}.
Given that ${s \in [1,3-\frac{1}{p'}]}$, we observe that 
${s-1-\frac1p \in [-1,1]}$.
We also require 
$\frac1p \in (\frac{s-\frac1p}{2}\delta, 1-\frac{2+\frac1p-s}{2}\delta)$. 
We observe
that $0<\delta = \frac1p-(\frac{\tilde{s}-\frac1p-1}{2}) = \frac12-(\frac{\tilde{s}-\frac3p}{2}) < \frac2{\alpha}$ since $\tilde{s} \in (1+\frac3p-\frac{4}{\alpha}, 1+\frac3p)$.  Since $s \geqs 1$, we have
that $1-\frac{2+\frac1p-s}{2}\delta \geqs 1 - \frac{1+\frac1p}{\alpha} = 1-\frac{1}{\alpha}-\frac{1}{\alpha p}$.  Requiring $\frac1p < 1-\frac{1}{\alpha}-\frac{1}{\alpha p}$
implies that $p > \frac{\alpha+1}{\alpha-1}$, which holds by assumption.  Finally, we observe that $\frac{s-\frac1p}{2}\delta < \frac{s-\frac1p}{\alpha} \leqs  \frac{3-\frac1p}{\alpha} < \frac1p$ if
$p < \frac{\alpha+1}{3}$, which holds by assumption.

We must now bound the terms in~\eqref{eq1:1aug13} involving the $L^{\infty}(\Sigma_I)$ norm of $\gamma_I(\phi)$ in terms of our sub-and-super solutions.  Let $u \in W^{s,p}$. 
Then there exists a sequence $\{u_m\} \subset C^{\infty}(\overline{\cM})$ such that $u_m \to u$ in $W^{s,p}$, and because $s > \frac3p$, $u_m \to u$ in $L^{\infty}$.
Moreover, by the continuity of $\{u_m\}$ we clearly have that
$$
\|\gamma_I(u_m)\|_{\infty;\Sigma_I}  =  \sup_{x\in \Sigma_I}|u_m(x)| \leqs  \|u_m\|_{\infty} \leqs  \|u\|_{\infty}+\e(m),
$$
where $\e(m) \to 0$ as $m \to \infty$.  We therefore have that $\|u\|_{\infty;\Sigma_I} \leqs  \|u\|_{\infty}$.
Using this fact, \eqref{eq1:1aug13} becomes
\begin{align}\label{eq4:1aug13}
\|h(\phi)\|_{s-1-\frac1p,p;\Sigma_I} 
   &\leqs  C_1\left( \|\tau\|_{s-1-\frac1p,p;\Sigma_I}\|\phi_+^2\|_{\infty} \right.\nonumber \\
   & \qquad + \|\theta_-\|_{s-1-\frac1p,p;\Sigma_I}\|\phi_+^2\|_{\infty}  \nonumber\\
   & \qquad \left.+\|a_I\|_{s-1-\frac1p,p;\Sigma_I} \right) \|\gamma_I(\phi)\|_{\tilde{s}-\frac1p,p;\Sigma_I}  \\
   & \qquad + C_2\left(\|\tau\|_{s-1-\frac1p,p;\Sigma_I} + \|\theta_-\|_{s-1-\frac1p,p;\Sigma_I}\right)\|B\|^5_{\infty}\|B\|_{\tilde{s}-\frac1p,p;\Sigma_I}. \nonumber
\end{align}

Similarly, we have that
\begin{align}\label{eq9:21aug13}
\|g^s\|_{s-1-\frac1p,p;\Sigma_E} \leqs  \|g\|_{s-1-\frac1p,p;\Sigma_I}+C\|a_E\|_{s-1-\frac1p,p;\Sigma_E}\|\gamma_E(\phi)\|_{\tilde{s}-\frac1p,p;\Sigma_E}.
\end{align}

Combining
Eqs~\eqref{eq2:1aug13}, \eqref{eq3:1aug13}, \eqref{eq4:1aug13},\eqref{eq9:21aug13} and utilizing the Trace Theorem to obtain the bound 
$$
\|\phi\|_{\tilde{s}-\frac1p,p;\Sigma_i} \leqs  \|\phi\|_{\tilde{s},p}, \quad \text{for} \quad i \in \{I,E\},
$$
we have that 
\begin{align}\label{eq5:1aug13}
\|F_{\tbw}^s(\phi)\|_{\mathbb{X}} \leqs  C_1(1+ \|a_{\tbw}\|_{s-2,p})\|\phi\|_{\tilde{s},p}+ C_2.
\end{align}

To finalize the proof of \eqref{e:T-cpt}, note that the operator $A_{L}^s$ is invertible by Lemma B.8 in~\cite{HoTs10a}, since the function $\ba_s$ is positive.
The inverse $(A_L^s)^{-1}:\mathbb{X} \to W^{s,p}$ is bounded by the Bounded Inverse Theorem; this gives \eqref{e:T-cpt}, with possibly different constants than in \eqref{eq5:1aug13}.

The continuity of the mapping $F_{\tbw}^s:[\phi_{-},\phi_{+}]_{s,p}\to \mathbb{X}$ for any $a_{\biw}\in W^{s-2,p}$ follows because
$F_{\tbw}^s$ is a composition of continuous maps, and the continuity of
$a_{\biw}\mapsto F^s_{\biw}(\phi)$ for fixed $\phi\in[\phi_{-},\phi_{+}]_{s,p}$ is obvious.
Being the composition of continuous maps, $(\phi,a_{\biw})\mapsto T_{\biw}^s(\phi)$ is also continuous.
\end{proof}


The following lemma shows that by choosing the shift sufficiently large, 
we can make the map $T^s$ monotone increasing.
This result is important for ensuring that the
Picard map $T$ for the Hamiltonian constraint is invariant
on the interval $[\phi_-,\phi_+]$ defined by sub- and super-solutions.
There is an obstruction that the scalar curvature and should be continuous,
which is a sufficient condition to guarantee that pointwise multiplication
in the space $W^{s,p}\otimes W^{s-2,p} \to W^{s-2,p}$ be continuous
operation.  This assumption can be handled in the general case by 
conformally transforming the metric to a metric with 
continuous scalar curvature and using the conformal covariance 
of the Hamiltonian constraint, cf. Section \ref{sec:proof1} and Lemma~\ref{l:conf-inv}.
(We omit explicitly writing the trace maps $\gamma_I$ and $\gamma_E$ for quantities evaluated on the boundaries $\Sigma_I$ and $\Sigma_E$ in the statement of Lemma~\ref{l:shift1} below without danger of confusion.)

\begin{lemma}
{\bf (Monotone increasing property of $T$)}
\label{l:shift1}
In addition to the conditions of Lemma~\ref{l:shift},
let $a_{\tiR}, H$ and $c$ be continuous and define the shift function $\ba_s = \ba + a_{\biw}\Psi$, where
$\ba$ and $a_{\biw}\Psi$ are as in Lemma~\ref{l:shift} and
\begin{align}\label{T:HC-E-def-alpha}
&a = \max\{1,a_R\}+5a_{\tau}\phi_+^4+3a_{\rho}\frac{\phi_+^2}{\phi_-^{6}}, \quad \psi = 7\frac{\phi_+^6}{\phi_-^{14}},  \\
&a_I = \max\{1,H\}+\frac32 |\tau|\phi_+^2+ \frac34|\theta_-|\phi_+^2+\frac{3B^6}{4}(2|\tau|+|\theta_-|/2)\frac{\phi_+^{2}}{\phi_-^6}, \quad \text{on $\Sigma_I$} \nonumber \\
&a_E = \max\{1,c\}, \quad \text{on $\Sigma_E$}. \nonumber
\end{align}
Then, for any fixed $a_{\tbw}\in W^{s-2,p}$, the map $\phi\mapsto T^s(\phi,a_{\tbw}) : [\phi_{-},\phi_{+}]_{s,p}\to W^{s,p}$ is monotone increasing.
\end{lemma}

\begin{proof}
By Lemma~B.7 in~\cite{HoTs10a} the shifted operator $A_L^s$ satisfies the maximum principle,
hence the inverse $(A_L^s)^{-1}:\mathbb{X} \to W^{s,p}$ is monotone increasing.

Now we will show that the operator $F_{\biw}^s$ is monotone
decreasing in $\phi$. Given any two functions $\phi_2$, $\phi_1\in
[\phi_{-},\phi_{+}]_{s,p}$ with $\phi_2\geqs\phi_1$, we have
\begin{align}
&F_{\biw}^s(\phi_2) - F_{\biw}^s(\phi_1)= \\
&\left( \begin{array}{c}a_{\tau} 
\bigl[ \phi_2^5 - \phi_1^5 \bigr]
- (a+a_{\biw}\psi) [\phi_2-\phi_1]
- a_{\rho} \bigl[ \phi_2^{-3} - \phi_1^{-3}\bigr]
- a_{\biw} \bigl[ \phi_2^{-7} - \phi_1^{-7}\bigr],\\
(\frac12\tau-\theta_-/4)[(\gamma_I(\phi_2))^3-(\gamma_I(\phi_1))^3]\\
-((\frac12\tau+\frac18|\theta_-|) B^6 )[(\gamma_I(\phi_2))^{-3}
-(\gamma_I(\phi_1))^{-3}] -a_I[\gamma_I(\phi_2)-\gamma_I(\phi_1)], \\
-a_E[\gamma_E(\phi_2)-\gamma_E(\phi_1)]\end{array}\right).  \nonumber
\end{align}
The inequalities \eqref{HC-phin}, the condition $0<\phi_1\leqs\phi_2$, and the choice of $\ba_s$ imply that
\[
F_{\biw}^s(\phi_2) - F_{\biw}^s(\phi_1) \leqs 0,
\]
which establishes that $F_{\biw}^s$ is monotone decreasing.

Both the operator $(A_L^s)^{-1}$ and the
map $-F_{\biw}^s$ are monotone increasing, therefore the composition map
$T^s(\cdot,a_{\biw}) = -(A_L^s)^{-1} F_{\biw}^s(\cdot)$ 
is also monotone increasing.
\end{proof}

\begin{lemma}
{\bf (Barriers for $T$ and the Hamiltonian constraint)}
\label{l:shiftsubsup}
Let the conditions of Lemma \ref{l:shift1} hold, with $\phi_{-}$ and $\phi_{+}$ sub- and super-solutions of the Hamiltonian constraint equation \eqref{HC-LYs1}, respectively.
Then, we have
$T^s(\phi_{+},a_{\tbw})\leqs\phi_{+}$ and $T^s(\phi_{-},a_{\tbw})\geqs\phi_{-}$.
\end{lemma}

\begin{proof}
We have
\begin{equation*}
\phi_{+}-T^s(\phi_{+},a_{\biw})
=(A_{\tiL}^s)^{-1}\bigl[A_{\tiL}^s(\phi_{+})+F_{\biw}^s(\phi_{+})\bigr],
\end{equation*}
which is nonnegative since $\phi_{+}$ is a super-solution and $(A_{\tiL}^s)^{-1}$ is linear and monotone increasing.
The proof of the other inequality is completely analogous.
\end{proof}

Since we are no longer using normal order cones, 
our non-empty, convex, closed interval $[\phi_-,\phi_+]_{s,p}$ 
is not necessarily bounded as a subset of $W^{s,p}$.
Therefore, we also need {\em a priori} bounds in the norm on $W^{s,p}$
to ensure the Picard iterates stay inside the intersection of the interval 
with the closed ball $\overline{B}_M$ in $W^{s,p}$ of radius $M$, centered at the origin.
We first establish a lemma to this effect that will be useful for 
both the non-CMC and CMC cases.

\begin{lemma}
{\bf (Invariance of $T$ on the ball $\overline{B}_M$)}
\label{T:HC-ball-gen}
Let the conditions of Lemma \ref{l:shift} hold, and let $a_{\tbw}\in W^{s-2,p}$.
Additionally assume that $s \in (1+\frac3p-\frac{4}{\alpha},\infty)$.  Then for any $\tilde{s}\in(1+\frac3p-\frac4{\alpha},\min\{s,1+\frac3p\})$ with $(\alpha >4)$, and for some $t\in(1+\frac3p-\frac{4}{\alpha},\tilde{s})$, there exists a closed ball $\overline{B}_M\subset W^{\tilde{s},p}$ 
of radius $M=\cO\left([1+\|a_{\tbw}\|_{s-2,p}]^{\tilde{s}/(\tilde{s}-t)}\right)$
such that
\begin{equation*}
\phi \in [\phi_-,\phi_+]_{\tilde{s},p}\cap \overline{B}_M
\quad\Rightarrow\quad
T^s(\phi,a_{\tbw})\in\overline{B}_M.
\end{equation*}
\end{lemma}

\begin{proof}
From Lemma \ref{l:shift}, there exist $t\in(1+\frac3p-\frac4{\alpha},\tilde{s})$ and $C_1, C_2 >0$ such that
$$
\|T^s(\phi,a_{\biw})\|_{\tilde{s},p}\leqs C_1(1+\|a_{\biw}\|_{s-2,p})\|\phi\|_{t,p}+C_2, \qquad
\forall\phi \in [\phi_-,\phi_+]_{\tilde{s},p}.
$$
For any $\varepsilon>0$, the norm $\|\phi\|_{t,p}$ can be bounded by the interpolation estimate
$$
\|\phi\|_{t,p}\leqs\varepsilon\|\phi\|_{\tilde{s},p}+C\varepsilon^{-t/(\tilde{s}-t)}\|\phi\|_p,
$$
where $C$ is a constant independent of $\varepsilon$.
Since $\phi$ is bounded from above by $\phi_+$, $\|\phi\|_p$ is bounded uniformly,
and now demanding that $\phi\in\overline{B}_M$, we get
\begin{equation}\label{e:Tbound}
\|T^s(\phi,a_{\biw})\|_{\tilde{s},p}\leqs C_1[1+\|a_{\biw}\|_{s-2,p}]\left(M\varepsilon+C\varepsilon^{-t/(\tilde{s}-t)}\right)+ C_2,
\end{equation}
with a possibly different constant $C_1$.
Choosing $\varepsilon$ such that $2\varepsilon C_1[1+\|a_{\biw}\|_{s-2,p}]=1$ and setting $M=2(CC_1[1+\|a_{\biw}\|_{s-2,p}]\varepsilon^{-t/(\tilde{s}-t)}+C_2)$, we can ensure that
the right hand side of \eqref{e:Tbound} is bounded by $M$.
\end{proof}

\section{Barriers for the Hamiltonian constraint}
\label{sec:barriers}

The results developed in~\S\ref{sec:hamiltonian} for a particular
fixed-point map $T$ for analyzing the Hamiltonian constraint equation 
and the coupled system rely on the existence
of generalized (weak) sub- and super-solutions, or barriers.
There, the Hamiltonian constraint with Robin boundary conditions was studied in isolation from the 
momentum constraint with Robin type boundary conditions, and these generalized barriers only needed to
satisfy the conditions given at the beginning of~\S\ref{sec:hamiltonian}
for a given fixed function $\biw$ appearing as a source term in the
nonlinearity of the Hamiltonian constraint.
Therefore, these types of barriers are sometimes referred to as 
{\em local barriers}, in that the coupling to the momentum constraint
is ignored. In order to establish existence results for the
coupled system in the non-CMC case, it will be critical that the
sub- and super-solutions satisfy one additional property that now
reflects the coupling, giving rise to the term {\em global barriers}. 
It will be useful now to define this global property precisely.

\begin{definition}
  \label{D:barriers}
A sub-solution $\phi_{-}$ is called {\bf global} iff it is a
sub-solution of \eqref{HC-LYs1} for all vector fields $\tbw_\phi$
solution of \eqref{MC-LYm1} with source function
$\phi\in[\phi_{-},\infty)\cap W^{s,p}$. A super-solution
$\phi_{+}$ is called {\bf global} iff it is a super-solution of
\eqref{HC-LYs1} for all vector fields $\tbw_\phi$ solution of
\eqref{MC-LYm1} with source function $\phi\in (0,\phi_{+}]\cap W^{s,p}$. A pair $\phi_-\leqs\phi_+$ of sub- and super-solutions is
called an {\bf admissible pair} if $\phi_-$ and $\phi_+$ are sub-
and super-solutions of \eqref{HC-LYs1} for all vector fields
$\tbw_\phi$ of \eqref{MC-LYm1} with source function $\phi\in
[\phi_{-},\phi_{+}]\cap W^{s,p}$.
\end{definition}

It is obvious that if $\phi_-$ and $\phi_+$ are respectively global
sub- and super-solutions, then the pair $\phi_-,\phi_+$ is
admissible in the sense above, provided they satisfy the compatibility
condition~\eqref{eqn:compat}.

Here our primary interests is in 
developing existence results for weak (and strong) non-CMC solutions to the 
coupled system with Robin boundary conditions which are
free of the near-CMC assumption. This assumption had appeared in two
distinct places in all prior literature in the case of closed manifolds~\cite{jIvM96,ACI08}; the first assumption appears in the
construction of a fixed-point argument based on strict
$k$-contractions, and the second assumption appears in the
construction of global super-solutions.  In the case of compact manifolds with
boundary, the only existence results to date require the mean curvature 
to be constant \cite{HoTs10a,SD04,DM05a}.
In this section, we construct global super-solutions for the coupled constraint equations with Robin boundary conditions
that are free of the near-CMC assumption, along with some compatible sub-solutions. These
sub- and super-solution constructions are needed for the general fixed-point result for
the coupled system, leading to our
main non-CMC results (Theorem~\ref{T:main1}).
 
Throughout this section, we will assume that the background metric $h$ belongs to $W^{s,p}$ with $p\in(1,\infty)$ and $s\in(\frac3p,\infty)\cap(1,2]$.
Recall that $r=\frac{3p}{3+(2-s)p}$, so that the
continuous embedding $L^{r}\hookrightarrow W^{s-2,p}$ holds.
Given a symmetric two-index tensor $\sigma\in L^{2r}$ and a vector
field $\biw\in\biW^{1,2r}$, introduce the functions
$a_{\sigma}=\frac18\sigma^2\in L^r$ and $a_{\cL\biw}=\frac18(\cL\biw)^2\in L^{r}$.
Note that under these conditions $a_{\biw}$ belongs to $L^r\hookrightarrow W^{s-2,p}$, and that if $a_{\sigma},a_{\cL\biw}\in L^\infty$, we have the pointwise estimate
\begin{align}\label{eq1:20aug13}
a_{\biw}^{\tiwedge}\leqs 2a_{\sigma}^{\tiwedge}+2a_{\cL\biw}^{\tiwedge}.
\end{align}
Here and in what follows, given any scalar function $u\in L^{\infty}$,
we use the notation
\[
u^{\tiwedge}:= \mbox{ess~sup}\, u,\qquad
u^{\tivee}:= \mbox{ess~inf}\, u.
\]
In the event that a given function $f$ is defined on some portion of the boundary (i.e. on $\Sigma_I$ or $\Sigma_E$), we
slightly abuse notation and let $f^{\tivee}$ and $f^{\tiwedge}$ denote 
$$
f^{\tivee} = \min_{\Sigma_i} (\mbox{ess~inf}|_{\Sigma_i}\ f), \quad \text{and} \quad f^{\tiwedge} = \max_{\Sigma_i} (\mbox{ess~sup}|_{\Sigma_i}\ f).
$$
In some places we will assume that when the vector field $\biw\in\biW^{1,2r}$ is given by the solution of the momentum constraint equation \eqref{MC-LYm1} with the source term $\phi\in W^{s,p}$,
\begin{equation}\label{CS-aLw-bound}
a_{\cL\biw}^{\tiwedge} \leqs \ttk(\phi_+):=
\ttk_1 \, \|\phi_+\|_{\infty}^{12} + \ttk_2,
\end{equation}
with some positive constants $\ttk_1$ and $\ttk_2$ and $\phi_+$ an { \em a priori}
upper bound on $\phi$. 
We can verify this assumption e.g. when the conditions of
Theorem~\ref{T:MC-E-Reg2} and Remark~\ref{rem1:18july13} are satisfied, since from Remark~\ref{rem1:18july13} we get 
\begin{align}\label{eq1:16oct13}
a_{\cL\biw}^{\tiwedge} =& \|\cL\biw\|_{\infty}^2 \\
 \leqs & C^2\left(\|\phi_+\|^6_{\infty}\|\bb_{\tau}\|_{z} +\|\bb_j\|_{e-2,q} +\|\tau B^6\|_{e-1-\frac{1}{q},q;\Sigma_I} \right. \nonumber \\
& \qquad + \left. \|\theta_- B^6\|_{e-1-\frac{1}{q},q;\Sigma_I}
+\|\sigma(\nu,\nu)\|_{e-1-\frac{1}{q},q;\Sigma_I}+\|\bX\|_{e-1-\frac{1}{q},q;\Sigma_I}\right)^2, \nonumber
\end{align}
given that $S(\nu,\nu) = (2\gamma_I(\tau)+|\theta_-|/2)B^6$.
If we apply Lemma 29 in \cite{HoTs10a} to the term 
$$
\|\tau B^6\|_{e-1-\frac{1}{q},q;\Sigma_I},
$$
appearing in \eqref{eq1:16oct13}, we obtain
\begin{align}
\|\tau B^6\|_{e-1-\frac{1}{q},q;\Sigma_I} \leqs C\left(\|\tau\|_{e-1-\frac{1}{q},q;\Sigma_I}\|B\|^5_{\infty}\|B\|_{\tilde{s},p}\right),
\end{align}  
where 
$
q \in (3, \frac{\alpha+1}{3}), ~\alpha > 8,~ e \in (1+\frac3q, 2],~ p \in (3,\infty) ~~\text{ and}~~ \tilde{s} \in (1+\frac3p - \frac{4}{\alpha}, 1+ \frac3p),
$
which is similar to conditions in Lemma~\ref{l:shift}.  A similar bound can obtained
for the term $\|\theta_- B^6\|_{e-1-\frac{1}{q},q;\Sigma_I},$  
and combining these estimates with \eqref{eq1:16oct13} we obtain bound \eqref{CS-aLw-bound} with the constants
\begin{align}\label{e:k1k2}
&\ttk_1 = 2 C^2( \|\bib_{\tau}\|_{z})^2, \\
&\ttk_2 = 2C^2\left(\|\bib_{j}\|_{e-2,q}+\|\sigma(\nu,\nu)\|_{e-1-\frac{1}{q},q;\Sigma_I}+\|\bX\|_{e-1-\frac{1}{q},q;\Sigma_I}\right. \nonumber\\ 
& \qquad \qquad \left.+ \|B\|^5_{\infty}\|B\|_{\tilde{s},p}\left(\|\tau\|_{e-1-\frac{1}{q},q;\Sigma_I}+\|\theta_-\|_{e-1-\frac{1}{q},q;\Sigma_I}\right)\right)^2. \nonumber
\end{align}

\subsection{Constant barriers}
   \label{sec:constant}
Now we will present some global sub- and super-solutions for the
Hamiltonian constraint equation \eqref{HC-LYs1} which are constant functions.
The proofs are based on the arguments in~\cite{HNT07b} for
the case of closed manifolds.  To simplify notation, we will omit
the trace operators $\Tr_{\tiI}$ and $\Tr_{\tiE}$ from the boundary
operators.

\begin{lemma}{\bf (Global super-solution)}
\label{L:HC-Sp}
Let $(\cM,h)$ be a 3-dimensional, smooth, compact
Riemannian manifold with metric $h \in W^{s,p}$, $s>\frac3p$ and
non-empty boundary satisfying the conditions \eqref{eq6:11july13}.
Suppose that the estimate \eqref{CS-aLw-bound} holds for the 
solution of the momentum constraint equation, 
and assume that $a_R$ is uniformly bounded from below, $a_\rho,a_\sigma\in L^\infty$, $c^{\tivee} >0$, 
$~(2\tau^{\tivee} +|\theta_-|^{\tivee}) > 0$ on $\Sigma_I$ and $g$ is uniformly bounded from above.
With the parameter $\varepsilon>0$ to be chosen later,
define the following rational polynomials
\begin{align}
&q_{1,\varepsilon}(\chi)=
(a_{\tau}^{\tivee}-\ttK_{1\varepsilon}) \, \chi^5
+ a_{\tiR}^{\tivee} \, \chi
- a_{\rho}^{\tiwedge}\,\chi^{-3}
- \ttK_{2\varepsilon} \chi^{-7},\\
&q_2(\chi) = \frac12H^{\tivee}\chi+\left(\frac12\tau^{\tivee} + \frac{|\theta_-|}{4}^{\tivee}\right)\chi^3-\left(\frac12\tau^{\tiwedge}+\frac18|\theta_-|^{\tiwedge} \right)B^6\chi^{-3}\\
&q_{3}(\chi) = c^{\tivee}\chi - g^{\tiwedge},
\end{align}
where $\ttK_{1\varepsilon}:=(1+\frac{1}\varepsilon)\ttk_1$ and $\ttK_{2\varepsilon}:=(1+\varepsilon)a_{\sigma}^{\tiwedge}+(1+\frac{1}{\varepsilon})\ttk_2$.

We distinguish the following two cases:

(a) In case $\ttk_1<a_{\tau}^{\tivee}$, choose $\displaystyle\varepsilon>\frac{\ttk_1}{a_{\tau}^{\tivee}-\ttk_1}$.
If $q_{1\varepsilon}$ has a root, let $\phi_1=\phi_1(a_{\tau}^{\tivee}-\ttK_{1\varepsilon},a_{\tiR}^{\tivee},a_{\rho}^{\tiwedge},\ttK_{2\varepsilon})$ 
be the largest positive root of $q_{1\varepsilon}$, and if $q_{1\varepsilon}$ has no positive roots, let $\phi_1=1$.
Similarly, let $\phi_2$ be the largest positive root of $q_2$ if its exists,
otherwise let $\phi_2 = 1$.  Now, the constant ${\phi_+=\max\{\phi_1, \phi_2, g^{\tiwedge}/c^{\tivee}\}}$ is a global super-solution of the Hamiltonian constraint equation (\ref{HC-LYs1}).

(b) In case $\ttk_1\geqs a_{\tau}^{\tivee}$, choose $\varepsilon>0$.  
In addition, assume that $a_{\tiR}^{\tivee}>0$ is sufficiently large and that both
$a_{\rho}^{\tiwedge}$ and $\ttK_{2\varepsilon}$ are sufficiently
small, so that $q_{1\varepsilon}$ has two positive roots, with the largest being as large as $\max\{\phi_3, g^{\tiwedge}/c^{\tivee}\}$, where
$\phi_3$ is the largest positive root of $q_3$. Then, the largest root
$\phi_+=\phi_2(a_{\tau}^{\tivee}-\ttK_{1\varepsilon},a_{\tiR}^{\tivee},a_{\rho}^{\tiwedge},\ttK_{2\varepsilon})$
of $q_{1\varepsilon}$ is a super-solution of the Hamiltonian constraint equation
(\ref{HC-LYs1}).
\end{lemma}

\begin{proof}
We look for a super-solution among the constant functions.  Let $\chi$ be any positive constant.
Then we have
\begin{align*}
A(\chi)+ F(\chi,\biw) = F(\chi,\biw) 
=\left( \begin{array}{c} a_{\tau} \chi^{5}+ a_{\tiR} \chi- a_{\rho}\chi^{-3}- a_{\biw} \chi^{-7} \\ \frac12H\chi+\left(\frac12\tau-\frac14\theta_-\right)\chi^3-\frac14S(\nu,\nu)\chi^{-3}\\ c\chi-g\end{array} \right) .
\end{align*}
In order for $\chi$ to be a super-solution of \eqref{HC-LYs1}, we require that $F(\chi,\biw) \geqs {\bf 0}$, which implies that each of the components
of $F(\chi,\biw)$ must be nonnegative.
Given any $\varepsilon>0$, the
inequality $2|\sigma_{ab} (\cL\biw)^{ab}|\leqs\varepsilon \sigma^2 +
\frac1\varepsilon(\cL\biw)^2$ implies that
\[\textstyle
8a_{\biw} = \sigma^2+(\cL \biw)^2 + 2 \sigma_{ab}(\cL\biw)^{ab}
\leqs (1+\varepsilon)\,\sigma^2
+ ( 1+\frac{1}{\varepsilon} ) \,(\cL \biw)^2,
\]
hence, taking into account \eqref{CS-aLw-bound}, for any
$\biw\in\biW^{1,2r}$ that is a solution of the momentum constraint
equation \eqref{MC-LYm1} with any source term $\phi\in(0,\chi]$, the
constant $a_{\biw}^{\tiwedge}$ must fulfill the inequality
\begin{equation}
\label{CS-aw-bound}\textstyle
a_{\biw}^{\tiwedge}
\leqs (1+\varepsilon)a_{\sigma}^{\tiwedge}+(1+\frac1\varepsilon)a_{\cL\biw}^{\tiwedge}
\leqs \ttK_{1\varepsilon}\|\phi_+\|_{\infty}^{12}+\ttK_{2\varepsilon}.
\end{equation}
Using the fact that $S(\nu,\nu) = (2\tau+|\theta_-|/2)B^6$ and letting $\phi_+ = \chi$, we have that for any constant $\chi>0$ and all $\phi\in(0,\chi]$, it holds that

\begin{align}
F(\chi,\biw_{\phi})
   &\geqs \left(\begin{array}{c} a_{\tau}^{\tivee} \chi^5+ a_{\tiR}^{\tivee} \chi- a_{\rho}^{\tiwedge} \chi^{-3}- \left( \ttK_{1\varepsilon} \, \|\phi_+\|_{\infty}^{12} + \ttK_{2\varepsilon}\right)\chi^{-7}\\ \frac12 H^{\tivee}\chi+\left(\frac12\tau -\frac14\theta_-\right)\chi^3-\frac14((2\tau+|\theta_-|/2)B^6)\chi^{-3} \\ c^{\tivee}\chi-g^{\tiwedge} \end{array}\right) 
  \nonumber \\
& \geqs  \left( \begin{array}{c} B_{\varepsilon}\chi^5+a_{R}^{\tivee}\chi-a_{\rho}^{\tiwedge}\chi^{-3}- \ttK_{2\varepsilon}\chi^{-7} \\  \frac12H^{\tivee}\chi+\left(\frac12\tau^{\tivee} + \frac{|\theta_-|}{4}^{\tivee}\right)\chi^3-\left(\frac12\tau^{\tiwedge}+\frac18|\theta_-|^{\tiwedge} \right)(B^{\tiwedge})^6\chi^{-3} \\ c^{\tivee}\chi -g^{\tiwedge}   \end{array} \right) 
 \nonumber \\ 
& = \left( \begin{array}{c} q_{1\varepsilon}(\chi) \\ q_2(\chi) \\ q_{3}(\chi) \end{array} \right),
\label{eq4:18july13}
\end{align}
where $B_{\varepsilon}:=a_{\tau}^{\tivee}-\ttK_{1\varepsilon}$.

Clearly we have
that $q_{2}(\chi) \geqs 0$ for $\chi$ sufficiently large if $2\tau^{\tivee} + |\theta_-|^{\tivee}>0$.  Similarly, $q_3(\chi) \geqs 0$ for
$\chi \geqs g^{\tiwedge}/c^{\tivee}$.  
We calculate the first and second derivative of $q_{1\varepsilon}$ as
\begin{equation}
\begin{split}
q_{\varepsilon}'(\chi)
&=
5B_{\varepsilon} \chi^4
+ a_{\tiR}^{\tivee}
+3 a_{\rho}^{\tiwedge} \chi^{-4}
+7 \ttK_{2\varepsilon} \chi^{-8},\\
q_{\varepsilon}''(\chi)
&=
20B_{\varepsilon} \chi^3
-12 a_{\rho}^{\tiwedge} \chi^{-5}
-56 \ttK_{2\varepsilon} \chi^{-9}.
\end{split}
\end{equation}

Consider the case (a). In this case, because of the choice
$\varepsilon>\frac{\ttk_1}{a_{\tau}^{\tivee}-\ttk_1}$,
we have $B_{\varepsilon}>0$, and so $q_{1\varepsilon}(\chi)>0$ for
sufficiently large $\chi$, and $q_{1\varepsilon}$ is increasing. The
function $q_{1\varepsilon}$ has no positive root only if
$a_{\rho}^{\tiwedge}=\ttK_{2\varepsilon}=0$. So if $q_{1\varepsilon}$
has no positive root, let $\phi_1 = 1$ and $q_{1\varepsilon}(\chi)\geqs0$ for all
$\chi\geqs0$. If $q_{1\varepsilon}$ has at least one positive root,
let $\phi_1$ be the largest positive root and $q_{1\varepsilon}(\chi)\geqs0$ for
all $\chi\geqs\phi_1$.  Similarly, let $\phi_2$ be the largest positive root of
$q_2$ if it exists, otherwise let $\phi_2 =1$.  Then $q_2(\chi) \geqs 0$ for
all $\chi \geqs \phi_2$ given that $2\tau^{\tivee} + |\theta_-|^{\tivee}>0 $.  
Recalling now that any constant $\chi$
satisfies $A(\chi)=0$, we conclude that 
\begin{align}
A(\chi) + F(\chi,\biw_\phi) \geqs 0 \quad
\quad \forall \,\chi\geqs\ \max\{\phi_1, \phi_2, g^{\tiwedge}/c^{\tivee}\}
\,,~~\forall \,\phi\in(0,\chi],
\end{align}
implying that $\phi_{+}=\max\{\phi_1, \phi_2, g^{\tiwedge}/c^{\tivee}\}$ is a global super-solution of the
Hamiltonian constraint (\ref{HC-LYs1}).

For the case (b), since $B_{\varepsilon}<0$ and
$a_{\rho}^{\tiwedge}$ and $\ttK_{2\varepsilon}$ are nonnegative, the
first derivative $q_{1\varepsilon}'(\chi)$ is strictly decreasing for
$\chi>0$, and since $q_{1\varepsilon}'(\chi)>0$ for sufficiently
small $\chi>0$ and $q_{\varepsilon}'(\chi)<0$ for sufficiently large
$\chi>0$, the derivative $q_{1\varepsilon}'$ has a unique positive
root, at which the polynomial $q_{1\varepsilon}$ attains its maximum
over $(0,\infty)$. This maximum is positive if both
$a_{\rho}^{\tiwedge}$ and $\ttK_{2\varepsilon}$ are sufficiently
small, and hence the polynomial $q_{1\varepsilon}$ has two positive
roots $\phi_{1\e}\leqs\phi_{2\e}$. Moreover, if $a_{R}^{\tivee}$ is sufficiently large, $\phi_{2\e} \geqs \max\{\phi_3,g^{\tiwedge}/c^{\tivee}\}$,
where $\phi_3$ is the largest positive root of $q_3$ if it exists, and $\phi_3= 1$ otherwise.
Similar to the above we conclude that
\[
A(\chi) + F(\chi,\biw_\phi) \geqs 0,
\qquad ~~\text{for}~~ \chi=\phi_{2\e}, \,\forall \,\phi\in(0,\chi],
\]
implying that $\phi_{+}= \phi_{2\e}$ is a global super-solution of
the Hamiltonian constraint (\ref{HC-LYs1}).
\end{proof}

\begin{remark}\label{rem1:8oct13}
In order for the condition $S(\nu,\nu) = (2\tau+|\theta_-|)B^6$ to imply the marginally trapped surface condition
\eqref{eq8:26jun13}, if suffices to construct a global super-solution $\phi_+$ and choose $B$ to be a constant
such that $B > \| \phi_+\|_{\infty}$.  In light
of \eqref{e:k1k2}, we observe that for both cases (a) and (b)
above we can choose $B$ large and $\|\tau\|_{e-1-\frac1q,q;\Sigma_I}$ and $\|\theta_-\|_{e-1-\frac1q,q;\Sigma_I}$ sufficiently
small so that the coefficient $K_{2\e}$ in \eqref{eq4:18july13} remains unchanged and $(2\tau^{\tiwedge}+|\theta_-|^{\tiwedge}/2)B^6$ decreases in size.
This ensures that we can choose $B > \|\phi_+\|_{\infty}$ in the above construction.
\end{remark}

Case (a) of the above lemma has the condition $\ttk_1<a_{\tau}^{\tivee}$, 
which is analogous to the near-CMC condition.
The above condition also requires 
that the extrinsic mean curvature $\tau$ is nowhere zero.
Noting that there are solutions even for $\tau\equiv0$ in some 
cases (cf.~\cite{jI95}), the condition $\inf\tau\ne 0$ appears 
as a rather strong restriction.
We see that case (b) of the above lemma removes this restriction, 
in exchange for the sign condition on $R$ and size conditions on $R$, $\rho$, $j$, and $\sigma$.

In the next Lemma we construct a global sub-solution using a pre-existing global super-solution.

\begin{lemma}{\bf(Global sub-solution)}
\label{L:HC-GSb}
Let $(\cM,h)$ be a 3-dimensional, smooth, compact
Riemannian manifold with metric $h \in W^{s,p}$, $s\geqs \frac3p$ and
non-empty boundary satisfying the conditions \eqref{eq6:11july13}.
Assume that $c$, $H$, $a_{\tiR}$ and $\tau$ are uniformly bounded from above, 
and $g$ is uniformly bounded from below.
Let $\phi_{+}>0$ be a global super-solution of the Hamiltonian
constraint and suppose that
$~(4\tau^{\tivee}+|\theta_-|^{\tivee}) > 0$ on $\Sigma_I$. 
Let $\phi_1$ be the unique positive root of the polynomial 
\begin{align}
q(\chi) = \frac12\max\{1,H^{\tiwedge}\}\chi+\left(\frac12\tau^{\tiwedge} + \frac{|\theta_-|}{4}^{\tiwedge}\right)\chi^3-\left(\frac12\tau^{\tivee}+\frac18|\theta_-|^{\tivee} \right)B^6\chi^{-3}
\end{align}
We distinguish between the following two cases:

\medskip

\noindent
(a) If $a_{\rho}^{\tivee} > 0$, let $\phi_2$ be the unique positive root of the polynomial 
\begin{align}
q_{\rho}(\chi) = a_{\tau}^{\tiwedge}\chi^8+\max\{1,a_R^{\tiwedge}\}\chi^4-a_{\rho}^{\tivee}.
\end{align}
Then 
$$
\phi_- = \min\{\phi_1,\phi_2,g^{\tivee}/H^{\tiwedge}\}
$$
is a global sub-solution.\\

\medskip

\noindent
(b) Let $a_{\sigma}^{\tivee}>\ttk(\phi_{+})$, where $\ttk$
is as in \eqref{CS-aLw-bound}. Then, for some
$\varepsilon\in(\ttk(\phi_{+})/a_{\sigma}^{\tivee},1)$, if $\phi_3$ is the unique
positive root of the polynomial
$$
q_\sigma(\chi)=
a_{\tau}^{\tiwedge} \,\chi^{12}
+ \max\{1,a_{\tiR}^{\tiwedge}\}\, \chi^8
- \ttK_{\varepsilon},
$$
where $\ttK_{\varepsilon}:=(1-\varepsilon)a_{\sigma}^{\tivee}-\bigl(\frac1\varepsilon-1\bigr)\ttk(\phi_{+})$,
then
$$
\phi_- =  \min\{\phi_1,\phi_3,g^{\tivee}/H^{\tiwedge}\}
$$
is a global sub-solution of \eqref{HC-LYs1}.
\end{lemma}

\begin{proof}
If $\chi$ is constant, then we have
\begin{align*}
A(\chi)+ F(\chi,\biw_{\phi}) = F(\chi,\biw_{\phi}) 
=\left( \begin{array}{c} a_{\tau} \chi^{5}+ a_{\tiR} \chi- a_{\rho}\chi^{-3}- a_{\biw} \chi^{-7} \\ \frac12H\chi+\frac12\tau\chi^3-\frac14S(\nu,\nu)\chi^{-3}-\frac14\theta_-\chi^3\\ c\chi-g\end{array} \right) .
\end{align*}
In order for $\chi$ to be a sub-solution of \eqref{HC-LYs1}, we require that $F(\chi,\biw_{\phi}) \leqs  {\bf 0}$, which implies that each of the components
of $F(\chi,\biw_{\phi})$ must be non-positive.

In case (a), we have that
\begin{align}
F(&\chi, \biw_{\phi}) \nonumber \\
&\leqs \left(\begin{array}{c} a_{\tau}^{\tiwedge} \chi^5+ a_{\tiR}^{\tiwedge} \chi- a_{\rho}^{\tivee} \chi^{-3}\\ \frac12 H^{\tiwedge}\chi+(\frac12\tau^{\tiwedge}+\frac14|\theta_-|^{\tiwedge})\chi^3-\frac14((2\tau^{\tivee}+|\theta_-|^{\tivee}/2)B^6)\chi^{-3} \\ c^{\tiwedge}\chi-g^{\tivee} \end{array} \right) 
   \label{eq3:24july13} \\
&\leqs \left(\begin{array}{c} a_{\tau}^{\tiwedge} \chi^5+ \max\{1,a_{\tiR}^{\tiwedge}\} \chi- a_{\rho}^{\tivee} \chi^{-3}\\ \frac12 \max\{1,H^{\tiwedge}\}\chi+(\frac12\tau^{\tiwedge}+\frac14|\theta_-|^{\tiwedge})\chi^3-\frac14((2\tau^{\tivee}+|\theta_-|^{\tivee}/2)(B^{\tivee})^6)\chi^{-3} \\ c^{\tiwedge}\chi-g^{\tivee} \end{array} \right). \nonumber
\end{align}
We observe that each component in~\eqref{eq3:24july13} will be non-positive provided that we have ${\phi_- = \chi = \min\{\phi_1,\phi_2, g^{\tivee}/c^{\tiwedge}\}}$,
where $\phi_- > 0$ given that $a_{\rho}>0$, $~(4\tau^{\tivee}+|\theta_-|^{\tivee}) > 0$ on $\Sigma_I$ and $g^{\tivee}/c^{\tiwedge} > 0$ by \eqref{eq1:12july13}.

In case (b), we observe that if $\chi>0$ is any constant function and $\biw\in\biW^{1,2r}$,
then we have
\begin{align}\label{eq5:24july13}
F(\chi,\biw_{\phi})&\leqs \left(\begin{array}{c}a_{\tau}^{\tiwedge} \chi^5+ C \chi- a_{\biw}^{\tivee} \chi^{-7} \\ 
\frac12 H^{\tiwedge}\chi+(\frac12\tau^{\tiwedge}+\frac14|\theta_-|^{\tiwedge})\chi^3-((\frac12\tau^{\tivee}+\frac18|\theta_-|^{\tivee})(B^{\tivee})^6)\chi^{-3} \\ c^{\tiwedge}\chi-g^{\tivee} \end{array} \right)  ,
\end{align}
where we have used that $a_{\rho}$ is nonnegative, and
introduced the constant $C=\max\{1,a_{\tiR}^{\tiwedge}\}$. 

Given any
$\varepsilon>0$, the inequality $2|\sigma_{ab}
(\cL\biw)^{ab}|\leqs\varepsilon \sigma^2 + \frac1\varepsilon(\cL\biw)^2$
implies that
\[\textstyle
8a_{\biw} = \sigma^2+(\cL\biw)^2 + 2 \sigma_{ab}(\cL\biw)^{ab}
\geqs (1-\varepsilon)\,\sigma^2
- ( \frac{1}{\varepsilon}-1 ) \,(\cL\biw)^2,
\]
hence, taking into account \eqref{CS-aLw-bound},
for any $\biw\in\biW^{1,2r}$ that is a solution of the momentum constraint equation \eqref{MC-LYm1} with any source term $\phi\in(0,\phi_{+}]$,
the constant $a_{\biw}^{\tivee}$ must fulfill the inequality
\begin{equation*}\textstyle
a_{\biw}^{\tivee}
\geqs (1-\varepsilon)a_{\sigma}^{\tivee}-(\frac1\varepsilon-1)a_{\cL\biw}^{\tiwedge}
\geqs (1-\varepsilon)a_{\sigma}^{\tivee}-(\frac1\varepsilon-1)\ttk(\phi_{+})=:\ttK_{\varepsilon}.
\end{equation*}
We use the above estimate in \eqref{eq5:24july13} to get, for any $\biw\in\biW^{1,2r}$ that
is a solution of the momentum constraint equation \eqref{MC-LYm1}
with any source term $\phi\in(0,\phi_{+}]$,
\begin{align}
F(\chi,\biw_{\phi})&\leqs \left(\begin{array}{c}a_{\tau}^{\tiwedge} \chi^5+ C \chi-\ttK_{\varepsilon} \chi^{-7}. \\ 
\frac12 H^{\tiwedge}\chi+(\frac12\tau^{\tiwedge}+\frac14|\theta_-|^{\tiwedge})\chi^3-((\frac12\tau^{\tivee}+\frac18|\theta_-|^{\tivee})(B^{\tivee})^6)\chi^{-3} \\ c^{\tiwedge}\chi-g^{\tivee} \end{array} \right)  .
\end{align}
Because of the choice
$\ttk(\phi_{+})/a_{\sigma}^{\tivee}<\varepsilon<1$, we have
$\ttK_{\varepsilon}>0$. So with the unique positive root $\phi_3$
of
\[
q_{\sigma}(\chi) :=
a_{\tau}^{\tiwedge} \,\chi^5
+ C\, \chi
- \ttK_{\varepsilon}\, \chi^{-7},
\]
we have $q_{\sigma}(\chi)\leqs0$ for any constant
$\chi\in(0,\phi_3]$.  Taking $\phi_- = \min\{\phi_1,\phi_3,g^{\tivee}/c^{\tiwedge}\}$, we have that
$F(\phi,\biw_{\phi}) \leqs  0$, which completes the proof.
\end{proof}

\subsection{Non-constant Barriers}\label{s:nonconstB}

The barriers constructed in the previous section require that the scalar curvature be strictly positive and sufficiently large
or that the near-CMC condition be satisfied.  Few restrictions are placed on the size of the data $\theta_-,\rho, \bj, \sigma(\nu,\nu)$ and $\bX$.  
In this section we develop non-constant global sub-and super-solutions using an auxiliary problem
similar to the one considered in Lemma~\ref{l:apriori} in Appendix~\ref{sec:app}.  The advantage of this construction
is that we only require the metric $h \in \cY^+$.  However, the tradeoff is that
we will require the data $|\theta_-|, \rho, \bj, \sigma(\nu,\nu)$ and $\bX$ to be sufficiently small.  Additionally,
we will either have to require that $b_{\tau}$ is sufficiently small or that $\delta > 0$ is sufficiently small,
where we recall that $g = \delta(c+ \cO(R^{-3}))$.  This assumption is analogous to the
smallness condition on the Dirichlet data $\phi_D$ in \cite{Dilt13a}. 

\begin{lemma}{\bf (Non-Constant Global super-solution for small $D\tau$)}
\label{L:HC-NCSp}
Let $(\cM,h)$ be a 3-dimensional, smooth, compact
Riemannian manifold with metric $h \in W^{s,p}$, $s>\frac3p$ in the positive Yamabe class and
non-empty boundary satisfying the conditions \eqref{eq6:11july13}.
Assume that the estimate \eqref{CS-aLw-bound} holds for the 
solution of the momentum constraint equation for two positive constants
$k_1$ and $k_2$, which can be chosen sufficiently small.
Additionally assume that $ a_\rho, a_\sigma\in L^\infty$, $a^{\tiwedge}_{\rho}, a_{\sigma}^{\tiwedge}$
are sufficiently small and that $~(2\tau^{\tivee} +|\theta_-|^{\tivee}) > 0$ on $\Sigma_I$.
Then if $u \in W^{s,p}$ satisfies 
\begin{align}\label{eq6:21aug13}
-\Delta u +a_Ru &= \Lambda_1>0,\\
\gamma_I\partial_{\nu}u +\frac12 H \gamma_Iu &= \Lambda_2>0, \nonumber \\
\gamma_E\partial_{\nu}u +c \gamma_Eu &= \Lambda_3>0, \nonumber
\end{align} 
for positive functions $ \Lambda_1,\Lambda_2,\Lambda_3$, then there exists
a constant $\beta > 0$ such that $\phi_+ = \beta u$ is a positive global super-solution of the 
Hamiltonian constraint equation (\ref{HC-LYs1}).
\end{lemma}
\begin{proof}
We first observe that by Lemmas~B.8 and B.7 in~\cite{HoTs10a}, the solution $u$ exists and is positive.
Evaluating~\eqref{HC-LYs1} at $\phi_+$ we have
\begin{align*}
A(\phi_+)+ F(\phi_+,\biw)  
=\left( \begin{array}{c} \beta\Lambda_1+a_{\tau} \phi_+^{5}- a_{\rho}\phi_+^{-3}- a_{\biw} \phi_+^{-7} \\ \beta\Lambda_2+\left(\frac12\tau -\frac14\theta_- \right)(\gamma_I\phi_+)^3-\frac14S(\nu,\nu)(\gamma_I\phi_+)^{-3}\\ \beta\Lambda_3-g\end{array} \right). 
\end{align*}
If we use the point-wise bound \eqref{eq1:20aug13} and the estimate \eqref{CS-aLw-bound}
with $k_1$ and $k_2$, we have that $a_{\biw} \leqs  K_1(\phi^{\tiwedge}_+)^{12}+K_2$, 
where $K_1 = 2k_1$ and $K_2 = 2a_{\sigma}^{\tiwedge}+2k_2$. Letting $k_3 = u^{\tiwedge}/u^{\tivee}$ and recalling that $\theta_- \leqs  0$, we then have 

\begin{align*}
A(\phi_+) & + F(\phi_+,\biw)  \\
&\\
\geqs & \left( \begin{array}{c} \beta\Lambda_1+a^{\tivee}_{\tau} \phi_+^{5}- a^{\tiwedge}_{\rho}\phi_+^{-3}- [K_1(\phi^{\tiwedge}_+)^{12}+K_2] \phi_+^{-7} \\ \beta\Lambda_2+\left(\frac12\tau -\frac14\theta_- \right)(\gamma_I\phi_+)^3-\left(\left(\frac12\tau +\frac18|\theta_-|\right)B^6\right)(\gamma_I\phi_+)^{-3}\\ \beta\Lambda_3-g^{\tiwedge}\end{array} \right)\\
\geqs & \left( \begin{array}{c} \beta\Lambda_1-\beta^5K_1k^{12}_3 u^{5}- \beta^{-3}a^{\tiwedge}_{\rho}u^{-3}- \beta^{-7}K_2 u^{-7} \\ \beta\Lambda_2+\beta^3\left(\frac12\tau -\frac14\theta_- \right)(\gamma_Iu)^3-\beta^{-3}\left(\frac12\tau +\frac18|\theta_-| \right)B^6(\gamma_Iu)^{-3}\\ \beta\Lambda_3-g^{\tiwedge}\end{array} \right)\\
\geqs & \left( \begin{array}{c} \beta\Lambda^{\tivee}_1-\beta^5K_1k^{12}_3 (u^{\tiwedge} )^{5}- \beta^{-3}a^{\tiwedge}_{\rho}(u^{\tivee})^{-3}- \beta^{-7}K_2 (u^{\tivee})^{-7} \\  \beta\Lambda_2^{\tivee}+\beta^3\left(\frac12\tau^{\tivee} +\frac14|\theta_-|^{\tivee} \right)(\gamma_Iu^{\tivee})^3-\beta^{-3}\left(\frac12\tau^{\tiwedge} +\frac18|\theta_-|^{\tiwedge} \right)(B^{\tiwedge})^6(\gamma_Iu^{\tivee})^{-3}\\ \beta\Lambda^{\tivee}_3-g^{\tiwedge}\end{array} \right)\\
\end{align*}

Therefore, $A(\phi_+)+ F(\phi_+,\biw) \geqs 0$ provided that we can choose $\beta, k_1, k_2, a_{\rho}$ and $a_{\sigma}$ 
so that 
\begin{align}
\beta\Lambda^{\tivee}_1-\beta^5K_1k^{12}_3 (u^{\tiwedge} )^{5}- \beta^{-3}a^{\tiwedge}_{\rho}(u^{\tivee})^{-3}- \beta^{-7}K_2 (u^{\tivee})^{-7} &\geqs 0,\label{eq2:21aug13}\\ 
\beta\Lambda_2^{\tivee}+\beta^3\left(\frac12\tau^{\tivee} +\frac14|\theta_-|^{\tivee} \right)(\gamma_Iu^{\tivee})^3
\qquad \qquad \qquad & \nonumber \\
-\beta^{-3}\left(\frac12\tau^{\tiwedge} +\frac18|\theta_-|^{\tiwedge} \right)(B^{\tiwedge})^6(\gamma_Iu^{\tivee})^{-3} &\geqs 0, \label{eq1:7oct13}\\
\beta\Lambda^{\tivee}_3-g^{\tiwedge} &\geqs 0. \label{eq2:7oct13}
\end{align}

We choose $\beta$ sufficiently large so that  
Eqs.~\eqref{eq1:7oct13}-\eqref{eq2:7oct13} are nonnegative.
Then choose $k_1$ so that 
\begin{align}\label{eq3:21aug13}
\beta\Lambda^{\tivee}_1-\beta^5K_1k^{12}_3 (u^{\tiwedge} )^{5} > 0.
\end{align}
Finally, choose $a_\rho, a_{\sigma}$ and $k_2$ sufficiently small so that
\begin{align}\label{eq2:8oct13}
\beta\Lambda^{\tivee}_1-\beta^5K_1k^{12}_3 (u^{\tiwedge} )^{5}- \beta^{-3}a^{\tiwedge}_{\rho}(u^{\tivee})^{-3}- \beta^{-7}K_2 (u^{\tivee})^{-7} \geqs 0.
\end{align}
For this choice of data, $\phi_+ = \beta u $ is a global super solution.
\end{proof}
 
\begin{remark}\label{rem2:8oct13}
As we mentioned in Remark~\ref{rem1:8oct13},
we require that $B > \| \phi_+\|_{\infty}$ in order for the condition $S(\nu,\nu) = (2\tau+|\theta_-|)B^6$ to imply that the marginally trapped surface condition
\eqref{eq8:26jun13}. We may again choose choose $B$ to be a large constant and require that $\|\tau\|_{e-1-\frac1q,q;\Sigma_I}$ and $\|\theta_-\|_{e-1-\frac1q,q;\Sigma_I}$ be sufficiently
small so that the coefficient $K_{2}$ in \eqref{eq2:8oct13} remains unchanged and $(2\tau^{\tiwedge}+|\theta_-|^{\tiwedge}/2)B^6$ decreases in size.
This ensures that we can choose $B > \|\phi_+\|_{\infty}$ in the above construction.
\end{remark}

 \begin{remark}\label{rem1:23sep13}
 The requirement that $k_1$ be sufficiently small places a restriction on the size of $b_{\tau} = \frac23D\tau$.
 While this is not ideal, the above result still allows for $\tau$ to have zeroes and not satisfy the inequality
 $$
 \frac{\|D\tau\|_{z}}{\tau^{\tivee}} \leqs  C< \infty,
 $$
 where $z \geqs 1$.  This is the near-CMC condition.  So the above barrier construction
 holds in the far-CMC setting, even though it places some restrictions on $D\tau$. 
 \end{remark}

Recalling that $g = \delta( c+ \cO(R^{-3}))$ for $\delta >0$, we show in the next Lemma that we may construct
global super-solutions if we replace the assumption that $D\tau$ be small with the assumption
that $\delta$ can be chosen arbitrarily small

\begin{lemma}{\bf (Non-Constant Global Super-Solution with small $\delta$)}\label{lem:16oct13}
Let the assumptions of Lemma~\ref{L:HC-NCSp} hold, with the exception that no smallness
assumptions are placed on $k_1$.  Then if $\delta>0$ can be chosen sufficiently small, there exists
a $\beta > 0$ such that if $B = \beta u$, then $\phi_+ = \beta u $ is a positive global super-solution of the 
Hamiltonian constraint equation (\ref{HC-LYs1}). 
\end{lemma}
\begin{proof}
If we set $B = \beta u$, and follow the same process as in the proof of Lemma~\ref{L:HC-NCSp}, we find
that the following three inequalities must be satisfied in order for $\beta u$ to be a global super-solution:
\begin{align}
\beta\Lambda^{\tivee}_1-\beta^5K_1k^{12}_3 (u^{\tiwedge} )^{5}- \beta^{-3}a^{\tiwedge}_{\rho}(u^{\tivee})^{-3}- \beta^{-7}K_2 (u^{\tivee})^{-7} &\geqs 0,\label{eq3:16oct13}\\ 
\beta\Lambda_2^{\tivee}+\beta^3\frac18|\theta_-|^{\tivee} (\gamma_Iu^{\tivee})^3 &\geqs 0, \label{eq4:16oct13}\\
\beta\Lambda^{\tivee}_3-g^{\tiwedge} &\geqs 0. \label{eq5:16oct13}
\end{align}
Eq.~\eqref{eq4:16oct13} is always true for positive $u$, and we may now choose $\beta > 0$, $a_{\sigma}$, $k_2$, $a_{\rho}$ sufficiently small
so that \eqref{eq3:16oct13} is true.  Finally, choosing $\delta$ sufficiently small will ensure that \eqref{eq5:16oct13} is satisfied.
\end{proof}

\begin{remark}\label{rem1:16oct13}
We note that the choice of $B$ in Lemma~\ref{lem:16oct13} ensures that enforcing the condition ${S(\nu,\nu) = (2\tau+|\theta_-|)B^6}$ will imply the marginally trapped surface condition \eqref{eq8:26jun13}.
\end{remark}

\begin{lemma}{\bf (Non-Constant Global Sub-solution)}
\label{L:HC-NCSub}
Let $(\cM,h)$ be a 3-dimensional, smooth, compact
Riemannian manifold with metric $h \in W^{s,p}$, $s>\frac3p$ and
non-empty boundary satisfying the conditions \eqref{eq6:11july13}.
Assume that $\tau \in L^{\infty}$,
$~(4\tau^{\tivee}+|\theta_-|^{\tivee}) > 0$ on $\Sigma_I$, $g^{\tivee} > 0$ and that $\phi_+$ is the global super-solution obtained
from Lemma~\ref{L:HC-NCSp} by solving~\eqref{eq6:21aug13}. 
We have the following two cases:

(a) If $a^{\tivee}_{\rho} > 0$, then there exists $\alpha > 0$ sufficiently small so that
$\phi_- = \alpha \phi_+ < \phi_+$ is a global sub-solution.

(b) Suppose that $a^{\tiwedge}_{\sigma} > \ttk(\phi_+)$, where $\ttk$ is as in \eqref{CS-aLw-bound}.  Then there
exists $\alpha > 0$ sufficiently small so that $\phi_- = \alpha \phi_+ < \phi_+$ is a global sub-solution.
\end{lemma}
\begin{proof}
Evaluating at $\phi_-$, where $\alpha$ is to be determined, we have that
\begin{align*}
A_L(\phi_-)+ F(\phi_-,\biw) = \left( \begin{array}{c} \alpha\beta\Lambda_1+a_{\tau} \phi_-^{5}- a_{\rho}\phi_-^{-3}- a_{\biw} \phi_-^{-7} \\ \alpha\beta\Lambda_2+\left(\frac12\tau -\frac14\theta_- \right)(\gamma_I\phi_-)^3-\frac14S(\nu,\nu)(\gamma_I\phi_-)^{-3}\\ \alpha\beta\Lambda_3-g\end{array} \right). 
\end{align*}
This implies that 
\begin{align}\label{eq7:21aug13}
&A_L(\phi_-) +F(\phi_-,\biw) \\
& \nonumber\\
\leqs & \left( \begin{array}{c} \alpha\beta\Lambda_1+\alpha^5a^{\tiwedge}_{\tau} \phi_+^{5}- \alpha^{-3}a^{\tivee}_{\rho}\phi_+^{-3}- \alpha^{-7}a^{\tivee}_{\biw} \phi_+^{-7} \\ \alpha\beta\Lambda_2+\alpha^3\left(\frac12\tau^{\tiwedge} +\frac14|\theta_-|^{\tiwedge} \right)(\gamma_I\phi_+)^3-\alpha^{-3}\left(\frac12\tau^{\tivee} +\frac18|\theta_-|^{\tivee} \right){(B^{\tivee})}^6 (\gamma_I\phi_+)^{-3}\\ \alpha\beta\Lambda_3-g^{\tivee}\end{array} \right). \nonumber
\end{align}

In case(a), because $\tau^{\tiwedge} < \infty, a_{\rho}^{\tivee} > 0, ~(4\tau^{\tivee} + |\theta_-|^{\tivee}) > 0$ and $g^{\tivee} > 0$, we may choose $\alpha$ sufficiently small so that each of the equations in \eqref{eq7:21aug13} is non-positive.  This
implies that ${A_L(\phi_-)+F(\phi_-,\biw) \leqs  0}$ and that $\phi_- = \alpha\phi_+$ is a global sub-solution.

\medskip

In case(b) we have that $a^{\tivee}_{\rho} = 0$, so we have that
\begin{align}\label{eq8:21aug13}
&A_L(\phi_-) +F(\phi_-,\biw) \\
& \nonumber\\
\leqs & \left( \begin{array}{c} \alpha\beta\Lambda_1+\alpha^5a^{\tiwedge}_{\tau} \phi_+^{5}- \alpha^{-7}a^{\tivee}_{\biw} \phi_+^{-7} \\ \alpha\beta\Lambda_2+\alpha^3\left(\frac12\tau^{\tiwedge} +\frac14|\theta_-|^{\tiwedge} \right)(\gamma_I\phi_+)^3-\alpha^{-3}\left(\frac12\tau^{\tivee} +\frac18|\theta_-|^{\tivee} \right)(B^{\tivee})^6(\gamma_I\phi_+)^{-3}\\ \alpha\beta\Lambda_3-g^{\tivee}\end{array} \right). \nonumber
\end{align}
The equations corresponding to the lower bounds for the Robin operators in \eqref{eq8:21aug13} remain unchanged.  So in order to guarantee that we can choose
$\alpha>0$ sufficiently small so that $\phi_-$ is a global sub-solution, we must show that $a^{\tivee}_{\biw}>0$
given the assumption that $a_{\sigma}^{\tivee} > k(\phi_+)$.  For $\ee >0$, the inequality $2|\sigma_{ab}(\cL\biw)^{ab}| \leqs  \ee\sigma^2+\frac{1}{\ee}(\cL\biw)^2$
implies that
\begin{align}
a^{\tivee}_{\biw}\geqs (1-\ee)a_{\sigma}^{\tivee}-\left(\frac{1}{\ee}-1\right)a^{\tiwedge}_{\cL\biw} \geqs (1-\ee)a^{\tivee}_{\sigma}-\left(\frac{1}{\ee}-1\right)\ttk(\phi_+) = \ttK_{\ee},
\end{align}
where $a_{\sigma}$ and $a_{\cL\biw}$ are defined in the paragraph preceding \eqref{eq1:20aug13} and $\ttk$ is the bound \eqref{CS-aLw-bound} 
on the momentum constraint.  Requiring that $\ee \in (0,1)$ and that $a^{\tivee}_{\cL\biw} > 0$, we find that $\ee \in (\ttk(\phi_+)/a^{\tiwedge}_{\sigma},1)$,
which is nonempty provided that $a_{\sigma}^{\tivee} > k(\phi_+)$.  Therefore, choosing $\ee \in (\ttk(\phi_+)/a^{\tiwedge}_{\sigma},1)$ we have that
$a^{\tivee}_{\cL\biw} > 0$, which allows us to choose $\alpha$ sufficiently small so that $A_L(\phi_-)+F(\phi_-,\biw) \leqs  0$, which
implies that $\phi_- = \alpha\phi_+$ is a global sub-solution.
\end{proof}

\begin{remark}\label{rem1:17oct13}
In Lemma~\ref{lem1:17oct13}, we assume that $\delta$ can be taken arbitrarily small and we set $B = \beta u $,
where $u$ solves \eqref{eq6:21aug13}.  The global sub-solution constructed in Lemma~\ref{L:HC-NCSub} does
not work for this choice of $B$.  We instead use the following Lemma to obtain a global sub-solution when
$\delta$ is small.  
\end{remark}

\begin{lemma}{\bf (Global Sub-Solution for small $\delta$)}\label{lem1:17oct13}
Let $(\cM,h)$ be a 3-dimensional, smooth, compact
Riemannian manifold with metric $h \in W^{s,p}$, $s>\frac3p$ and
non-empty boundary satisfying the conditions \eqref{eq6:11july13}.
Assume that $\tau \in L^{\infty}$,
$~(4\tau^{\tivee}+|\theta_-|^{\tivee}) > 0$ on $\Sigma_I$, $g^{\tivee} > 0$ and that $B = \phi_+ = \beta u $ is the global super-solution obtained
from Lemma~\ref{L:HC-NCSp} by solving~\eqref{eq6:21aug13}. 
Suppose $v \in W^{s,p}$ is a positive solution to
\begin{align}\label{eq1:17oct13}
-\Delta  v + a_Rv&= {\lambda}_1 > 0, \\
\partial_{\nu} v + \frac12 H v &= {\lambda}_2 > 0, \nonumber\\
\partial_{\nu} v + c v &= {\lambda}_3 > 0, \nonumber
\end{align}
where ${\lambda_1}, {\lambda_2}, {\lambda_3}$ are positive
functions chosen so that $v$ is distinct from $u$.
We have the following two cases:

(a) If $a^{\tivee}_{\rho} > 0$, then there exists $\alpha > 0$ sufficiently small so that
$\phi_- = \alpha v < \phi_+$ is a global sub-solution.

(b) Suppose that $a^{\tiwedge}_{\sigma} > \ttk(\phi_+)$, where $\ttk$ is as in \eqref{CS-aLw-bound}.  Then there
exists $\alpha > 0$ sufficiently small so that $\phi_- = \alpha v< \phi_+$ is a global sub-solution.
\end{lemma}
\begin{proof}
The existence of a positive $v$ solving~\eqref{eq1:17oct13}
follows from Theorem~2.1 in~\cite{HoTs10a}.
Evaluating at $\phi_-$, where $\alpha$ is to be determined, we have that
\begin{align*}
A_L(\phi_-)+ F(\phi_-,\biw) = \left( \begin{array}{c} \alpha\beta{\lambda}_1+a_{\tau} \phi_-^{5}- a_{\rho}\phi_-^{-3}- a_{\biw} \phi_-^{-7} \\ \alpha\beta{\lambda}_2+\left(\frac12\tau -\frac14\theta_- \right)(\gamma_I\phi_-)^3-\frac14S(\nu,\nu)(\gamma_I\phi_-)^{-3}\\ \alpha\beta{\lambda}_3-g\end{array} \right). 
\end{align*}
This implies that 
\begin{align}\label{eq3b:17oct13}
A_L(\phi_-) & +F(\phi_-,\biw) \\
& \nonumber\\
\leqs & \left( \begin{array}{c} \alpha\beta{\lambda}_1+\alpha^5a^{\tiwedge}_{\tau} v^{5}- \alpha^{-3}a^{\tivee}_{\rho}v^{-3}- \alpha^{-7}a^{\tivee}_{\biw} v^{-7} \\ \alpha\beta{\lambda}_2+\alpha^3\left(\frac12\tau^{\tiwedge} +\frac14|\theta_-|^{\tiwedge} \right)(\gamma_I v)^3-\alpha^{-3}\left(\frac12\tau^{\tivee} +\frac18|\theta_-|^{\tivee} \right)B^6 (\gamma_I v)^{-3}\\ \alpha\beta{\lambda}_3-g^{\tivee}\end{array} \right). \nonumber
\end{align}

In case(a), because $\tau^{\tiwedge} < \infty, a_{\rho}^{\tivee} > 0, ~(4\tau^{\tivee} + |\theta_-|^{\tivee}) > 0$, and $g^{\tivee} > 0$, we may choose $\alpha$ sufficiently small so that each of the equations in \eqref{eq7:21aug13} is non-positive.  This
implies that ${A_L(\phi_-)+F(\phi_-,\biw) \leqs  0}$ and that $\phi_- = \alpha\phi_+$ is a global sub-solution.  Moreover, because $u^{\tivee} > 0$, we can choose $\alpha >0$ so that $\phi_- = \alpha v < \phi_+ = \beta u$.

In case (b), the proof follows by making an argument similar to case (b) in the proof of Lemma~\ref{L:HC-NCSub}, with $\phi_- = \alpha v$.
\end{proof}

\begin{remark}\label{rem6:27sep13}
In practice, it will be impossible to construct the global sub-solution in Lemma~\ref{L:HC-NCSub}(b)
from using the global super-solution obtained in Lemma~\ref{L:HC-NCSp} given the smallness
assumptions on $a_{\sigma}$ and the dependence of $\ttk_2$ on $\sigma$. We
include the construction here for completeness as an alternative to the condition
that $\rho \ne 0$.
We note that in the closed case, it was shown by Maxwell in~\cite{dM09}
that under suitable smoothness assumptions on the metric (that it is
in atleast $W^{2,p}$), the known decay and other properties of the Green's 
function for the Laplace-Beltrami operator (cf.~\cite{Aubi82}) implied
that it was sufficient to construct only a global super-solution for 
completion of the Schauder argument in Theorem~\ref{T:FIXPT2}.
This allowed Maxwell in~\cite{dM09} to partially extend the far-from-CMC 
results in~\cite{HNT07b} for closed manifolds to the vacuum case ($\rho = 0$)
when the metric is in $W^{2,p}$ or better; the vacuum case for rough metrics 
on closed manifolds remains open.
In~\cite{Dilt13a}, Dilts followed closely Maxwell's argument in~\cite{dM09}
and showed that under the same smoothness assumptions on the background 
metric, and by also making smoothness assumptions on the boundary and
exploiting some additional results from~\cite{HoTs10a}, the standard
estimates for the Green's function from~\cite{Aubi82} can again be used
to exploit Maxwell's technique for avoiding the sub-solution 
on compact manifolds with boundary that have smooth metrics.
Additional assumptions can be placed on the other data to avoid
assuming that $\rho \ne 0$ to obtain a global sub-solution.  The following
Lemma, based on barrier constructions pioneered in \cite{dM05}
(see also~\cite{HNT07b} for a detailed discussion of related constructions
based on auxiliary problems),
provides a method to obtain a global sub-solution in vacuum with additional,
mild assumptions on $a_R$, $\sigma$ and $a_{\tau}$.  
\end{remark}

\begin{lemma}{\bf (Global Sub-solution with $\rho \equiv 0$) }\label{lem2:17oct13}
Let $(\cM,h)$ be a 3-dimensional, smooth, compact
Riemannian manifold with metric $h \in W^{s,p}$, $s>\frac3p$ and
non-empty boundary satisfying the conditions \eqref{eq6:11july13}.
Assume that $\tau \in L^{\infty}$,
$~(4\tau^{\tivee}+|\theta_-|^{\tivee}) > 0$ on $\Sigma_I$, $g^{\tivee} > 0$
and $|\sigma| > 0$.
Additionally assume that there exists $\gamma_1>0$ such that
$a_R+\gamma_1 a_{\tau} \geqs 0$ and $\gamma_2$ so that
$H+\gamma_2 (2\tau+|\theta_-|) \geqs 0$.  Let $v \in W^{s,p}$ be a positive solution to
\begin{align}\label{eq1b:17oct13}
-\Delta  v + (a_R+\gamma_1 a_{\tau})v&=a_{\biw} > 0 \\
\partial_{\nu} v + \frac12 (H+\gamma_2 (2\tau+|\theta_-|)) v &= {\eta}_2 > 0, \nonumber\\
\partial_{\nu} v + c v &= {\eta}_3 > 0, \nonumber
\end{align}
where $\eta_2, \eta_3$ are positive functions.  Then there exists $\alpha > 0$
such that $\phi_- = \alpha v \leqs \phi_+$ is a global sub-solution, where $\phi_+$
is any positive global super-solution.
\end{lemma}
\begin{proof}
The function $v$ exists and is positive by
Lemmas~B.7 and~B.8 in~\cite{HoTs10a}.
Evaluating at $\phi_-$, where $\alpha$ will be determined, we have that
\begin{align*}
A_L(\phi_-)+ F(\phi_-,\biw) = \left( \begin{array}{c} \alpha a_{\biw}-\alpha\gamma_1a_{\tau}v+a_{\tau} \phi_-^{5}- a_{\biw} \phi_-^{-7} \\  \alpha\eta_2-\alpha\gamma_2(2\tau+|\theta_-|)v+\left(\frac12\tau -\frac14\theta_- \right)\phi_-^3-\frac14S(\nu,\nu)\phi_-^{-3}\\ \alpha\eta_3-g\end{array} \right). 
\end{align*}
This implies that 
\begin{align}\label{eq3:17oct13}
A_L(\phi_-) & +F(\phi_-,\biw) \\
& \leqs \left( \begin{array}{c}
    (\alpha^5(v^{\tiwedge})^5-\alpha\gamma_1 v^{\tivee})a^{\tiwedge}_{\tau}+( \alpha  - \alpha^{-7}(v^{\tivee})^{-7})a^{\tivee}_{\biw} \\
     \alpha(\eta^{\tiwedge}_2-\gamma_2(2\tau^{\tivee}+|\theta_-|^{\tivee})v^{\tivee})+\alpha^3\left(\frac12\tau^{\tiwedge} +\frac14|\theta_-|^{\tiwedge} \right)(v^{\tiwedge})^3 -\Theta \\
     \alpha\eta^{\tiwedge}_3-g^{\tivee}\end{array} \right), \nonumber
\end{align}
where $\Theta = \alpha^{-3}\left(\frac12\tau^{\tivee} +\frac18|\theta_-|^{\tivee} \right)(B^{\tivee})^6 (v^{\tiwedge})^{-3}$.
By choosing $\alpha > 0$ sufficiently small, each of the components in \eqref{eq3:17oct13} can be made
non-positive.  Moreover, for any positive super-solution $\phi_+ > 0$, $\alpha $ can be chosen sufficiently small
so that $\phi_- = \alpha v < \phi_+$. 
\end{proof}

\subsection{Obstacles to Global Barriers for Arbitrary 
            $h \in \cY^+$ and $\tau$ }

In Section~\ref{s:nonconstB}, we showed that we can obtain global super solutions provided that certain data is sufficiently small and that 
either $h$ has sufficiently large scalar curvature and $\tau$ is arbitrary,
or that $h \in \cY^+$ and $D\tau$ is sufficiently small.  However, it has proven to be extremely difficult to construct global
super solutions where both $h$ and $\tau$ are freely specifiable.
We now give an analysis which helps to explain why this is the case.

As we saw in Theorem~\ref{L:HC-NCSp}, in order for 
$\phi_+ = \beta u$ to be a global super solution, $\beta $ must satisfy  
\begin{align*}
 \beta\Lambda^{\tivee}_1-\beta^5K_1k^{12}_3 (u^{\tiwedge} )^{5}> 0, \quad \text{and} \quad \beta\Lambda^{\tivee}_3-g^{\tiwedge} \geqs 0,
\end{align*}
where $k_3 = u^{\tiwedge}/u^{\tivee}$.  This implies that
\begin{align}\label{eq7:24sep13}
\frac{g^{\tiwedge}}{\Lambda^{\tivee}_3} \leqs  \beta < \frac{(\Lambda^{\tivee}_1)^{\frac14} (u^{\tivee})^3}{K_1^{\frac14}(u^{\tiwedge})^{\frac{17}{4}}}.
\end{align}

As we mentioned earlier,
it is impossible to choose $g^{\tiwedge}$ small without affecting the size of $\Lambda^{\tivee}_3$ (which
ultimately depends on $u$).  Therefore, if we cannot choose $D\tau$ to be small (which makes $K_1$ small), there is no guarantee that $\beta$ can be chosen to satisfy the
above conditions without knowing more about $u^{\tiwedge}, u^{\tivee}$.  

In an attempt to at least partially resolve this issue, we consider the following auxiliary problem: find $u$
that solves
\begin{align}\label{eq1:24sep13}
-\Delta u + a_R u &= f_1(u,\Lambda), \\
\Tr_{\tiI}\partial_{\nu} u+ \frac12 H\Tr_{\tiI}u &= f_2(\Tr_{\tiI}u,\Lambda),\nonumber\\
\Tr_{\tiE}\partial_{\nu} u+ c\Tr_{\tiE}u &= f_3(\Tr_{\tiE}u,\Lambda),\nonumber
\end{align}
where $\Lambda$ is a real valued parameter and $f_1, f_2$ and $f_3$
are positive functions.  
The idea is to choose the functions $f_1, f_2$ and $f_3$ so that we can
solve \eqref{eq1:24sep13} using the method of sub-and super-solutions. 
This will allow us to obtain a solution $u$ that is point-wise bounded by
the sub-and super-solutions, which will give us some control of 
$u^{\tivee}$ and $u^{\tiwedge}$.  The easiest approach to constructing
barriers for \eqref{eq1:24sep13} is to look for constant barriers.
Therefore we require that $f_1, f_2$ and $f_3$ satisfy one of the
two following conditions:

\begin{enumerate}

\item[{\bf 1)}]The functions
\begin{align}
&g_1(x) = a_R^{\tiwedge}x - f_1(x,\Lambda),\nonumber\\
&g_2(x) = \frac12H^{\tiwedge}x - f_2(x,\Lambda),\nonumber\\
&g_3(x) = c^{\tiwedge} x - f_3(x,\Lambda), \nonumber\\
&h_1(x) =a_R^{\tivee}x - f_1(x,\Lambda),\nonumber\\
&h_2(x) = \frac12H^{\tivee}x - f_2(x,\Lambda),\nonumber\\
&h_3(x) = c^{\tivee}x - f_3(x,\Lambda),\nonumber
\end{align}
all have at exactly one positive root. 

\medskip

\item[{\bf 2)}]\label{eq5:24sep13} All of the functions $g_i$ and $h_i$ have at least two roots, where
$\alpha_{i1}, \alpha_{i2}$ are the two smallest positive roots for each $g_i$ and 
$\gamma_{i1}$ and $\gamma_{i2}$ are the two smallest positive roots for each $h_i$.
Additionally assume that $\gamma_{j1} < \gamma_{i2}$
and $\alpha_{i1} < \alpha_{j2}$ for each $1 \leqs  i,j \leqs  3$.
\end{enumerate}

We observe that an unfortunate consequence of requiring
that one of the above two conditions be satisfied is that $R, H$ and $c$ must be positive,
given that $f_1, f_2$ and $f_3$ are strictly positive.  Therefore we
are not entirely free to specify $h$ in this construction.  While this limitation is
not ideal, the auxiliary problem \eqref{eq1:24sep13} is a natural starting point
in our attempts at constructing a global super-solution with a freely specifiable
$\tau$ and as few conditions on $h$ as possible.  

In case ${\bf 1}$, let $\alpha$ by the smallest root
of the $g_i$ and let $\gamma$ be the largest root of the $h_i$, and in case
${\bf 2}$ let $\alpha$ be the smallest root of the $g_i$ and $\gamma$ be the smallest
root of the $h_i$.  It is easily checked that $\alpha$ and $\gamma$ are sub- and super-solutions 
for~\eqref{eq1:24sep13}. Therefore, using the techniques outlined in this paper,
this problem can be solved to obtain $u$
which satisfies $\alpha \leqs  u^{\tivee}$ and $u^{\tiwedge} \leqs  \gamma$.  
Moreover, based on the definition of $\alpha$ and $\gamma$, $g_i(u) \geqs 0$ and $h_i(u) \geqs 0$ for each $1 \leqs  i \leqs  3$.

\begin{remark}
We note that in general, in order to solve~\eqref{eq1:24sep13} we only require that
the functions $f_1,f_2, f_3$ be chosen so that there exists an interval $I_1$ such that the
functions $h_i$ are nonnegative on this interval.  Similarly, we also require that
there exist and interval $I_2$ with $\sup I_2 < \inf I_1$ such that the $g_i$
are non-positive on $I_2$.  Then way may choose a super-solution $\beta \in I_1$ and
a sub-solution in $\alpha \in I_2$.  For this more general collection of $f_i$, it is unclear
whether $h_i(u) > 0$ or $g_i(u) >0$ for $\alpha \leqs  u \leqs  \gamma$.  This is not a necessary
condition, and the following discussion suggests that this is not ideal.  However, this
assumption allows for the following heuristic that illustrates the difficulties with constructing
barriers with minimal assumptions on $\tau$ and $h$.
\end{remark}

So in addition to the above conditions, if we can choose
$f_1, f_2,f_3, \Lambda$ and $\beta$ so that
$$
\frac{g^{\tiwedge}}{f_3(u,\Lambda)^{\tivee}} \leqs  \beta <   \frac{(f_1(u,\Lambda)^{\tivee})^{\frac14} (u^{\tivee})^3}{K_1^{\frac14}(u^{\tiwedge})^{\frac{17}{4}}},
$$
we will have a global super solution with freely specifiable $\tau$ and $h$ with positive scalar curvature.  Setting $\beta = \frac{g^{\tiwedge}}{f_3(u,\Lambda)^{\tivee}} $,
we find that we need to choose our functions and parameters to satisfy
\begin{align}\label{eq2:24sep13}
g^{\tiwedge} < \frac{f_3(u,\Lambda)^{\tivee}(f_1(u,\Lambda)^{\tivee})^{\frac14} (u^{\tivee})^3}{K_1^{\frac14}(u^{\tiwedge})^{\frac{17}{4}}}.
\end{align}

Implicitly $\alpha$ and $\gamma$ are functions of $\Lambda$. So the hope is that one can choose $f_1, f_2$ and $f_3$
and utilize the point-wise estimates $\alpha(\Lambda)$ and $\gamma(\Lambda)$ to determine 
if the above expression can be made sufficiently large by varying $\Lambda$ .  In particular, the uneven exponents on $u^{\tivee}$ and
$u^{\tiwedge}$ suggest that if one can choose $f_1, f_2, f_3$ so that $ \alpha(\Lambda) \to 0$, $\gamma(\Lambda) \to 0$ and $\alpha(\Lambda) \sim \gamma(\Lambda)$
as $\Lambda \to 0$, and
$$
\lim_{\Lambda \to 0} \frac{f_3(u,\Lambda)^{\tivee}(f_1(u,\Lambda))^{\tivee}}{\gamma(\Lambda)^{\frac54}} = \infty,
$$
then we can obtain our global super solution.
However, we observe that
$$
f_1(u,\Lambda)^{\tivee} \leqs  a_R^{\tivee}u^{\tivee}, \quad \text{and} \quad f_3(u,\Lambda)^{\tivee} \leqs  c^{\tivee}u^{\tivee},
$$
given that $g_i(u) \geqs 0$ and $h_i(u)\geqs 0$ for $\alpha \leqs  u \leqs  \gamma$.  Therefore,
\begin{align}\label{eq3:23sep13}
\frac{f_3(u,\Lambda)^{\tivee}(f_1(u,\Lambda)^{\tivee})^{\frac14} (u^{\tivee})^3}{K_1^{\frac14}(u^{\tiwedge})^{\frac{17}{4}}} \leqs  \frac{c^{\tivee}(a_R^{\tivee})^{\frac14} (u^{\tivee})^{\frac{17}{4}}}{K_1^{\frac14}(u^{\tiwedge})^{\frac{17}{4}}}.
\end{align}
Clearly $u^{\tivee}/u^{\tiwedge} \leqs  1$, and given that $c- g = \cO(R^{-3})$ it is highly likely that $c^{\tivee} \leqs  g^{\tiwedge}$. 
So without a largeness assumption on $R$ or a smallness assumption on $D\tau$, it will not always be the case that $g$ satisfies
$$
g^{\tiwedge} < \frac{c^{\tivee}(a_R^{\tivee})^{\frac14} (u^{\tivee})^{\frac{17}{4}}}{K_1^{\frac14}(u^{\tiwedge})^{\frac{17}{4}}},
$$
much less~\eqref{eq2:24sep13}.  

The attempted construction above shows that an auxiliary problem of the form~\eqref{eq1:24sep13}, with positive $f_1, f_2$ and $f_3$
satisfying {\bf 1} or {\bf 2}, will
not work in general if one hopes to obtain global super solutions with freely specifiable $\tau$ and $h$ with positive scalar
curvature.  Semilinear problems such a this are a natural place to start when attempting to construct barriers with minimal assumptions
on $\tau$ and $h$ given that {\em a priori} estimates and sub-and super-solutions are readily attained.  The above discussion suggests that one might require a
variational approach such as in Theorem~2.1 in~\cite{HoTs10a}.  However, the drawback of such an approach is that there are no standard techniques
for determining point-wise estimates of the solution $u$, which makes it difficult to verify the inequality \eqref{eq7:24sep13} without
additional assumptions on $\tau$ or $h$.

\section{Proof of the main results}
\label{sec:proof}

We now use the global barriers that we were able to construct above,
together with the results from Section~\ref{sec:momentum} 
and \ref{sec:hamiltonian},
to apply the coupled fixed point Theorem~\ref{T:FIXPT2} to prove 
Theorems~\ref{T:main2} and \ref{T:main1}.
We first prove Theorem~\ref{T:main2}.
The proof of Theorem~\ref{T:main1} involves only minor modifications of the
proof of Theorem~\ref{T:main2}.

\subsection{Proof of Theorem~\ref{T:main2}}
\label{sec:proof1}
Our strategy will be to prove the theorem first for the case $s\leqs2$, and then to bootstrap to include the higher regularity cases.

{\em Step 1: The choice of function spaces.}
We have the (reflexive) Banach spaces
$X=W^{s,p}$
and
$Y=\biW^{e,q}$,
where $p,q \in (3,\frac{\alpha+1}{3})$, $(\alpha > 8)$, $s=s(p) \in (1+\frac{3}{p},2]$,
and $e=e(p,s,q)\in(1,s]\cap(1+\frac3q,s-\frac3p+\frac3q]$.
We have the ordered Banach space
$Z=W^{\tilde{s},p}$
with the compact embedding $X=W^{s,p}\hookrightarrow W^{\tilde{s},p}=Z$,
for $\tilde{s} \in (1+\frac{3}{p}-\frac{4}{\alpha}, 1+\frac3p)$.
The interval $[\phi_{-},\phi_{+}]_{\tilde{s},p}$ is
nonempty (by compatibility of the barriers we will choose below),
and by Lemma~\ref{L:wsp-interval} on page \pageref{L:wsp-interval}
it is also convex with respect to the vector space structure of
$W^{\tilde{s},p}$ and closed with respect to the norm topology
of $W^{\tilde{s},p}$.
We then take $U=[\phi_-,\phi_+]_{\tilde{s},p} \cap \overline{B}_M$ 
for sufficiently
large $M$ (to be determined below), where
$\overline{B}_M$ is the closed ball in $Z=W^{\tilde{s},p}$
of radius $M$ about the origin, ensuring that $U$ is
non-empty, convex, closed, and bounded as a subset of $Z=W^{\tilde{s},p}$.

{\em Step 2: Construction of the mapping $S$.}
We have $\bib_j\in\biW^{e-2,q}$, and $\bib_\tau\in\biL^{z}$ with $z=\frac{3p}{3+(2-s)p}$.  The assumptions on $e$ imply that $\biL^{z}\hookrightarrow\biW^{e-2,q}$.
Similarly, $\gamma_{E}\tau \in W^{1-\frac1z,z}(\Sigma_I) \hookrightarrow W^{e-1-\frac1q,q}(\Sigma_I)$ and $\theta_- \in W^{s-1-\frac1p,p}(\Sigma_I) \hookrightarrow
W^{e-1-\frac1q,q}(\Sigma_I).$  Because $\gamma_I(\phi_+) \in W^{s-\frac1p,p}(\Sigma_I)$ and $${W^{s-\frac1p,p}(\Sigma_I) \otimes W^{e-1-\frac1q,q}(\Sigma_I) \to W^{e-1-\frac1q,q}(\Sigma_I)}$$ is point-wise bounded by Corollary~A.5(a) in~\cite{HoTs10a} or Corollary~3(a) in~\cite{HNT07b}, we have that $\bV^a\nu_a \in W^{e-1-\frac1q,q}(\Sigma_I)$.
Moreover, by Theorems~\ref{thm1:15july13} and \ref{T:MC-E-Reg2} the momentum constraint equation with boundary conditions~\eqref{eq4:8aug13} and \eqref{eq5:8aug13} is 
has a unique solution $\biw \in W^{e,q}$ 
for any ``source'' $\phi\in [\phi_-,\phi_+]_{\tilde{s},p}$.
The ranges for the exponents ensure that Lemma~\ref{T:MC-E-Lip1} holds, so that the momentum
constraint solution map 
$$S : [\phi_{-},\phi_{+}]_{\tilde{s},p} \to\biW^{e,q}=Y,$$
is continuous.

{\em Step 3: Construction of the mapping $T$.}
Define $r=\frac{3p}{3+(2-s)p}$, so that the
continuous embedding $L^{r}\hookrightarrow W^{s-2,p}$ holds.
Since the pointwise multiplication is bounded on $L^{2r}\otimes L^{2r}\to L^{r}$, and $\biw\in\biW^{e,q}\hookrightarrow\biW^{1,2r}$, we have $a_{\biw}\in W^{s-2,p}$ by $\sigma\in L^{2r}$.
The embeddings $W^{1,z}\hookrightarrow W^{e-1,q}\hookrightarrow L^{2r}$ also guarantee that $a_{\tau}=\frac1{12}\tau^2\in W^{s-2,p}$.
We have the scalar curvature $R\in W^{s-2,p}$, and these considerations show that the Hamiltonian constraint equation is well defined
with $[\phi_-,\phi_+]_{s,p}$ as the space of solutions. 
Similarly, $\gamma_I\tau \in W^{1-\frac1z,z}(\Sigma_I) \hookrightarrow W^{s-1-\frac1p,p}(\Sigma_I)$ and the fact that
$$W^{s-\frac1p,p}(\Sigma_I)\otimes W^{s-1-\frac1p,p}(\Sigma_I) \to W^{s-1-\frac1p,p}(\Sigma_I)$$ is a point-wise bounded map
imply that the Robin boundary conditions are well-defined provided $\phi \in [\phi_-,\phi_+]_{s,p}$.

Suppose for the moment that the scalar curvature $R$ of the background metric $h$ is continuous, 
and by using the map $T^s$ introduced in Lemma \ref{l:shift}, define the map $T$ by $T(\phi,\biw)=T^s(\phi,a_{\biw})$,
where $a_{\biw}$ is now considered as an expression depending on $\biw$.
Then Lemma~\ref{l:shift} implies that the map $T:[\phi_{-},\phi_{+}]_{\tilde{s},p}\times\biW^{e,q}\to W^{s,p}$ is continuous for any reasonable shift $a_s$,
which, by Lemma~\ref{l:shift1}, can be chosen so that $T$ is monotone in the first variable.
Combining the monotonicity with Lemma \ref{l:shiftsubsup}, we infer that the interval $[\phi_{-},\phi_{+}]_{\tilde{s},p}$
is invariant under $T(\cdot,a_{\biw})$ if $\biw\in S([\phi_{-},\phi_{+}]_{\tilde{s},p})$.
Since $\biL^z\hookrightarrow\biW^{e-2,q}$, from Theorem \ref{T:MC-E-Reg2} we have 
\begin{align*}
\|\cL\biw\|_{\infty} \leqs & C\left(\|\phi\|^6_{\infty}\|\bb_{\tau}\|_{z} +\|\bb_j\|_{e-2,q}\right.
+\|(2\tau+|\theta_-|/2)\|_{e-1-\frac{1}{q},q;\Sigma_I}\|\gamma_I(\phi_+)\|^{6}_{\infty}
\\
& \left. +\|\sigma(\nu,\nu)\|_{e-1-\frac{1}{q},p;\Sigma_I}+\|\bW\|_{e-1-\frac{1}{q},p;\Sigma_I}\right), \nonumber
\end{align*}
for any $\biw\in S([\phi_-,\phi_+]_{\tilde{s},p})$.
In view of Lemma~\ref{T:HC-ball-gen}, this shows that
there exists a closed ball $\overline{B}_M\subset W^{\tilde{s},p}$ such that
\begin{equation*}
\phi \in [\phi_-,\phi_+]_{\tilde{s},p}\cap \overline{B}_M,
\quad \biw \in S([\phi_-,\phi_+]_{\tilde{s},p}\cap \overline{B}_M)
\quad\Rightarrow\quad
T(\phi,\biw)\in\overline{B}_M.
\end{equation*}
Under the conditions in the above displayed formula, from the invariance of the interval $[\phi_{-},\phi_{+}]_{\tilde{s},p}$, we indeed have $T(\phi,\biw)\in U=[\phi_-,\phi_+]_{\tilde{s},p}\cap \overline{B}_M$.

However, the scalar curvature of $h$ may be not continuous, and in general it is not clear how to introduce a shift so that the resulting operator is monotone.
Nevertheless, we can conformally transform the metric into a metric with continuous, positive scalar curvature and positive boundary mean curvature by Theorem~2.2(c) in~\cite{HoTs10a}.
By using the conformal covariance of the 
Hamiltonian constraint (cf. Lemma~\ref{l:conf-inv}), we will be able to construct an appropriate mapping $T$.
Let $\tilde{h}=\psi^4h$ be a metric with continuous positive scalar curvature $\tilde{R}$ and boundary mean curvature $\tilde{H}$, 
where $\psi\in W^{s,p}$ is the (positive) conformal factor of the scaling satisfying $\Tr_{\tiE}\partial_{\nu}\psi = 0$.
Such a conformal factor exists by adapting the proof of Theorem~2.1 in~\cite{HoTs10a} to allow
for the specified boundary condition $\Tr_{\tiE}\partial_{\nu}\psi = 0$. 
Let $\tilde{T}^s$ be the mapping introduced in Lemma \ref{l:shift}, 
corresponding to the Hamiltonian constraint equation with the background metric $\tilde{h}$, 
coefficients $\tilde{a}_{\tau}=a_{\tau}$, $\tilde{a}_{\rho}=\psi^{-8}a_{\rho}$, and Robin boundary conditions
given by 
\begin{align}
\Tr_{\tiI}\partial_{\nu}\phi+\frac12\tilde{H}\Tr_{\tiI}\phi+\left(\frac12\tau - \frac14\theta_-  \right)\phi^3-\psi^{-6}S(\nu,\nu)(\Tr_{\tiI}\phi)^{-3} &= 0,\\
\Tr_{\tiE}\partial_{\nu}\phi + \psi^{-2}c\Tr_{\tiE}\phi - \psi^{-3}g &= 0.
\end{align}

With $\tilde{a}_{\biw}=\psi^{-12}a_{\biw}$, this {\em scaled} Hamiltonian constraint equation has sub- and super-solutions $\psi^{-1}\phi_{-}$ and $\psi^{-1}\phi_{+}$,
respectively, as long as $\phi_{-}$ and $\phi_{+}$ are sub- and super-solutions respectively of the original Hamiltonian constraint equation
(see~\cite{HoTs10a}).
We choose the shift in $\tilde{T}^s$ so that it is monotone in $[\psi^{-1}\phi_{-},\psi^{-1}\phi_{+}]_{\tilde{s},p}$. 
Then by the monotonicity and the above mentioned sub- and super-solution property under conformal scaling, for $\biw\in S([\phi_{-},\phi_{+}]_{\tilde{s},p})$, $\tilde{T}^s(\cdot,\psi^{-12}a_{\biw})$ is invariant on $[\psi^{-1}\phi_{-},\psi^{-1}\phi_{+}]_{\tilde{s},p}$.
Finally, we define
$$
T(\phi,\biw)=\psi\tilde{T}^s(\psi^{-1}\phi,\psi^{-12}a_{\biw}),
$$
where, as before, $a_{\biw}$ is considered as an expression depending on $\biw$.
From the pointwise multiplication properties of $\psi$ and $\psi^{-1}$,
the map $T:[\phi_{-},\phi_{+}]_{\tilde{s},p}\times\biW^{e,q}\to W^{s,p}$ is continuous,
and from the monotonicity and Lemma \ref{T:HC-ball-gen} , $T(\cdot,\biw)$ is invariant on $U=[\phi_-,\phi_+]_{\tilde{s},p}\cap \overline{B}_M$ for $\biw\in S(U)$,
where $M$ is taken to be sufficiently large.
Moreover, if the fixed point equation
$$
\phi=\psi\tilde{T}^s(\psi^{-1}\phi,\psi^{-12}a_{\biw}),
$$
is satisfied, then $\psi^{-1}\phi$ is a solution to the scaled Hamiltonian constraint equation with $\tilde{a}_{\biw}=\psi^{-12}a_{\biw}$,
and so by conformal covariance,
$\phi$ is a solution to the original Hamiltonian constraint equation 
(see~\cite{HoTs10a}).

{\em Step 4: Barrier choices and application of the fixed point theorem.}
At this point, Theorem \ref{T:FIXPT2} implies the Main Theorem \ref{T:main2},
provided that we have an admissible pair of barriers for the Hamiltonian constraint and
we can choose $B$ so that $B > \|\phi_+\|_{\infty}$ so that the
marginally trapped surface conditions \eqref{eq8:26jun13} are satisfied. 
The ranges for the exponents ensure through
Theorems~\ref{thm1:15july13} and \ref{T:MC-E-Reg2} that we can use the
estimate \eqref{CS-aLw-bound};
see the discussion following the estimate on page \pageref{CS-aLw-bound}.
In this case we use the global super-solution constructed in Lemma~\ref{L:HC-Sp}(a)
and the global sub-solution constructed in Lemma~\ref{L:HC-GSb}(a) or (b) or Lemma~\ref{lem2:17oct13}, depending
on whether $a_{\rho} \ne 0$ or $a^{\tivee}_{\sigma} > k(\phi_+)$.  Remark~\ref{rem1:8oct13}
implies that we can choose $B$ so that $B> \|\phi_+\|_{\infty}$.
This concludes the proof for the case $s\leqs2$.

{\em Step 5: Bootstrap.}
Now suppose that $s>2$.
First of all we need to show that the equations are well defined in the sense that the involved operators are bounded in appropriate spaces.
All other conditions being obviously satisfied, we will show that the Hamiltonian constraint is well-defined by showing
that $a_{\biw}\in W^{s-2,p}$ for any $\biw\in\biW^{e,q}$.
Since $\sigma$ and $\cL\biw$ belong to $W^{e-1,q}$, it suffices to show that the pointwise multiplication is bounded on $W^{e-1,q}\otimes W^{e-1,q}\to W^{s-2,p}$,
and by employing Corollary~A.5(b) in~\cite{HoTs10a}, we are done as long as $s-2\leqs e-1\geqs0$, $s-2-\frac3p<2(e-1-\frac3q)$, and $s-2-\frac3p\leqs e-1-\frac3q$.
After a rearrangement these conditions read as $e\geqs1$, $e\geqs s-1$, $e>\frac3q+\frac{d}2$, and $e\geqs\frac3q+d-1$, with the shorthand $d=s-\frac3p>1$, the latter inequality by the hypothesis of the theorem.
We have $d-1>\frac{d}2$ for $d>2$, and $1\geqs\frac{d}2$ for $d\leqs2$,
meaning that the condition $e>\frac3q+\frac{d}2$ is implied by the hypotheses $e\geqs\frac3q+d-1$ and $e>1+\frac3q$.
Similarly, given that $S(\nu,\nu) = (|\theta_-|/2)\Tr_{\tiI}(\phi_+)^6 \in W^{s-1-\frac1p,p}(\Sigma_I)$ and 
pointwise multiplication is bounded on $W^{s-1-\frac1p,p}(\Sigma_I)\otimes W^{s-1,p}(\Sigma_I) \to W^{s-1-\frac1p,p}(\Sigma_I)$,
the Robin boundary operators are well-defined.
So we conclude that the constraint equations with the specified Robin boundary conditions are well defined.

Next, we will treat the equations as equations defined with $s=e=2$ and with $p$ and $q$ appropriately chosen.
This is possible, since if the quadruple $(p,s,q,e)$ satisfies the hypotheses of the theorem,
then $(\tilde p,\tilde s=2,\tilde q,\tilde e=2)$ satisfies the hypotheses too, provided that 
$2-\frac3{\tilde p}\leqs s-\frac3p$, and 
$1<2-\frac3{\tilde q}\leqs e-\frac3q$.
Since the latter conditions reflect the Sobolev embeddings $W^{s,p}\hookrightarrow W^{2,\tilde p}$ and $W^{e,q}\hookrightarrow W^{2,\tilde q}\hookrightarrow W^{1,\infty}$,
the coefficients of the equations can also be shown to satisfy sufficient conditions for posing the problem for $(\tilde p,2,\tilde q,2)$.
Finally, we have $\tau\in W^{s-1,p}\hookrightarrow W^{1,\tilde{p}}=W^{1,z}$ since $z=\tilde{p}$ by $\tilde{s}=2$ for this new formulation.
Now, by the special case $s\leqs2$ of this theorem that is proven in the above steps, 
under the remaining hypotheses including the conditions on the metric and the near-CMC condition,
we have $\phi\in W^{2,\tilde{p}}$ with $\phi>0$ and $\biw\in\biW^{2,\tilde{q}}$ solution to the coupled system.

To complete the proof we only need to show that these solutions indeed satisfy $\phi\in W^{s,p}$ and $\biw\in\biW^{e,q}$.
Suppose that $\phi\in W^{s_1,p_1}$ and $\biw\in\biW^{e_1,q_1}$, with 
$1<s_1-\frac3{p_1}\leqs s-\frac3p$,
$1<e_1-\frac3{q_1}\leqs e-\frac3q$,
$\max\{2,s-2\}\leqs s_1\leqs s$, and
$\max\{2,e-2\}\leqs e_1\leqs\min\{e,s\}$.
Then we have $\bib_\tau\phi^6+\bib_j\in\biW^{e-2,q}$, 
and so Corollary~B.4 in~\cite{HoTs10a} implies that $\biw\in\biW^{e,q}$.
This implies that $a_{\biw}\in W^{s-2,p}$, and by employing Corollary~B.4 in~\cite{HoTs10a} once again, we get $\phi\in W^{s,p}$.
The proof is completed by induction.
\qed

\subsection{Proof of Theorem~\ref{T:main1}}

The proof is identical to the proof of Theorem~\ref{T:main2},
except for the particular barriers used.
In the proof of Theorem~\ref{T:main2}, the near-CMC condition
is used to construct global barriers satisfying
$$
0 < \phi_{-} \leqs \phi_{+} < \infty,
$$
for all three Yamabe classes, and then the supporting
results for the operators $S$ and $T$ established 
in~\S\ref{sec:momentum} and~\S\ref{sec:hamiltonian}
are used to reduce the proof to invoking the Coupled 
Fixed-Point Theorem~\ref{T:FIXPT2}.
The construction of $\phi_{+}$ is in fact the only place in the
proof of Theorem~\ref{T:main2} that requires the near-CMC condition.

{\em Cases (b) and (c).}
Here, the proof is identical to that of Theorem~\ref{T:main2},
except that the additional conditions 
made on the background metric $h_{ab}$ (that it be in $\cY^{+}(\cM)$),
and on the data 
(the smallness conditions on $|\theta_-|, D\tau, \sigma$, $\rho$, and $j$)
allow us to make use of the alternative construction 
of a global super-solution given 
in Lemma~\ref{L:HC-NCSp}, together with compatible global 
sub-solutions given in Lemma~\ref{L:HC-NCSub}(a) or Lemma~\ref{lem2:17oct13},
depending on whether $\rho \ne 0$.  Therefore we can apply Theorem~\ref{T:FIXPT2}
to solve the coupled conformal equations \eqref{eq4:11july13}-\eqref{eq2:8aug13}
with boundary conditions \eqref{eq5:11july13}-\eqref{eq1:12july13}, where
$S(\nu,\nu) = (2\tau+|\theta_-|)B^6$ and $B \in (W_+^{s,p}\backslash \{0\}) \cap L^{\infty}$ is
freely specified.
Furthermore, Remark~\ref{rem2:8oct13} implies that we may 
choose $B$ to be constant such that $B > \|\phi_+\|_{\infty}$ so that the marginally trapped surface 
conditions \eqref{eq8:26jun13} are satisfied.

{\em Case (a).}
Again, the proof is identical to that of Theorem~\ref{T:main2},
except that the additional conditions 
made on the background metric $h_{ab}$ (that it be in $\cY^{+}(\cM)$),
and on the data 
(the smallness conditions on $|\theta_-|, \delta, \sigma$, $\rho$, and $j$)
allow us to make use of the alternative construction 
of a global super-solution given 
in Lemma~\ref{lem:16oct13}, together with compatible global 
sub-solutions given in Lemma~\ref{lem1:17oct13} or Lemma~\ref{lem2:17oct13},
depending on whether $\rho \ne 0$.
Therefore we can apply Theorem~\ref{T:FIXPT2}
to solve the coupled conformal equations \eqref{eq4:11july13}-\eqref{eq2:8aug13}
with boundary conditions \eqref{eq5:11july13}-\eqref{eq1:12july13}, where
$S(\nu,\nu) = (2\tau+|\theta_-|)B^6$ and $B  = \beta u \in (W_+^{s,p}\backslash \{0\}) \cap L^{\infty}$ 
is obtained by solving \eqref{eq6:21aug13}.
Furthermore, Remark~\ref{rem1:16oct13} implies that this choice of $B$
ensures that the marginally trapped surface conditions
\eqref{eq8:26jun13} are satisfied.
Theorem \ref{T:main1} now follows.
\qed

\section*{Acknowledgments}
   \label{sec:ack}

The authors would like to thank D. Maxwell and J. Dilts for
helpful comments.
MH was supported in part by
NSF Awards~1065972, 1217175, and 1262982.
CM was supported by NSF Award~1065972.
GT was supported by an NSERC Canada Discovery Grant 
and by an FQRNT Quebec Nouveaux Chercheurs Grant.

\appendix
\section{Some key technical tools and some supporting results}
\label{sec:app}

The results in this article leverage and then build on the analysis 
framework and the supporting technical tools developed in our two
previous articles~\cite{HNT07b,HoTs10a}, including the material
contained in the appendices of both works.
We have made an effort to use completely consistent notation with 
these two prior works, and have also endeavored to avoid as much a 
possible any replication of the technical tools.
In particular, we have made use of a number of results 
from~\cite{HNT07b,HoTs10a} on: topological fixed-point theorems,
ordered Banach spaces, monotone increasing maps, Sobolev spaces on 
closed manifolds, elliptic operators and maximum principles,
Yamabe classification of non-smooth metrics,
and conformal covariance of the Hamiltonian constraint.
Although these technical tools represent the bulk of the
results we need in order to extablish the main results of
the paper, we will need the two additional sets of results below.

\bigskip
\noindent{\bf A priori estimates for the auxillary problem.} 
The first result we need are a priori $L^\infty$-estimates 
for solutions to a class of auxilliary problems.

\begin{lemma}\label{l:apriori}
Let the assumptions for Lemma~B.7 in~\cite{HoTs10a} hold, and
and let $f$ and $g$ playing the roles of $\alpha$ and $\beta$,
respectively in Lemma~B.7(a) in~\cite{HoTs10a}.
Then the solution $u$ to the
boundary value problem
\begin{align}\label{eq1:19aug13}
-\Delta u + fu &= \Lambda_1>0, \\
\Tr_{\tiN}\partial_{\nu}u+g\Tr_{\tiN}u&= \Lambda_2>0, \nonumber\\
\Tr_{\tiD}u &= \lambda >0, \nonumber
\end{align}
satisfies the following inequalities,
\begin{align}
&u^{\tiwedge} \le \beta < \infty \quad \text{if $~f^{\tivee} >0~$ and $~g^{\tivee} > 0~$}, \\
&u^{\tivee} \ge \alpha>0 \quad \text{if $~f^{\tiwedge} <\infty~$ and $~g^{\tiwedge} < \infty~$}, \nonumber
\end{align} 
where 
\begin{align}\label{eq2:19aug13}
\beta = \max\left\{ \frac{\Lambda^{\tiwedge}_1}{f^{\tivee}}, \frac{\Lambda^{\tiwedge}_2}{g^{\tivee}} , \lambda^{\tiwedge} \right\} \quad \text{and} \quad \alpha = \min\left\{   \frac{\Lambda_1^{\tivee}}{f^{\tiwedge}}, \frac{\Lambda_2^{\tivee}}{g^{\tiwedge}} , \lambda^{\tivee}  \right\}.
\end{align}
\end{lemma}
\begin{proof}
The fact that $u$ exists and is positive follows from 
Lemmas~B.8 and~B.7 in~\cite{HoTs10a}.
Define 
$$
H^1_{0,D} = \{ w \in W^{1,2}~:~ w = 0 ~~\text{on}~~\Sigma_D\}. 
$$
Then the functions $(u-\beta)^+$ and $(u-\alpha)^-$ are in $H^1_{0,D}$ given the definition of $\alpha$ and $\beta$.  
Define the sets $\cY^+ = \{x\in \overline{\cM}:u \ge \beta\}$, $\cY^- = \{x\in \overline{\cM}:u \le \alpha\}$. Let $dx$ be the
measure induced by the metric and $ds$ the corresponding boundary meaure, we have
\begin{align}\label{eq3:18aug13}
\|\nabla & (u-\beta)^+\|_{2} \\
&= \int_{\cM}\nabla (u-\beta)^+ \nabla (u-\beta)^+~dx \nonumber \\
&= \int_{\cM \cap \cY^+} \nabla u \nabla (u-\beta)^+~dx \nonumber \\
&= \int_{\cM \cap \cY^+}(\Lambda_1- f u  )(u-\beta)~dx + \int_{\Sigma_N\cap \cY^+} (\Lambda_2- g\gamma_N(u) ) \gamma_N(u-\beta)~ds \nonumber\\
&\leqs  \int_{\cM \cap \cY^+}(\Lambda^{\tiwedge}_1- f^{\tivee} u  )(u-\beta)~dx + \int_{\Sigma_N\cap \cY^+} (\Lambda^{\tiwedge}_2- g^{\tivee}\gamma_N(u) ) \gamma_N(u-\beta)~ds \nonumber \\
&\leqs 0, \nonumber
\end{align}
where the above quantity is non-positive by the definition of $\beta$.
Therefore we may conclude that $(u-\beta)^+$ is constant and that
either $u \le \beta~a.e$ or $u$ is a constant larger than $\beta$.  But this
is impossible given that $\gamma_D u = \lambda \le \beta$.  So $u \le \beta~a.e$.
We may use a similar argument involving $(u-\alpha)^-$ and the set $\cY^-$ to conclude that $u \ge \alpha~a.e.$
\end{proof}

\bigskip
\noindent{\bf Conformal invariance of the Hamiltonian constraint.}
The second result we need is a modification of Lemma~4.1 in~\cite{HoTs10a}, 
which concerns conformal invariance of the Hamiltonian constraint equation
on compact manifolds with certain types of boundary conditions.

Let $\cM$ be a smooth, compact, connected $n$-dimensional manifold with boundary $\Sigma = \Sigma_I \cup \Sigma_E$, $\Sigma_I\cap \Sigma_E = \varnothing$, equipped with a Riemannian metric $h_{ab}\in W^{s,p}_{\mathrm{loc}}$, where we assume throughout this section that $p\in(1,\infty)$, $s\in(\frac{n}p,\infty)\cap[1,\infty)$ and that $n\geqs3$.  Let $\Tr_{\tiI}$ and $\Tr_{\tiE}$ be the trace operators on
$\Sigma_I$ and $\Sigma_E$ respectively.
We consider the following model for the Hamiltonian constraint with Robin boundary conditions on $\Sigma_I$ an $\Sigma_E$:
\begin{equation*}\textstyle
F(\phi):=
\left(
\begin{array}{c}
-\Delta\phi+\frac{n-2}{4(n-1)}R\phi+a\phi^{t}\\
\Tr_{\tiI}\partial_{\nu}\phi+\frac{n-2}{2}H\Tr_{\tiI}\phi+b(\Tr_{\tiI}\phi)^e\\
\Tr_{\tiE}\partial_{\nu}\phi-c\Tr_{\tiE}\phi - f
\end{array}
\right)=0,
\end{equation*}
where $t,e\in\R$ are constants, $R\in W^{s-2,p}(\cM)$ and $H\in W^{s-1-\frac1p,p}(\Sigma_I)$ are respectively the scalar and mean curvatures of the metric $g$, and the other coefficients satisfy ${a\in W^{s-2,p}(\cM)}$, ${b\in W^{s-1-\frac1p,p}(\Sigma_I)}$, and ${c,f \in W^{s-\frac1p,p}(\Sigma_E)}$.
Setting $r=\frac4{n-2}$,
we will be interested in the transformation properties of $F$ under the conformal change $\tilde{h}_{ab}=\theta^rh_{ab}$ of the metric with the conformal factor $\theta\in W^{s,p}(\cM)$ satisfying $\theta>0$.
To this end, we consider
\begin{equation*}\textstyle
\tilde{F}(\psi):=
\left(
\begin{array}{c}
-\tilde\Delta\psi+\frac{n-2}{4(n-1)}\tilde{R}\psi+\tilde{a}\psi^{t}\\
\Tr_{\tiI}\partial_{\tilde\nu}\psi+\frac{n-2}{2}\tilde{H}\Tr_{\tiI}\psi+\tilde{b}(\Tr_{\tiI}\psi)^e\\
\Tr_{\tiE}\partial_{\nu}\phi-\tilde{c}\Tr_{\tiE}\phi - \tilde{f}
\end{array}
\right)=0,
\end{equation*}
where $\tilde{\Delta}$ is the Laplace-Beltrami operator associated to the metric $\tilde{g}$,
$\tilde{\nu}$ is the outer normal to $\Sigma$ with respect to $\tilde{h}$,
$\tilde{R}\in W^{s-2,p}(\cM)$ and $\tilde{H}\in W^{s-1-\frac1p,p}(\Sigma)$ are respectively the scalar and mean curvatures of $\tilde{h}$,
and $\tilde{a}\in W^{s-2,p}(\cM)$, $\tilde{b}\in W^{s-1-\frac1p,p}(\Sigma_I)$, and $\tilde{c},\tilde{f} \in W^{s-\frac1p,p}(\Sigma_E)$.
The following is a variation of Lemma~4.1 in~\cite{HoTs10a} which we need to incoporate the exterior boundary condition.

\begin{lemma}\label{l:conf-inv}
Let
$\tilde{a}=\theta^{t-r-1}a$, $\tilde{b}=\theta^{e-\frac{r}2-1}b$, and $\tilde{c}=\theta^{-\frac{r}{2}}c, \tilde{f} = \theta^{-\frac{r}{2}-1} f$.
Then if $\Tr_{\tiE}\partial_{\nu} \theta = 0$, we have
\begin{equation*}
\begin{split}
\tilde{F}(\psi)=0
\quad\Leftrightarrow\quad
F(\theta\psi)=0,\\
\tilde{F}(\psi)\geqs0
\quad\Leftrightarrow\quad
F(\theta\psi)\geqs0,\\
\tilde{F}(\psi)\leqs0
\quad\Leftrightarrow\quad
F(\theta\psi)\leqs0.
\end{split}
\end{equation*}
\end{lemma}

\begin{proof}
One can derive the following relations
\begin{equation*}
\begin{split}
\tilde{R}&\textstyle=\theta^{-r}R-\frac{4(n-1)}{n-2}\theta^{-r-1}\Delta\theta,\\
\tilde{\Delta}\psi&=\theta^{-r}\Delta\psi+2\theta^{-r-1}\nabla^a\theta\nabla_a\psi.\\
\end{split}
\end{equation*}
Combining these relations with
\begin{equation*}
\Delta(\theta\psi)=\theta\Delta\psi+\psi\Delta\theta+2\nabla^a\theta\nabla_a\psi,
\end{equation*}
we obtain
\begin{equation*}\textstyle
-\tilde{\Delta}\psi+\frac{n-2}{4(n-1)}\tilde{R}\psi
=\theta^{-r-1}
\left(-\Delta(\theta\psi)+\frac{n-2}{4(n-1)}R\theta\psi\right).
\end{equation*}
On the other hand, we have
\begin{equation*}
\begin{split}
\tilde{H}&\textstyle=\theta^{-\frac{r}2}H+\frac{2}{n-2}\theta^{-\frac{r}2-1}\partial_{\nu}\theta,\\
\partial_{\tilde\nu}\psi&=\theta^{-\frac{r}2}\partial_{\nu}\psi,
\end{split}
\end{equation*}
where traces are understood in the necessary places.
The above imply that
\begin{equation*}\textstyle
\partial_{\tilde\nu}\psi+\frac{n-2}{2}\tilde{H}\psi
=
\theta^{-\frac{r}2-1}
\left(\partial_{\nu}(\theta\psi)+\frac{n-2}{2}{H}\theta\psi\right),
\end{equation*}
and the proof follows.
\end{proof}

\bibliographystyle{abbrv}
\bibliography{Caleb,Caleb2,Caleb3,Caleb4,Caleb5,mjh,ref-gn,ref-gn2,books,papers}


\end{document}